\documentclass[final]{siamltex}
\usepackage{amsmath,amsfonts,amssymb}
\usepackage[usenames]{color}

\usepackage{bm}

\usepackage{braket}

\usepackage{xspace}

\usepackage{paralist}

\usepackage{graphicx}

\usepackage{epstopdf}

\usepackage{boxedminipage}

\usepackage{tikz}
\usetikzlibrary{arrows,positioning,backgrounds}

\usepackage{tabularx}

\usepackage{multirow}

\usepackage{rotating}

\usepackage{subfig}

\usepackage{algorithm}
\usepackage{algpseudocode}

\usepackage[font=small]{caption}

\usepackage[colorlinks,urlcolor=blue,citecolor=blue,linkcolor=blue]{hyperref}

\makeatletter
\renewcommand{\ALG@beginalgorithmic}{\small}
\makeatother

\algrenewcommand\algorithmicprocedure{\textbf{method}}
\algtext*{EndProcedure}%
\algrenewcommand\algorithmicforall{\textbf{for each}}

\newcommand{\Sec}[1]{\hyperref[sec:#1]{\S\ref*{sec:#1}}} %
\newcommand{\App}[1]{\hyperref[sec:#1]{Appendix~\ref*{sec:#1}}} %
\newcommand{\Eqn}[1]{\hyperref[eq:#1]{(\ref*{eq:#1})}} %
\newcommand{\Fig}[1]{\hyperref[fig:#1]{Fig.\,\ref*{fig:#1}}} %
\newcommand{\Tab}[1]{\hyperref[tab:#1]{Tab.\,\ref*{tab:#1}}} %
\newcommand{\Thm}[1]{\hyperref[thm:#1]{Thm.\,\ref*{thm:#1}}} %
\newcommand{\Lem}[1]{\hyperref[lem:#1]{Lem.\,\ref*{lem:#1}}} %
\newcommand{\Prop}[1]{\hyperref[prop:#1]{Prop.~\ref*{prop:#1}}} %
\newcommand{\Cor}[1]{\hyperref[cor:#1]{Cor.~\ref*{cor:#1}}} %
\newcommand{\Def}[1]{\hyperref[def:#1]{Defn.~\ref*{def:#1}}} %
\newcommand{\Alg}[1]{\hyperref[alg:#1]{Alg.\,\ref*{alg:#1}}} %
\newcommand{\Ex}[1]{\hyperref[ex:#1]{Ex.~\ref*{ex:#1}}} %
\newcommand{\Clm}[1]{\hyperref[clm:#1]{Claim~\ref*{clm:#1}}} %
\newcommand{\Step}[1]{\hyperref[step:#1]{Step~\ref*{step:#1}}} %

\newcommand\Exp[1]{\mathbb{E}[#1]}
\newcommand{\Prob}[1]{{\rm Prob}\left\{#1\right\}}
\newcommand{\Ceil}[1]{\lceil #1 \rceil}

\newcommand{\Floor}[1]{\lfloor #1 \rfloor}
\newcommand{\eps}{\varepsilon}
\newcommand{\qtext}[1]{\quad\text{#1}\quad}

\newcommand\Offset{225}
\newcommand\Slope{0.33}

\newcommand\MC[3]{\multicolumn{#1}{#2}{#3}}

\begin{document}

\title{Counting Triangles in Massive Graphs with MapReduce\thanks{This work was funded by the GRAPHS Program at DARPA and by the Applied Mathematics Program at the U.S.\@ Department of Energy. Sandia National Laboratories is a multi-program laboratory managed and operated by Sandia Corporation, a wholly owned subsidiary of Lockheed Martin Corporation, for the U.S. Department of Energy's National Nuclear Security Administration under contract DE-AC04-94AL85000.}}

\author{%
  Tamara G. Kolda\footnotemark[2] 
  \and Ali Pinar\footnotemark[2] 
  \and Todd Plantenga\footnotemark[2] 
  \and C.\@ Seshadhri\footnotemark[2] 
  \and Christine Task\footnotemark[3] 
}
\date{}
\maketitle

\renewcommand{\thefootnote}{\fnsymbol{footnote}}
\footnotetext[2]{Sandia National Laboratories, Livermore, CA.
  Email: \{tgkolda,apinar,tplante,scomand\}@sandia.gov}
\footnotetext[3]{Department of Computer Science,
  Purdue University, Lafayette, IN. Email: ctask@purdue.edu} 
\renewcommand{\thefootnote}{\arabic{footnote}}

\begin{abstract}
Graphs and networks are used to model interactions in a variety of
contexts. There is a growing need to quickly assess the
characteristics of a graph in order to understand its underlying
structure. Some of the most useful metrics are triangle-based and give
a measure of the connectedness of mutual friends.
This is often summarized in terms of clustering coefficients, which measure the 
likelihood that two neighbors of a node are themselves connected.
Computing these measures exactly for large-scale networks 
is prohibitively expensive in both memory and time.
However, a recent \emph{wedge sampling} algorithm has proved
successful in efficiently and accurately estimating clustering coefficients.
In this paper, we describe how to implement this approach in MapReduce
to deal with massive graphs.
We show results on publicly-available networks, the
largest of which is 132M nodes and 4.7B edges, as well as artificially
generated networks (using the Graph500 benchmark), the largest of
which has 240M nodes and 8.5B edges. 
We can estimate the clustering coefficient by degree bin
(e.g., we use exponential binning) and the number of triangles per bin, as
well as the global clustering coefficient and total number
of triangles, in an average of 
\Slope\@ seconds per million edges plus 
overhead (approximately \Offset\@ seconds total for our configuration).
The technique can also be used to study triangle
statistics such as the ratio of the highest and lowest degree, and we
highlight differences between social and non-social networks.
To the best of our knowledge, these are the largest 
triangle-based graph computations published to date.
\end{abstract}

\pagestyle{myheadings}
\thispagestyle{plain}
\markboth{\sc T.~G.~Kolda, A.~Pinar, T.~Plantenga, C.~Seshadhri, and C.~Task}%
{\sc Sampling Triangles in Massive Graphs with MapReduce}

\section{Introduction}
\label{sec:introduction}

Over the last decade, graphs have emerged as the standard
for modeling interactions between entities in a wide variety of
applications. Graphs are used to model infrastructure networks, the
world wide web, computer traffic, molecular interactions, ecological
systems, epidemics, co-authors, citations, and social interactions,
among others. Understanding the frequency of small subgraphs has
been an important aspect of graph analysis.

Despite the differences in the motivating applications,
some topological structures have emerged to be important across all
these domains.  The most important such subgraph
is the triangle (3-clique). 
Many networks, especially social networks, are known to have many triangles.
This is thought to be because social interactions exhibit
homophily (people befriend similar people) and transitivity (friends of
friends become friends).  The notion of \emph{clustering coefficient}
is inspired by this observation, and is the standard method
of summarizing triangle counts~\cite{WaSt98,NeStWa01}.  It is well known that some
networks, especially social networks, have much higher clustering
coefficients than random networks~\cite{Ne01,NePa03,OnCa04}. 
Triangle measures are important for understanding network
structure and evolution \cite{Hu10,SaCaWiZa10,FoFoGrPa11,UgKaBaMa11} and
reproducing the degree-wise clustering coefficients of a network is
important for generative models \cite{GuKr09,SaCaWiZa10,SeKoPi12}.

\subsection{Our Contributions}

For large graphs, computing triangle-based measures can be
expensive. The standard approach is to find all wedges, i.e., paths of
length 2, and check to see if they are closed, i.e., the edge that
complete the triangle exists.  
Previous work presents
a wedge-sampling approach for approximating clustering coefficients~\cite{ScWa05,SePiKo13};
this is in contrast to sampling single edges, which is a more obvious
but less reliable technique. In \cite{SePiKo13}, it is shown that the
wedge-sampling approach is orders of magnitude faster than enumeration
and is both faster and has less variance than edge-sampling techniques.

In this paper, we show that the wedge-sampling approach scales to massive networks using MapReduce, a framework well-suited to sampling. 
Previous distributed triangle counting algorithms have had to deal with  problems of
finding triangles where edges are stored on different processors and 
skewed degree distributions lead to load balancing issues~\cite{SuVa11,ArKhMa12}.
In contrast, our wedge sampling approach in MapReduce deals with these issues seamlessly
and recommends sampling as a general technique for large graphs since it
leads to fast serial algorithms as well as 
scalable parallel implementations.
We describe our contributions in more detail as follows:
\begin{asparaitem}
\item We present a parallelization of our wedge-based sampling
  algorithm in the MapReduce framework.
  The premise of wedge sampling is to set up a distribution on
  the vertices (as potential wedge centers) and use that to sample the
  actual wedges. Designing a serial algorithm is easier, since the information to compute the distribution
  and then form the wedges is all local. In the MapReduce
  implementation, edges are distributed arbitrarily; therefore, 
  it takes several passes to compute
  the necessary distribution, create the sample
  wedges, and finally check if they are closed.
\item 
  Additionally, we show that MapReduce enables computing multiple
  clustering coefficients for the same graph (e.g., binned by degree)
  for essentially the same cost as computing the single global
  clustering coefficient. Since the clustering coefficient generally
  varies with degree, it is helpful to see the profile versus a single
  value because these profiles are useful in graph characterization and modeling. 
\item We give extensive experimental results, on both real-world networks
  from the Laboratory for Web Algorithms as well as artificial
   networks created according to the Graph500 benchmark.
  We have multiple examples with more than a billion edges. To the
  best of our knowledge, these are the largest triangle computations
  to date. 
\item Results demonstrate the efficiency of our algorithm. For instance, we estimate
  the cost of computing clustering coefficients per (logarithmically)
  binned degree to be an average of \Slope\@ seconds per million
  edges plus overhead, which is approximately \Offset\@ seconds total for our
  32-node Hadoop cluster. 
  Hence, a graph with over 9B edges requires
  less than one hour of computation.
  Note that global clustering coefficient and
  total triangles are also computed.
\item A straightforward implementation requires that the entire edge
  list be ``shuffled'' three times. We show how to greatly reduce the
  shuffle volume with some clever implementation strategies that are
  able to filter the edge list during the ``map'' phase. We discuss
  the implementation details and show  comparisons of
  performance. 
\item A feature of wedge-based sampling is that the closed wedges are 
  uniform random triangles. Hence, we also give experimental results
  characterizing triangles in terms of the their
  minimum and maximum degrees. Triangles from social
  networks tend to be somewhat assortative whereas triangles from
  other types of networks are not.
\end{asparaitem}

\subsection{Related Work on Triangle Counting}
Enumeration algorithms for finding triangles are either node- or
edge-centric.   Node-centric algorithms iterate over all nodes and,
for each node $v$, check all pairs among the neighbors of $v$ for
being connected. Edge-centric algorithms, on the other hand, go over
all edges $(u,v)$ and seek common neighbors of $u$ and $v$.   
Chiba and Nishizeki~\cite{ChNi85} proposed a node-centric algorithm
that orders the vertices  by  degree and processes each edge only
once,  by its lower-degree vertex.  They showed that this algorithm
runs in $O(m\alpha(G))$-time, where $m$ is the number of edges, and
$\alpha(G)$ is the arboricity of the graph $G$ (arboricity is defined
as the minimum number of forests into which its edges can be
partitioned and can be considered as a measure of how dense the graph
is). Schank and Wagner~\cite{ScWa05}  used the same idea for their
{\em forward} algorithm.  
Chu and Cheng studied an I/O efficient implementation of the same algorithm~\cite{ChCh11}. 
Latapy proved that  the forward algorithm runs in $O(m^{3/2})$-time and proposed improvements  that reduce the search space \cite{latapy08}. Latapy also  showed that  the runtime of this algorithm becomes $O(mn^{1/\alpha})$ for  graphs with power-law degree distributions, where $\alpha$ is the power-law coefficient and $n$ is the number of vertices~\cite{latapy08}.
Arifuzzaman et al.~\cite{ArKhMa12} give a massively parallel algorithm for computing clustering coefficients.  Pearce et al.~\cite{PGA13} used triangle counting as an application to show the effectiveness  external memory  algorithms  for massive graph analysis.

Enumeration algorithms however, can be expensive, due to  the  extremely large number of triangles (see e.g., \Tab{networks}), 
even for graphs even of moderate size (millions of vertices). 
Much theoretical work has been done on characterizing the hardness of exhausting triangle enumeration and finding weighted triangles~\cite{Pa10,WiWi10}.
Eigenvalue/trace based methods adopted  by Tsourakakis~\cite{Ts08} and Avron~\cite{Av10}  compute
estimates of the total and per-degree number of triangles. However,
the compute-intensive nature of eigenvalue computations (even just a
few of the largest magnitude) makes these methods intractable
on large graphs.  

Most relevant to our work are sampling mechanisms.
Tsourakakis et al.~\cite{TsDrMi09} initiated the sparsification methods, the most important of which
is Doulion~\cite{TsKaMiFa09}.  This method sparsifies the graph by retaining each edge with probability $p$; counts  the triangles  in the sparsified graph; and multiplies this count by $p^{-3}$ to predict the  number of triangles in the original graph.  
One of the main benefits of Doulion is its ability to reduce large graphs to smaller ones that can be loaded into memory. 
However, the estimates can suffer from high variance~\cite{YoKi11}.
Theoretical analyses of this algorithm (and its variants) have been the subject of various studies~\cite{KoMiPeTs12,TsKoMi11,PaTs12}. 
Another sampling approach has been proposed by Kolountzakis et al.~\cite{KoMiPeTs12}, which involves both edge 
and triple-node sampling (a generalization of wedge-sampling).
A MapReduce implementation of their method could potentially use many of the same techniques presented here.
Alternative sampling mechanisms have been proposed for streaming and semi-streaming algorithms \cite{BaKuSi02, JoGh05, BeBoCaGi08,BuFrLeMaSo06}.
Most recently, Jha et al.~\cite{JSP13} showed how wedge-sampling can be performed when the graph is observed as a stream of edges and generalized their method for graphs with repeated edges~\cite{JSP13a}. An alternative approach and its parallelization were proposed by  Tangwongsan et al.~\cite{TaPaTi13}.
Many of these sampling procedures given above are by their very nature quite amenable to a MapReduce implementation.

The wedge-sampling approach used in this paper, first discussed by Schank and Wagner~\cite{ScWa05-2}, is a sampling approach with the high accuracy and speed advantages of other sampling-based methods (like
Doulion) but a hard bound on the variance. Previous work by a subset of the authors of this paper~\cite{SePiKo13} presents
a detailed empirical study of wedge sampling.
It was also shown that wedge sampling can compute a variety of triangle-based
metrics including degree-wise clustering coefficients and uniform randomly
sampled triangles. This distinguishes wedge sampling from previous sampling methods
that can only estimate the total number of triangles.

\subsection{Related Work on MapReduce for Graph Analytics}
MapReduce \cite{DeGh08} is a conceptual programming model for
processing massive data sets. The most popular implementation is the
open-source Apache Hadoop \cite{Hadoop} along with the Apache Hadoop
Distributed File System (HDFS) \cite{Hadoop}, which we have used in our
experiments. MapReduce assumes that the data is distributed across storage
in roughly equal-sized blocks.  The
MapReduce paradigm divides a parallel program into two parts: a
\emph{map} step and a \emph{reduce} step. During the map step, each block
of data is assigned to a \emph{mapper} which processes the data block to
emit key-value pairs. The mappers run in parallel and are ideally
local to the block of data being processed, minimizing communication
overhead. In between the map and reduce steps, a parallel \emph{shuffle}
takes place in order to group all values for each key together. 
This step is hidden from the user and is extremely efficient.
For every key, its values are grouped together and sent to a
\emph{reducer}, which processes the values for a single key and writes
the result to file. All keys are processed in parallel.

MapReduce has been used for network and graph analysis in a variety of
contexts. 
It is a natural choice, if for no other reason than the
fact that it is widely deployed~\cite{Li12}.
Pegasus \cite{KaTsFa09} is a general library for large-scale graph
processing; the largest graph they considered was 1.4M vertices and
6.6M edges and the PageRank analytic, but they did not report
execution times.
Lin and Schatz \cite{LiSc10} propose some special techniques for
graph algorithms such as PageRank that depend on matrix-vector products.
MapReduce sampling-based techniques that reduce the overall graph size
are discussed by Lattanzi et al.~\cite{LaMoSuVa11}.

In terms of triangle counting and computing clustering coefficients,
Cohen \cite{Co09} considers several different analytics
including triangle and rectangle enumeration. 
Plantenga \cite{Pl12} has studied subgraph isomorphism (i.e.,
finding small graph patterns such as triangles), including Cohen's
algorithm as a special case. 
(We use Plantenga's implementation of
Cohen's Triangle enumeration algorithm for comparison in our
subsequent numerical results.)
For a non-triangle pattern, Plantenga's SGI code ran on a 7.6B vertex
graph with 107B undirected edges in 620 minutes on a 64-node Hadoop cluster.
Wu et al.~\cite{WuDoKeCa11} have also studied triangle
enumeration using MapReduce with running times of roughly 175 seconds
on a graph with 1.6M nodes and 5.7M edges. 
Suri and Vassilvitskii \cite{SuVa11} proposed a MapReduce
implementation for exact per-node clustering coefficients. 
Most naive partitioning schemes do not give efficient parallelization because of high-degree vertices,
and their result involves new partitioning methods to avoid this
problem.
We discuss how both \cite{Pl12} and \cite{SuVa11} compare to our method in \Sec{timings}.
SAHAD \cite{ZhWaBuKh12} has a Hadoop program that uses sampling techniques
based on graph coloring to find subgraphs, but is limited to tree patterns.
Ugander et al.~\cite{UgKaBaMa11} analyzed the Facebook graph with 721M nodes and 
69B edges (representing friendships) on a 2,250 node Hadoop cluster.
They sampled 500,000 nodes and computed the exact local clustering coefficient for 
each sampled node.
They reported (binned) averages of the local clustering coefficients.
We note that the binned local clustering coefficient is different than the binned global clustering coefficient which we calculate in this paper. 
Additionally, the approach of computing the exact local clustering coefficient for a set of sample nodes is non-trivial for general graphs since assembling the neighbors for high-degree nodes is extremely expensive. In the case of the Facebook graph, the maximum number of neighbors is only 5,000 (per Facebook policy); compare to our graphs which have nodes with over 1 million neighbors.

\clearpage
\section{Background}

\subsection{Global Clustering Coefficient}

Let $G=(V,E)$ be an undirected graph
with $n=|V|$ nodes and $m=|E|$ undirected edges. We
assume the vertices are indexed by $i=1,\dots,n$. Let $d_i$ denote
the degree of vertex $i$; degree-zero vertices are ignored.
A \emph{wedge}  is a length-2 path.
Let $p_i$ denote the number of wedges centered at vertex $i$; i.e.,
$p_i = {d_i \choose 2} = \frac{d_i (d_i - 1)}{2}$.
A wedge is \emph{closed} if its endpoints are connected and \emph{open} otherwise. 
The \emph{center} of a wedge is the middle vertex.
A \emph{triangle} is a  cycle with three vertices. A closed wedge
forms a triangle; conversely, a triangle corresponds to \emph{three}
closed wedges. 
Let $t_i$ denote the number of triangles containing node $i$, which
is equal to the number of closed wedges centered at node $i$.
The \emph{node-level clustering coefficient} (first used in \cite{WaSt98}) is
\begin{equation*}
  c_i = \frac{t_i}{p_i} = \frac{\text{number of triangles incident to node
      $i$}}{\text{number of wedges centered at node $i$}}.
\end{equation*}
Thus, $c_i$ measures how tightly the neighbors of a
vertex are connected amongst themselves.  

We define $W$ to be the set of all wedges in $G$ and $p = |W| = \sum_i p_i$.
We partition $W$ into two disjoint subsets as follows:
\begin{align*}
  W_{0} &= \Set{w \in W | w \text{ open}}, \\
  W_{3} &= \Set{w \in W | w \text{ closed}} .
\end{align*}
The subscript of 3 for the closed wedges indicates that each triangle
creates 3 wedges in $W_3$.
Let $t = \frac{1}{3} \sum_i t_i = \frac{1}{3} |W_3|$ denote the total number of triangles
(since each triangle is counted thrice).
The \emph{(global) clustering coefficient}  (also known as the
transitivity) \cite{NeStWa01} of an undirected graph is 
given by
\begin{equation}
  \label{eq:CC}
  c 
  = \frac{|W_3|}{|W|} 
  = \frac{\sum t_i}{\sum p_i} 
  = \frac{3t}{p} 
  = \frac{3 \times \text{total number of triangles}}{\text{total number of wedges}}.
\end{equation}
At the global
level, $c$ is an indicator of how tightly nodes of
the graph are connected.  

\subsection{Binned Degree-wise Clustering Coefficient} 
In this paper, we will be using the binned degree-wise clustering
coefficients, which measure how tightly the neighborhood of vertices
of a specified degree group are connected.  
Let $D \subseteq \set{d_i,d_j,\dots}$ be a subset of degrees (recall that we ignore
degree-zero nodes). We define
$V_D = \set{i \in V | d_i \in D}$ and $n_D = |V_D|$.
In many cases, we are interested in a single degree, i.e.,
 if $D=\set{d}$ then 
$V_d$ is the set of nodes of degree $d$ and $n_d$ is the number of
nodes of degree $d$.

We define $W_D$ to be the set of all wedges centered at a node in $V_D$ and $p_D$ to be the
total number of wedges centered at nodes in $V_D$, i.e., $p_D =
|W_D|$. If $D = \Set{d}$, then $p_d = n_d {d \choose 2}$.
We partition the set $W_D$ into four disjoint
subsets as follows:
\begin{align*}
  W_{D,0} &= \Set{w \in W_D | w \text{ open}}, \\
  W_{D,q} &= \Set{w \in W_D | w \text{ closed and has $q$ nodes in $V_D$}} 
  \quad \text{for } q=1,2,3. 
\end{align*}
Define $p_{D,q} = |W_{D,q}|$ for $q=0,1,2,3$. Since $p_D = \sum_q p_{D,q}$,
 we can define
 \emph{binned degree-wise clustering coefficient}, $c_D$, as the fraction of
closed wedges in $W_D$; i.e., 
\begin{equation}
  \label{eq:cd}
  c_D = ({p_{D,1} + p_{D,2} + p_{D,3}})/{p}.
\end{equation}
The formula for triangles is more complex and given by
\begin{equation*}
  \label{eq:td}
  t_D = p_{D,1} + \frac{1}{2} \cdot p_{D,2} + \frac{1}{3} \cdot p_{D,3},
\end{equation*}
since for each triangle there is either one wedge in $W_{D,1}$, two
wedges in $W_{D,2}$ or three wedges in $W_{D,3}$.
\Fig{example} shows examples of these quantities when the bins are all
singletons: $\set{1}, \set{2}, \set{3}, \set{4}$.  

\begin{figure}[t!]
  \centering
  \begin{tikzpicture}[scale=0.5,
    nd/.style={circle,draw,fill=white,minimum size=4mm,inner sep=0pt}]
    \node (1) at (0,0) [nd] {1};
    \node (2) at (3,2) [nd] {2};
    \node (3) at (3,-2) [nd] {3};
    \node (4) at (6,1) [nd] {4};
    \node (5) at (8,-2) [nd] {5};
    \node (6) at (9,1.5) [nd] {6};
    \draw (1) to (2);
    \draw (1) to (3);
    \draw (2) to (4);
    \draw (3) to (4);
    \draw (3) to (5);
    \draw (4) to (5);
    \draw (4) to (6);
    \node [below right, text width=7cm, inner sep=0pt] at (10.25,2.25) 
    {\footnotesize %
      $n = 6$, $m = 7$, $\{d_i\} = \{2,2,3,4,2,1\}$\\[1mm]
      $p=12$,  $\{p_i\} = \{1,1,3,6,1,0\}$ \\[1mm]
      $t=1$, $\{t_i\} = \{0,0,1,1,1,0\}$ \\[1mm]
      $c=0.25$, $\{c_i\} = \{0,0,1/3,1/6,1,0\}$ \\[1mm]
      $\{n_d\} = \{1,3,1,1\}$,
      $\{p_d\} = \{0,3,3,6\}$ \\[1mm]
      $\{t_d\} = \{0,1,1,1\}$,
      $\{c_d\} = \{0,1/3,1/3,1/6\}$ %
    };
  \end{tikzpicture}
  \caption{Example graph with various quantities highlighted.}
  \label{fig:example}
\end{figure}
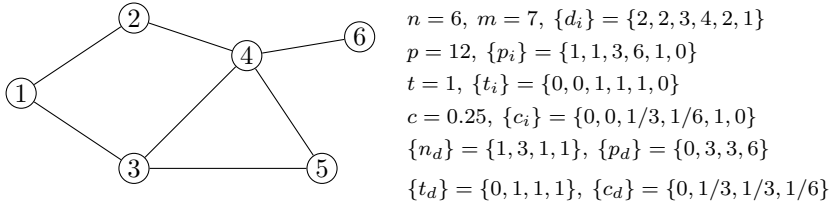

\section{Wedge Sampling for Triadic Measures}
\label{sec:sampling}
For a more detailed exposition of wedge sampling and empirical tests of its behavior,
we refer the reader to~\cite{SePiKo13}. 
For completeness, we review the relevant concepts and calculations here.

\subsection{Hoeffding's Inequality}

The following result is a simple corollary of Hoeffding's Inequality \cite{Ho63} (refer to Theorem 1.5 in \cite{DuPa}); 
the proof can be found in \cite{SePiKo13}. 
We say that $\eps$ is the \emph{error} and $(1-\delta)$ is the
\emph{confidence}. 

\begin{theorem} 
  \label{thm:Hoeffding}   
  Let $X_1, X_2, \ldots, X_k$ be independent random variables with $0
  \leq X_i \leq 1$ for all $i=1,\dots,k$.  
  Define $\bar X = \frac{1}{k} \sum_{i=1}^k X_i$. 
  Let $\mu = \Exp{\bar X}$. 
  For any positive $\eps, \delta$, setting $k \geq \Ceil{0.5 \eps^{-2}\ln(2/\delta)}$ yields
  \begin{displaymath}
    \Prob{|\bar X - \mu| \geq \eps } \leq \delta.    
  \end{displaymath}
\end{theorem}

\subsection{Binned Degree-wise Clustering Coefficients and Triangles}

The strategy for computing the clustering coefficient per degree (or
degree range) is similar to that described for the degree-wise clustering coefficient in \cite{SePiKo13}. 

\begin{theorem}[Binned Degree-wise Clustering Coefficient]
  \label{thm:cd-est}
  For $\eps, \delta > 0$, set 
  $k \geq \Ceil{0.5\, \eps^{-2} \ln(2/\delta)}$.
  For $i=1,\dots,k$, choose wedge $w_i$ uniformly at random (with
  replacement) from $W_D$
  and let $X_i$ be defined as 
    \begin{displaymath}
      \label{eq:Xi}
      X_i =
      \begin{cases}
        1, & \text{if $w_i$ is closed}, \\
        0, & \text{otherwise}.
      \end{cases}
    \end{displaymath}
  Then \\[-5ex]
  \begin{equation*}
    \Prob{ |\hat c_D - c_D| \geq \eps } \leq \delta
    \qtext{for}
    \hat c_D = \frac{1}{k} \sum_{i=1}^k X_i.
  \end{equation*}
\end{theorem}
\begin{proof}
  Observe that $c_D = \Exp{\bar X}$ since it is the
  probability that a random wedge in $W_D$ is closed. The proof
  follows immediately from \Thm{Hoeffding}.
\end{proof}

\paragraph{Choosing Uniform Random Wedges}
We do not want to form all wedges explicitly. Instead, we
\emph{implicitly} generate random wedges. 
Observe that the number of wedges centered at vertex $i$ is 
exactly ${d_i \choose 2}$, and $p = \sum_i {d_i \choose 2}$.
That leads to the following procedure.
To select a random wedge, recall that $p_D = |W_D|$. Therefore, first choose vertex $i \in
V_D$ with probability ${d_i \choose 2} / p_D$. 
Second, choose two distinct neighbors of vertex $i$ to form a
random wedge. 
To set up this distribution, we need to compute the degree
distribution. 
If $D = \Set{d}$ (a
singleton), then all nodes in $V_d$ are equally probable.
If $D = \set{0,+\infty}$, then the weight of vertex $i$ is ${d_i \choose 2} / p$.

Estimating the number of triangles is slightly more complicated
since each closed wedge may have 1,2, or 3 vertices in $V_D$. %

\begin{theorem}[Degree-wise Triangle Count \cite{SePiKo13}]
  \label{thm:td-est}
  Let the conditions of \Thm{cd-est} hold.
  For each $w_i$, 
  let $Y_i$ be defined as
  \begin{equation*}
    Y_i =
    \begin{cases}
      1, & \text{if } w \in W_{D,1}, \\
      \frac{1}{2}, & \text{if } w \in W_{D,2}, \\
      \frac{1}{3}, & \text{if } w \in W_{D,3}, \\
      0, & \text{if } w \in W_{D,0} \text{ (open) }.\\
    \end{cases}
  \end{equation*}
  Then
  \begin{displaymath}
    \Prob{ | \hat t_D - t_D | \geq \eps \cdot p_D } \leq \delta
    \qtext{for}
    \hat t = p_D \cdot \frac{1}{k} \sum_{i=1}^k Y_i.
  \end{displaymath}
\end{theorem}
\begin{proof}
  We claim $\Exp{Y} = t_D$.
  Suppose that $w$ is selected from $W_D$ uniformly at random.
  Observe that
  \begin{align*}
    \Exp{Y} 
    &
    = \Prob{w \in W_{D,1}} 
    + \frac{\Prob{w \in W_{D,2}}}{2} 
    + \frac{\Prob{w \in W_{D,3}}}{3} \\
    &
    = 1 \cdot \frac{p_{D,1}}{p_D}
    + \frac{1}{2} \cdot \frac{p_{D,2}}{p_D}
    + \frac{1}{3} \cdot \frac{p_{D,3}}{p_D} \\
    &
    = t_D/p_D,
  \end{align*}
  per \Eqn{td}. Hence, from \Thm{Hoeffding} we have
  \begin{displaymath}
    \Prob{ | \hat t_D/p_D - t_D/p_D | \geq \eps } \leq \delta,
  \end{displaymath}
  and the theorem follows by multiplying the inequality by $p_D$.
\end{proof}

\subsection{Computing a Random Sample of the Triangles}

In addition to knowing the number of triangles in a graph, it may also
be interesting to consider the properties of those triangles. For
instance, Durak et al.~\cite{DuPiKo12} consider the differences in
node degrees in a triangle.
 
It turns out that the closed wedges discovered during the wedge
sampling procedure are triangles sampled uniformly (with
replacements). Hence, we can study these randomly sampled triangles to
estimate the overall characteristics of triangles in the graph.
\begin{theorem}
  Let $W_s$ be a random sample of the wedges of  a graph $G$,  and let
  $T_s \subseteq W_s$ triangles that are formed by the closed wedges
  in $W_s$.  
  Then each triangle in $T_s$ is a uniform random sample from the
  triangles of $G$.  
\end{theorem}
\begin{proof} 
  The proof depends on observing that a triangle being chosen depends
  only on one of its 3 wedges being chosen.  Since the wedge sample is
  uniformly random, each triangle is equally likely to be picked, and
  there is no dependency between any pair of triangles, which implies
  a uniform sample.
\end{proof}
\subsection{Practical Performance of Wedge Sampling} 
Earlier work by a subset of the authors~\cite{SePiKo13} provides a
thorough study on  how the  techniques described above perform in
practice.  
As expected, tremendous  improvements are achieved in runtimes
compared to full enumeration, especially for large graphs, since the
number of samples is independent of graph size. 
Specifically, we see speed-ups of more than 1000X with errors in the
clustering coefficient of less than 0.002.
Additionally, in comparison to the Doulion method (an edge-sampling
technique) we obtain speed-ups of 5X or more while obtaining the same
accuracy.
The ability to adapt our wedge-sampling method to computing binned degree-wise
clustering coefficients and triangle sampling are also benefits in
comparison to edge-based sampling.
 
Our goal in this work is to 
implement the wedge sampling approach within the MapReduce framework
and provide evidence that it can scale to much larger problems.

\section{MapReduce Implementation}

\newcommand{\PhaseOne}{Compute Degree-based Statistics}
\newcommand{\PhaseOneA}{Compute Degree per Vertex}
\newcommand{\PhaseOneB}{Compute Number of Wedges per Bin}
\newcommand{\PhaseOneC}{Gather Wedges per Bin}
\newcommand{\PhaseTwo}{Select Wedge Samples}
\newcommand{\PhaseTwoA}{Select Sample Wedge Centers}
\newcommand{\PhaseTwoB}{Gather Sample Wedge Centers}
\newcommand{\PhaseTwoC}{Create Sample Wedges}
\newcommand{\PhaseThree}{Check Sample Wedge Closure}
\newcommand{\PhaseThreeA}{Gather Sample Wedge Closure Hashes}
\newcommand{\PhaseThreeB}{Check Sample Wedge Closure}
\newcommand{\PhaseFour}{Post-processing}
\newcommand{\PhaseFourA}{Find Degree of First Vertex per Wedge}
\newcommand{\PhaseFourB}{Find Degree of Second Vertex per Wedge}
\newcommand{\PhaseFourC}{Summarize Results}

\subsection{Overview}

We now present a MapReduce algorithm for estimating the clustering
coefficients and number of triangles in a graph. For details on
MapReduce, we refer the reader to Lin and Dyer \cite{LiDy10}; we have
emulated their style in our algorithm presentations.
We use the open-source Hadoop implementation of MapReduce, and the
Hadoop Distributed File System (HDFS) for storing data.
Each MapReduce job takes one or more distributed files as input.
These files are automatically stored as \emph{splits} (also known sometimes as blocks),
and one  mapper is launched per split.
The mappers produce key-value pairs. All values with the same key are
sent to the same reducer. The number of reducers is specified by the user.
Each MapReduce job produces a single HDFS output file.
A MapReduce job accepts \emph{configuration parameters},
which are passed along as data to the mapper and reducer functions; we
discuss these in more detail in the sections that follow. 
The set of MapReduce jobs in our algorithm is coordinated by a Hadoop Java
program running on a single \emph{client node}.

In our code, we assume the nodes are binned by degree as
discussed in \Sec{bin}. Computation of the global clustering
coefficient is a special case which can be computed by either looking
only at a single bin containing all degrees or using a
weighted average of the binned clustering coefficients (see \Sec{Phase4c}).

Our input is an undirected edge list where the node identifiers are
64-bit integers; we assume no duplicates or
self-edges and no particular ordering. 
We divide our MapReduce algorithm
into three \emph{major phases} plus post-processing, as presented in
\Fig{overview}.  Each major phase makes a complete pass through the
edge list.
The first phase sets up the distribution on wedges.
The second phase creates the sample wedges.
Finally, the third phase checks whether or not the
sample wedges are closed.
In all three phases, we have strategies to reduce the data volume in
the shuffle phase (between the map and reduce), discussed in detail
in the sections that follow.

\begin{figure}[htp]
  \centering
  \includegraphics[width=\textwidth]{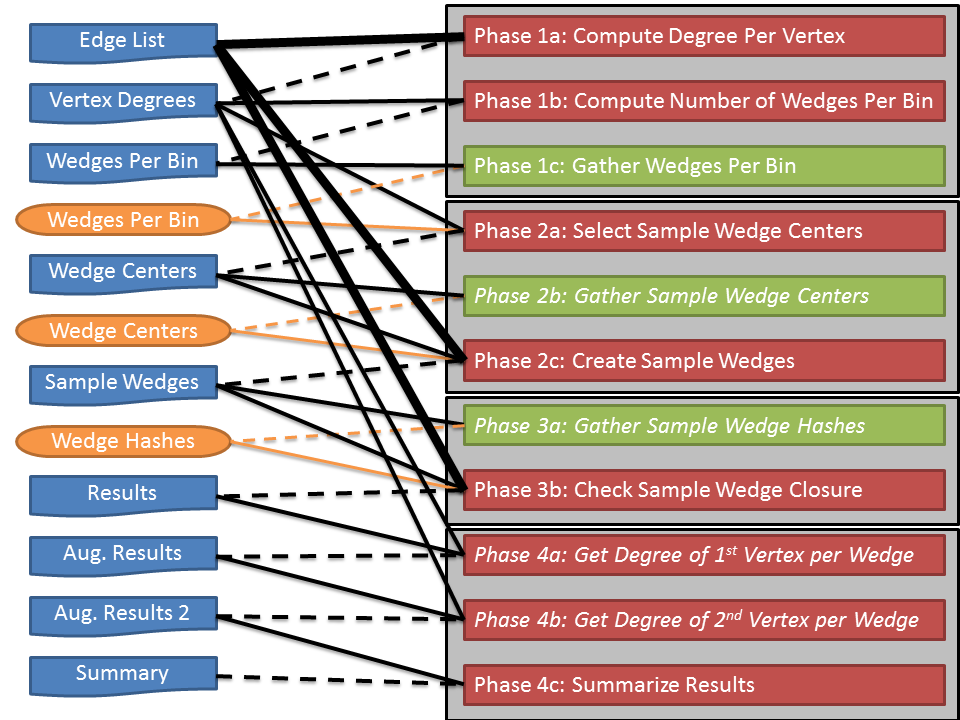}
  \caption{Algorithm overview for estimating clustering coefficients
    and counting triangles, both binned by degree. Red
    boxes indicate a MapReduce job, while green represents a serial
    operation on the client node.
    Blue boxes indicate data files. The edge list is provided by
    the user; all other data files are produced by the method. 
    Orange boxes indicate data that is passed as a ``configuration
    parameter'' to all mappers. Solid lines
    indicate consumption of data. Dotted lines indicate creation of data.
}
  \label{fig:overview}
\end{figure}

\subsection{Binning}
\label{sec:bin}

We define degree bins in a parameterized way as follows. 
Let $\tau$ be the number of singleton bins, and let $\omega > 1$ be the
rate of growth on the bin sizes. The first $\tau$ bins are singletons
containing 
degrees $1, 2, \cdots, \tau$
respectively. The remaining bins grow exponentially in size.

We describe the lowest degree of bin $k$ as
\begin{equation}
\label{eq:binlodeg}
  \textsc{BinLoDeg}(k) =
  \begin{cases}
    k, & \text{if } k \leq \tau, \\
    \Ceil{ (\omega^{(k-\tau)} - 1)/(\omega - 1) } + \tau,
    & \text{otherwise.}
  \end{cases}
\end{equation}
The highest degree for bin $k$ is just one less than the lowest degree
of bin $k+1$.
For a given degree $d$, we can easily look up its bin  as
\begin{equation}
\label{eq:binid}
  \textsc{BinId}(d) =
  \begin{cases}
    d, & \text{if } d \leq \tau, \\
    \Floor{\log(1 + (\omega - 1)(d - \tau)) / \log(\omega)} + \tau, &
    \text{otherwise.}
  \end{cases}
\end{equation}
In our implementation, $\tau$ and $\omega$ are communicated to each
MapReduce job as configuration parameters.

For $\tau=2$ and $\omega=2$, the bins are
$\set{1}$, $\set{2}$, $\set{3,4}$, $\set{5,6,7,8}$, $\set{9,\dots,16}$,
$\set{17, \dots, 32}$, and so on.
Note that the bin $\set{1}$ cannot have any
wedges, so we just ignore it.
Let $\bar d$ be an upper bound on the highest degree for a given
graph. Then choosing $\tau=1$ and $\omega = \bar d$ yields bins
$\set{1}, \set{2,\dots,\bar d}$.
In other words, we have a single bin containing all vertices
(excepting degree-1 vertices).
On the other hand, choosing $\tau=\bar d$ yields
$\set{1}, \set{2}, \dots , \set{\bar d}$.
Here, every bin is a singleton.

We are not constrained to equation \Eqn{binid} for computing the
bins; we can use any procedure such that each degree is assigned to a
single bin. Likewise, \Eqn{binlodeg} is optional and used to reduce
the shuffle volume in Phase 2c.

\clearpage
\subsection{Phase 1: \PhaseOne}

\subsubsection{Phase 1a: \PhaseOneA}
Phase 1a is a straightforward MapReduce task---computing the degree of
each vertex. The \textsc{Map} and \textsc{Reduce} functions are
described in \Alg{phase1a}.  The input is the \textbf{edge list file};
each entry is a pair of vertex IDs $(v,w)$ that define an edge. 
The \textsc{Map} function is called for
each edge $(v,w)$ and emits two key-value pairs keyed to the vertex IDs
and having a value of 1.  The \textsc{Reduce} function gathers all the
values for each vertex and sums them to determine the degree. The
final output to HDFS is a \textbf{vertex degree file}; each entry is of the
form $(v,d)$ where $v$ is a vertex ID and $d$ is its degree. 

\begin{algorithm}[t]
  \caption{\PhaseOneA\@ (Phase 1a)}
  \label{alg:phase1a}
  \begin{algorithmic} 
    \Procedure{Map}{$v,w$} 
    \Comment{Input is \textbf{edge list file}}
    \State \textsc{Emit}$(v, 1)$
    \State \textsc{Emit}$(w, 1)$
    \Comment{Emit both for an undirected graph}
    \EndProcedure
    \State
    \Procedure{Reduce}{$v, \set{ x_1, x_2, \ldots }$}
    \State $d \gets \textsc{Sum}(\set{ x_1, x_2, \ldots })$
    \Comment{Compute degree}
    \State \textsc{Emit}$(v, d)$
    \Comment{Output is \textbf{vertex degree file}}
    \EndProcedure
  \end{algorithmic}
\end{algorithm}

\Alg{phase1a} shows a simple version of the code.
To make the code more efficient, we collect local counts within each
mapper (using a Java \texttt{Map} container) and emit the totals.
This technique is called an in-memory combiner \cite{LiSc10}.  
We found the in-memory combiner
to reduce shuffle volume more than employing the reducer as a combiner.

\subsubsection{Phase 1b: \PhaseOneB}

Phase 1b works with the output of Phase 1a (\textbf{vertex degree
  file}) to compute the number of wedges per bin.
The \textsc{Map} and \textsc{Reduce} functions for Phase 1b are
presented in \Alg{phase1b}. The input is the list of degrees per
vertex. The \textsc{Map} function is called for each vertex (with its
associated degree) and emits the number of wedges for that vertex,
keyed to the appropriate bin.  The \textsc{Reduce} function simply
combines the results for each bin. The final output is a
\textbf{wedges per bin file}; each entry is of the form $(b,n_b,p_b)$
where $b$ is the bin ID, $n_b$ is the number of vertices in the bin, and
$p_b$ is the number of wedges in the bin.

\begin{algorithm}[t]
  \caption{\PhaseOneB\@ (Phase 1b)}
  \label{alg:phase1b}
  \begin{algorithmic}
    \State \textbf{parameters:} $\tau, \omega$ \Comment{Binning parameters}
    \State %
    \Procedure{Map}{$v, d$} 
    \Comment{Input is \textbf{vertex degree file}}
    \State $b \gets \textsc{BinId}(d,\tau,\omega)$
    \Comment{Compute bin ID}
    \State $n \gets 1$
    \Comment{Number of vertices}
    \State $p \gets d \cdot (d-1)/2$
    \Comment{Number of wedges}
    \State \textsc{Emit}$(b, (n,p))$
    \EndProcedure
    \State %
    \Procedure{Reduce}{$b, \set{ (n_1,p_1), (n_2,p_2), \ldots }$}
    \Comment{One reduce function per bin}
    \State $n \gets \textsc{Sum}(\set{n_1, n_2, \ldots})$
    \Comment{Number of vertices in bin}
    \State $p \gets \textsc{Sum}(\set{ p_1, p_2, \ldots })$
    \Comment{Number of wedges in bin}
    \State \textsc{Emit}$(b,n,p)$ 
    \Comment{Output is \textbf{wedges per bin file}}
    \EndProcedure
  \end{algorithmic}
\end{algorithm}

Once again, we have shown a simple version of the algorithm in
\Alg{phase1b}. To make the code more efficient, we 
collect local counts within each
mapper (using a Java \texttt{Map} container)
and emit the totals. 

For the case of a single bin, strictly speaking, Phase 1b is
unnecessary. Instead, we could have used a Hadoop \emph{global counter} to tally the total wedges in the
reduce step of Phase 1a. 

\subsubsection{Phase 1c: \PhaseOneC}

From the \textbf{wedges per bin file} (output of Phase 1b),
we create a \textbf{wedges per bin object}, which acts as a function
$\theta$ such that $\theta(b)$ is the number of wedges in bin $b$.
The work is performed entirely in our main program running on the client node.
It reads Phase 1b output from HDFS, stores wedges per bin values in a
Java \texttt{Map} container, and launches the next MapReduce job (Phase 2a),
sending the \emph{serialized} container as a configuration parameter.

\subsection{Phase 2: \PhaseTwo}

\subsubsection{Phase 2a: \PhaseTwoA}

The input to Phase 2a is the \textbf{vertex degree file} along with
the \textbf{wedges per bin object}, which is passed as a configuration parameter.
Phase 2a calculates the number of sample wedges centered at each
vertex. The \textsc{Map} function is shown in \Alg{phase2a}.
The \textsc{Map} function is called for each (vertex ID, degree)
pair. From this, we can calculate the expected number of wedges that
would be sampled from the vertex for a uniform random sample, 
$q^*$. This number is unlikely to be integral. Rounding
up would produce far too many wedges. Instead, we use probabilistic
rounding. For instance, if $q^* = 0.1$, then there is a 10\% change of
producing $q=1$ wedges and a 90\% chance of producing no wedges,
$q=0$. We are only off by at most one, so if $q^* = 1.1$, then there
is a 10\% change of producing $q=2$ wedges and a 90\% chance of
producing $q=1$ wedge. Hence, the expected number of wedges for this vertex is exactly $q^*$.
Only vertices with at least one sample wedge
are emitted. 
The final output is a \textbf{wedge centers file}; each entry is of
the form $(v,d,q,p)$ where $v$ is the 
vertex ID, $d$ is the vertex degree, $q$ is the number of sample wedges centered
at that vertex, and $p$ is the total number of wedges in the bin
containing $v$.
The \textsc{Reduce} function is just the identity map and is not shown.

\begin{algorithm}[t]
  \caption{Determine number of samples per vertex (Phase 2a)}
  \label{alg:phase2a}
  \begin{algorithmic}
    \State \textbf{parameter:} $k$ 
    \Comment{Desired number of samples  per bin}
    \State \textbf{parameter:} $\theta$ 
    \Comment{Represents \textbf{wedges per bin object}}
        \State %
    \Procedure{Map}{$v, d$} \Comment{Input is \textbf{vertex degree file}}
    \State $b \gets \textsc{BinId}(d)$
    \Comment{Compute bin ID}
    \State $p \gets \theta(b)$
    \Comment{Total number of wedges in bin containing $v$}
    \State $q^* \gets (d \cdot (d-1)/2) \cdot k / p$  
    \Comment{Ideal number of samples, likely noninteger}
    \State $x \gets \textsc{rand }([0,1])$  \Comment{Uniform random number in $[0,1]$}
    \State $q \gets \set{x \geq (q^* - \Floor{q^*})}\text{ ? }\Ceil{q^*}:\Floor{q^*}$
    \Comment{Number of \emph{sample} wedges centered at $v$}
    \If{$q \geq 1$}
    \Comment{Skip vertices with no samples}
    \State \textsc{Emit}$(v,d,q,p)$
    \Comment{Output is \textbf{wedge centers file}}
    \EndIf
    \EndProcedure
  \end{algorithmic}
\end{algorithm}

\subsubsection{Phase 2b: \PhaseTwoB}

Phase 2b is an optional step that generates a Java \texttt{Map} of wedge centers and
their bin IDs based on the output of Phase 2a (\textbf{wedge centers file}). 
We represent this object as a function $\gamma$ such that
\begin{displaymath}
  \gamma(v) = 
  \begin{cases}
    0 & \text{if Phase 2b is skipped}, \\
    1 & \text{if $v$ is not a wedge center}, \\
    b \geq 2 & \text{if $v$ is a wedge center, in which case $b$ is
      the bin ID}.
  \end{cases}
\end{displaymath}
This \textbf{wedge centers object} has one value for every vertex appearing
in a wedge center.
It is serialized and
passed as a configuration parameter to Phase 2c, where it is used
to filter the edges that are emitted
by the \textsc{Map} function.

Note that Hadoop imposes a limit on the size of the configuration parameters
(5MB by default).
If the number of wedge centers is too large (a few hundred thousand samples
will exceed 5MB), then other options must be explored.
One alternative is to pass the container to the \textsc{Map} tasks using
the Hadoop distributed cache; however, we have not implemented this idea.

Phase 2b is optional, and can be skipped if there are too many wedge centers.
We demonstrate the benefits of this step in \Sec{results}.

\subsubsection{Phase 2c: \PhaseTwoC}

In Phase 2c, the goal is to take each sample wedge center (from the
\textbf{wedge center file}), collect its neighbors (from the
\textbf{edge list file}), and create a set of sample wedges. We merge
each vertex and its neighbors at the reduce phase.  If it exists, the
optional \textbf{wedge center object} is used to filter the edges that
are shuffled, ignoring all edges that are not adjacent to a sampled
wedge center.
The algorithm is shown in \Alg{phase2c}. For clarity, we give a separate \textsc{Map}
function for each input type. In the actual implementation, we have to
determine the input type on the fly, because both input files are of Hadoop
type \texttt{Text}.
For input from the \textbf{wedge centers file},
 the \textsc{Map} function simply passes along its degree and
sample wedge count (i.e., the number of wedges to be sampled from the vertex). 

For input from the \textbf{edge list file}, the \textsc{Map} function checks
to see if the edge is adjacent to a wedge center. If so, it is passed
based on the outcome of a random coin flip. The aim of the \textsc{Reduce} phase
is to generate random wedges centered at a vertex (say $v$). The most na\"{i}ve \textsc{Map} implementation
would forward all edges incident to $v$, so that wedges can be selected from them.
A major problem with this is that if the number of samples $k$ is much 
less than the degree of $v$, most of the communication is unnecessary.
For example, the highest degree vertices of a social network graph might
link to millions of edges, but $k$ is in the tens of thousands or less;
therefore, most of the incident edges will not participate in sampled
wedges centered at these vertices.
We have a probabilistic fix to address this situation.

We do not have the vertex degree readily available, but we do know its
bin and therefore a lower bound on its degree.
Consider a vertex $v$ of degree greater than $d_{\min}$, where
$2k \leq d_{\min}/2$.  We send just some of the incident edges to $v$,
with independent probability $\phi = 4k/d_{\min} \leq 1$.
Then the expected number of edges to send is $4k(d_v/d_{\min})$.
Note that this expectation is at least $4k$. Getting less than $2k$ edges is potentially disastrous,
but the probability of this is minuscule. By a multiplicative Chernoff bound (given below),
the probability of such an event is $\exp(-k/8)$. For $k = 1000$ (a tiny sample size), the probability
is less than $10^{-55}$.

\begin{theorem}[Multiplicative Chernoff Bound \cite{DuPa}] %
Let $X = \sum_{i \leq r} X_i$, where each $X_i$ is independently distributed in $[0,1]$.
Then
  \begin{displaymath}
    \Prob{ X \leq (1-\delta)\Exp{X} } \leq \exp(-\delta^2\Exp{X} /2).
  \end{displaymath}
\end{theorem}%
If $d_v$ is not too far from $d_{\min}$, then the expectation $4k(d_v/d_{\min})$
is potentially much smaller than $d_v$. Hence, we get the desired number of random edges
without sending too many.

Even with this improvement, the data passed forward
may be too large to fit 
into the reducer's memory. We use a feature of Hadoop
called \emph{secondary sort} to ensure that the data arrives
pre-sorted. 
Note that the key used for passing along the vertex
information is $v\mathtt{:}0$ and the key for the edges is of the form
$v\mathtt{:}y$, where $y$ is a random positive integer. 
This data is all mapped to the key $v$, but the values
following the colon control the sort of the values associated with
$v$. The secondary key of zero ensures that the degree and wedge count
data are first. The secondary keys for edges ($y$) ensure that
the adjacent edges are randomly sorted; otherwise, Hadoop would present the edges
in their order of arrival, which could bias the selection.

\begin{algorithm}[p]
  \caption{\PhaseTwoC\@ (Phase 2c)}
  \label{alg:phase2c}
  \begin{algorithmic}
    \State \textbf{parameters:} $\tau, \omega$ \Comment{Binning parameters}
    \State \textbf{parameter:} $k$ 
    \Comment{Desired number of samples  per bin}
    \State \textbf{parameter:} $\gamma$
    \Comment{Represents \textbf{wedge centers object}}
    \State %
    \Procedure{Map}{$v, d, q, p$}
    \Comment{Input is \textbf{wedge centers file}}
    \State \textsc{Emit}$(v\mathtt{:}0, (d,q,p))$
    \Comment{Note secondary sort key}
    \EndProcedure
    \State %
    \Procedure{Map}{$v,w$} 
    \Comment{Input is \textbf{edge list file}}
    \State \textsc{EdgeHelper}($v,w$)
    \State \textsc{EdgeHelper}($w,v$)
    \EndProcedure
    \State %
    \Procedure{EdgeHelper}{$v,w$}
    \State $b \gets \gamma(v)$
    \Comment{Extract bin ID}
    \If{$b = 0$}
    \Comment{Phase 2b was skipped}
    \State $\phi \gets 1$
    \Comment{Always emit the edge}
    \ElsIf{$b=1$}
    \Comment{Vertex is not a wedge center}
    \State $\phi \gets 0$
    \Comment{Never emit the edge}
    \Else %
    \Comment{Vertex is a wedge center}
    \State $d_{\min} = \textsc{BinLoDeg}(b,\tau,\omega)$
    \Comment{Lower bound degree of $v$}
    \State $\phi \gets 2 \cdot ({2k}/{d_{\min}})$
    \Comment{Proportion of edges to emit for $v$}
    \EndIf
    \State $x \gets \textsc{Rand}([0,1])$ 
    \Comment{Uniform random number in $[0,1]$}
    \If{$x \leq \phi$}
    \Comment{Probabilisticly downselect}
    \State $y \gets \textsc{Rand}(\set{1,\dots, \mathtt{maxlongint}})$
    \Comment{Random long integer}
    \State \textsc{Emit}$(v\mathtt{:}y, w)$
    \Comment{Note secondary sort key}
    \EndIf
    \EndProcedure
    \State %
    \Procedure{Reduce}{$v,\set{x_1,x_2,\dots}$}
    \If{$x_1$ is a wedge center}
    \Comment{If it exists, the wedge center information is first}
    \State $(d,q,p) \gets x_1$ 
    \Comment{Unpack wedge center information}
    \State $(d',\set{(i_{\ell},j_{\ell})}_{\ell=1}^q\}) \gets \textsc{Sampling}(d,q)$
    \Comment{Determine sample wedges}
    \State $\set{w_1, \dots, w_{d'}} \gets \set{x_2, x_3, \dots, x_{d'+1}}$
    \Comment{Read only $d'$ neighbors}
    \ForAll{$\ell = 1,\dots,q$}
    \State $h \gets \textsc{Hash}(w_{i_q},w_{j_q})$
    \Comment{Hash of edge that would close this wedge}
    \State \textsc{Emit}$(h, v, w_{i_q}, w_{j_q}, p, d)$
    \Comment{Output is \textbf{sample wedges file}}
    \EndFor
    \EndIf
    \EndProcedure    
    \State %
     \Procedure{Sampling}{$d,q$} 
    \Comment{Subroutine for simulated sampling}
    \ForAll{$\ell=1,\dots,q$} \Comment{Generate endpoints for each wedge}
    \State $i_{\ell} \gets \textsc{Rand}\set{1,\dots,d}$
    \State $j_{\ell} \gets \textsc{Rand}\set{1,\dots,d}\setminus \set{i_{\ell}}$
    \EndFor
    \State $\mathcal{S} \gets \set{i_1,\dots,i_q} \cup \set{j_1,\dots,j_q}$ 
    \Comment{Gather unique indices (duplicates removed)}
    \State $d' \gets |\mathcal{S}|$
    \Comment{Number of edges needed}
    \State Define mapping $\pi: \mathcal{S} \rightarrow \set{1,\dots,d'}$ \Comment Renumber from 1 to $d'$
    \State \Return $d'$ and pairs $\set{(\pi(i_{\ell}),\pi(j_{\ell}))}_{\ell=1}^q\}$
    \EndProcedure
  \end{algorithmic}
\end{algorithm}

From the secondary sort, the wedge center must be first in the values
list at the reduce phase, if it
exists. If it does not exist, then there is nothing to do.
Recall that for each wedge center $v$, we have its degree, $d$, and a
desired number of wedge samples, $q$.  Each wedge must be randomly
sampled \emph{with} replacement. The two edges of a single wedge are
sampled \emph{without} replacement. So, wedge sampling requires a
minimum of 2 and a maximum of $2q$ edges. If $2q > d$, some
edges are necessarily reused. If $2q \ll d$, it is more 
likely that every wedge centered at $v$ will have two unique edges; however, 
due to the birthday paradox, there remains a non-negligible 
likelihood of wedge overlap even for large $d$. 
For these large $d$, we want to
avoid reading all neighbors into memory since the list is quite long, but we still want to accurately reproduce uniform sampling with replacement.  
We do this
by using a simulated sampling procedure explained below. 
It only requires reading   the first $d'$
neighbors into memory where $d' \leq \min\set{d,2q}$.

Procedure \textsc{Sampling} produces $q$ uniform random wedges (with replacement) centered at $v$. Number the edges incident to $v$
arbitrarily from $1$ to $d$. A uniform random wedge is represented as a uniform random pair of indices $(i,j)$ ($i \neq j, i \leq d, j \leq d$). We can repeat this random index selection $q$ times to implicitly sample
$q$ random wedges, each of which is just represented as a pair of indices. Observe that the total
number of indices in the union of these pairs is at most $d'$, so all we need are the the first $d'$ uniform randomly ordered edges
obtained as the output of the Map phase. We map these sampled wedge indices randomly to the index set $\{1,2,\ldots,d'\}$
through a permutation. Now, each wedge is indexed as a pair $(i,j)$ ($i \neq j, i \leq d', j \leq d'$).
From the list of edges/neighbors $\{x_1, x_2, \ldots,x_{d'}\}$, we can generate these random edges.
This is what is done in \textsc{Sampling} and \textsc{Reduce} in \Alg{phase2c}.

The final output of this phase
is a \textbf{sample wedge list file}, where each entry is of the form
$(h, v_0, v_1, v_2, p, d_0)$.  The number $h$ is a hash of the desired closure
edge $(v_1, v_2)$ (a key which allows the undirected edges from the edge list to be correctly matched with the closure requests in Phase 3b.g),  
the wedge is defined by $(v_1,v_0,v_2)$,
$p$ is the total number of wedges in the bin containing $v_0$,
and $d_0$ is the degree of vertex $v_0$.

As mentioned above, Phase 2b is optional.  If skipped, the MapReduce
shuffle brings adjacent edges of a wedge center together in the reduce phase.
We defer calculation of the number of sample wedges to the reduce phase,
but otherwise proceed as defined above.
Note that in many cases the reducer collects zero samples and does no work.

\subsection{Phase 3: \PhaseThree}

\subsubsection{Phase 3a: \PhaseThreeA}

Phase 3a (optional) assembles a list of all the unique edge hashes from
the \textbf{sample wedges file} and stores it as a Java \texttt{Set}
object.  We denote this \textbf{wedge hashes object} by
$\xi = \set{h_1,h_2,\dots}$. This is similar to the
procedure in Phase 2b, which assembles the list of wedge centers. 
We set $\xi=\emptyset$ if Phase 3a is skipped.

\subsubsection{Phase 3b: \PhaseThreeB}

Phase 3b is the last major step and checks the wedge closures, as shown in
\Alg{phase3b}.  The inputs are the \textbf{sample wedges file} created by
Phase~2c and the original \textbf{edge list file}. 
We also pass the optional \textbf{wedge hashes object} ($\xi$) as a configuration
parameter. If $\xi$ is nonempty, it is used to filter the edges passed
forward to the reduce function.
(Note that we could skip Phase
3a and forward every edge forward to the reducers, but this would
result in much greater data shuffling in Phase 3b.) 
Note that more than one edge may hash to the same value; hence, we
loop through all edges that arrive at  the reducer to verify that there is a match
before declaring a wedge as closed.
Likewise, more than one wedge may be closed by a single edge.
The output of this phase is the \textbf{results file (ver.~0)}; each
entry is of the form ($\sigma, v_0,v_1,v_2, p, d_0$) where $\sigma$
indicates if the wedge is open or closed and everything else is the
same as for the \textbf{sample wedges file}.

\begin{algorithm}[t]
  \caption{\PhaseThreeB\@ (Phase 3b)}
  \label{alg:phase3b}
  \begin{algorithmic}
    \State \textbf{parameter:} $\xi  = \set{h_1,h_2,\dots}$
    \Comment{Represents \textbf{wedge hashes object}}
    \State %
    \Procedure{Map}{$h,v_0,v_1,v_2,p,d_0$} 
    \Comment{Input is \textbf{sample wedges file}}
    \State \textsc{Emit}($h$, $(v_0,v_1,v_2,p,d_0)$)
    \EndProcedure
    \State %
    \Procedure{Map}{$w_1,w_2$} 
    \Comment{Input is \textbf{edge list file}}
    \State $h \gets \textsc{Hash}(w_1,w_2)$
    \Comment{Hash of edge}
    \If{($\xi = \emptyset$) or ($h \in \xi$)} 
    \State \textsc{Emit}($h$, $(w_1,w_2)$)
    \EndIf
    \EndProcedure
    \State %
    \Procedure{Reduce}{$h, \set{x_1,x_2,\dots}$}
    \State Sort the values $\set{x_1,x_2,\dots}$ into 
    $\mathcal{E}$ (edges) and $\mathcal{W}$ (wedges)
    \ForAll{$w \in \mathcal{W}$}
    \State $(v_0,v_1,v_2,p,d_0) \gets w$ \Comment{Unpack wedge data}
    \State $\sigma \gets \text{``open"}$ \Comment{By default, wedges are open}
    \ForAll{$e \in \mathcal{E}$}
    \State $(w_1, w_2) \gets e$ \Comment{Unpack edge data}
    \If{($w_1 = v_1$ and $w_2 = v_2$) or ($w_2 = v_1$ and $w_1 = v_2$)}
    \State $\sigma \gets \text{``closed"}$ 
    \EndIf
    \EndFor
    \State \textsc{Emit}($\sigma, v_0,v_1,v_2, p, d_0$) 
    \Comment{Output is \textbf{results file (ver.~0)}}
    \EndFor
    \EndProcedure
  \end{algorithmic}
\end{algorithm}

\subsection{Phase 4: \PhaseFour}

\subsubsection{Phases 4a \& 4b: Find degrees of wedge endpoints}

Phases 4a \& 4b augment each sample wedge with the degrees of $v_1$ and
$v_2$. This information is needed for estimating the number of
triangles per bin. If only the clustering coefficients are required,
these two steps can be omitted.
\Alg{phase4a} shows Phase 4a; the procedure for Phase 4b is analogous
and so is omitted. The final output of Phase 4b is the \textbf{results
file (ver.~2)}; each line is of the form ($\sigma, v_1,v_o,v_2, p, d_0, d_1,
d_2$) where $d_1$ and
$d_2$ are the degrees of vertices $v_1$ and $v_2$,
respectively, while the remainder is the same as for the
\textbf{results file (ver.~0)}.

\begin{algorithm}[t]
  \caption{\PhaseFourA\@ (Phase 4a)}
  \label{alg:phase4a}
  \begin{algorithmic}
    \State %
    \Procedure{Map}{$\sigma, v_0,v_1,v_2, p,d_0$} 
    \Comment{Input is \textbf{results file (ver.~0)}}
    \State \textsc{Emit}($v_1\texttt{:}1$, $(\sigma,v_0,v_1,v_2,p,d_0)$)
    \EndProcedure
    \State %
    \Procedure{Map}{$v, d$} 
    \Comment{Input is \textbf{vertex degree file}}
    \State \textsc{Emit}($v\texttt{:}0$, $d$)
    \EndProcedure
    \State %
    \Procedure{Reduce}{$v, \set{x_1,x_2,\dots}$}
    \Comment{Add the degree of $v_1$ for  each wedge.}
    \State $d_1 \gets x_1$ 
    \Comment{First value is the degree of the vertex}
    \ForAll{$x \in \set{x_2,x_3\dots}$}
    \Comment{Remaining values, if any, comprise sample wedges}
    \State \textsc{Emit}($x, d_1$) 
    \Comment{Output is \textbf{results file (ver.~1)}}
    \EndFor
    \EndProcedure
  \end{algorithmic}
\end{algorithm}

\subsubsection{Phase 4c: \PhaseFourC}
\label{sec:Phase4c}

\begin{algorithm}[t]
  \caption{\PhaseFourC\@ (Phase 4c)}
  \label{alg:phase4c}
  \begin{algorithmic}
    \State \textbf{parameters:} $\tau, \omega$ \Comment{Binning parameters}
    \State %
    \Procedure{Map}{$\sigma, v_0, v_1, v_2, p, d_0, d_1, d_2$} 
    \Comment{Input is \textbf{results file (ver.~2)}}
    \State $b_0 \gets \textsc{BinId}(d_0, \tau, \omega)$
    \State $b_1 \gets \textsc{BinId}(d_1, \tau, \omega)$
    \State $b_2 \gets \textsc{BinId}(d_2, \tau, \omega)$
    \State \textsc{Emit}($b$, ($\sigma, p, b_1, b_2$))
    \EndProcedure
    \State %
    \Procedure{Reduce}{$b, \set{x_1,x_2,\dots}$}
    \State $p_0, p_1, p_2, p_3 \gets 0$
    \ForAll{$x \in \set{x_1,x_2\dots}$}
    \State $(\sigma,p,b_1,b_2) \gets x$ \Comment{Unpack value}
    \If{$\sigma$ = ``open''}
    \State $q_0 \gets q_0 + 1$
    \Else
    \State $i \gets 1 + (b=b1) + (b=b2)$
    \State $q_i \gets q_i + 1$
    \EndIf
    \EndFor
    \State $c \gets (q_1 + q_2 + q_3) / (q_0 + q_1 + q_2 + q_3)$
    \State $t \gets p \cdot (q_1 + q_2/2 + q_3/3) / (q_0 + q_1 + q_2 + q_3)$
    \State \textsc{Emit}($b, q_0, q_1, q_2, q_3, c, p, t$)
    \Comment{Output is \textbf{summary file}}
    \EndProcedure
  \end{algorithmic}
\end{algorithm}

Phase 4c tallies the final results per bin, using the logic in \Alg{phase4c}.
Its output is the \textbf{summary file}. Each line is of the form 
$b, q_0, q_1, q_2, q_3, c, p, t$ where $b$ is the bin ID, $q_0$ is the
number of open wedges, $q_i$ is the number of closed wedges with $i$
vertices in the bin, $c$ is the clustering coefficient estimate, $p$
is the number of wedges in the bin, and $t$ is the estimated number of
triangles with one or more vertices in the bin.

We can estimate the global clustering coefficient from the
degree-binned clustering coefficients as follows. Let
$\hat c_b$ and $p_b$ be the clustering coefficient estimate and total number of
wedges for bin $b$. Let $p = \sum_b p_b$ be the total number of
wedges. Then the estimates for the global clustering coefficient and
total number of triangles are given by
\begin{equation}
  \label{eq:global_est}
  \hat c \approx \sum_b \frac{p_b}{p} \cdot \hat c_b,
  \qtext{and}
  \hat t = \hat c \cdot \frac{p}{3}.
\end{equation}
Let $b_{\max}$ denote the total number of bins. We assume that every
bin has $k$ samples producing an error bound of $\eps$ with confidence
$(1-\delta)$. Then we argue that $|c-\hat c| \leq \eps$ with
confidence $(1-b_{\max}\cdot \delta)$.

\subsection{Performance Analysis}

\Tab{volumes} presents the input, communication, and output volume for each phase.
Let $n$ denote the number of nodes, $m$ denote the number of edges (each undirected edge is counted just once), $q$ denote the total number of sampled wedges, and let $b_{\max}$ denote the number of bins. Note that the communications in Phases 2c and 3b can be substantially higher ($m+q$) if the Phase 2b or 3a is skipped. Our experimental results show that Phase 1a is by far the most expensive, which is consistent with our performance analysis because Phase 1a communicates the most data, $2m$ key-values pairs. All other communications are size $n$ or $q$.

\newcolumntype{Y}{>{\footnotesize\raggedright\arraybackslash}X}
\begin{table}[h]\footnotesize
\caption{Input, shuffle, and output volumes for MapReduce phases.}
\label{tab:volumes}
\begin{tabularx}{\linewidth}{|l|Y|Y|Y|}
\hline
Phase & Mapper Input & Key-Value Pairs & Reducer Output \\ \hline
1a 
& $\bm{m}$: Process $m$ edges from \textbf{edge list file}.
& $\bm{2m}$: Communicate 2 key-value pairs per edge.
& $\bm{n}$: Output $n$ vertex-degree pairs to \textbf{vertex degrees file}.
\\ \hline
1b 
& $\bm{n}$: Process $n$ vertex-degree pairs from \textbf{vertex degrees file}.
& $\bm{n}$: Communicate 1 key-value pair per vertex.
& $\bm{b_{\max}}$: Output data for each bin to \textbf{wedges per bin file}.
\\ \hline
2a
& $\bm{n}$: Process $n$ vertex-degree pairs from \textbf{vertex degrees file}.
& \multicolumn{2}{>{\raggedright}p{0.55\textwidth}|}{$\bm{q}$: Communicate/output approximately $q$ sample wedge centers in \textbf{wedge centers file}.}
\\ \hline
2c
& $\bm{m+q}$: Process $m$ edges from \textbf{edge list file} and approximately $q$ sample wedge centers from \textbf{wedge centers file}.
& $\bm{O(qk)+q}$: Communicate $O(k)+1$ key-value pair per sample wedge center. 
& $\bm{q}$: Output the sample wedges to the \textbf{sample wedges file}.
\\ \hline
3b
& $\bm{q+m}$: Process wedges in \textbf{sample wedges file} and edges from \textbf{edge list file}.
& $\bm{2q}$: Communicate 1 message per wedge and 1 message per hash-matching edge. 
& $\bm{q}$: Output close/open data for each sample wedge into \textbf{results file (ver. 0)}.
\\ \hline
4a/4b 
& $\bm{n+q}$: Process $q$ sampled wedges from \textbf{results file (ver. 0/1)} and $n$ vertex-degree pairs to \textbf{vertex degrees file}.
& $\bm{n+q}$: Communicate 1 key-value pair for each vertex and each edge.
& $\bm{q}$: Output augmented data for each sample wedge into \textbf{results file (ver. 1/2)}.
\\ \hline
4c
& $\bm{q}$: Process $q$ sampled wedges from \textbf{results file (ver. 2)}.
& $\bm{3q}$: Communicate 3 key-values pairs per wedge.
& $\bm{b_{\max}}$: Output results per bin in \textbf{summary files}.
\\ \hline
\end{tabularx}
\end{table}

\section{Experimental Results}
\label{sec:results}

\subsection{Data Description}
\label{sec:data}

We obtained real-world graphs from the Laboratory for Web Algorithms
(\url{http://law.di.unimi.it/datasets.php}), which were compressed using
LLP and WebGraph \cite{BoVi04,BoRoSaVi11}. We selected ten larger
graphs for which the complete edge lists were available. 
We also consider three artificially-generated graphs according to the
Graph500 benchmark ~\cite{GRAPH500}, which uses
Stochastic Kronecker Graphs (SKG)~\cite{LeChKlFa10} for its graph
generator with [0.57,0.19;0.19,0.05] as the $2 \times 2$ generator matrix.
We have added noise with a parameter of 0.1, as proposed
in~\cite{SePiKo11,SePiKo13} to avoid oscillatory degree distributions. 
These graphs are generated in MapReduce.
All networks
are treated as undirected for our study; in other words, if $x
\rightarrow y$, $y \rightarrow x$, or both, we say that edge $(x,y)$
exists. Briefly, the networks are described as follows.
\begin{compactitem}
\item amazon-2008 \cite{BoVi04,BoRoSaVi11}: A graph describing similarity among books as
  reported by the Amazon store.
\item ljournal-2008 \cite{ChKuLaMi09,BoVi04,BoRoSaVi11}: Nodes represent users on
  LiveJournal. Node $x$ connects to node $y$ if $x$ registered $y$ as
  a friend.
\item hollywood-2009, hollywood-2011 \cite{BoVi04,BoRoSaVi11}: This is a graph of actors. Two
  actors are joined by an edge whenever they appear in a movie together.
\item twitter-2010 \cite{KwLePaMo10,BoVi04,BoRoSaVi11}: Nodes are Twitter users, and
  node $x$ links to node $y$ if $y$ follows $x$.
\item it-2004 \cite{BoVi04,BoRoSaVi11}: Links between web pages on the .it domain, provided by IIT.
\item uk-2005-05, uk-2006-06, uk-union-2006-06-2007-05 (shorted to
  uk-union)
  \cite{BoSaVi08,BoVi04,BoRoSaVi11}: Links between web pages on the .uk domain. (We
  ignore the time labeling on the links in the last graph.)
\item sk-2005 \cite{BoCoSaVi04,BoVi04,BoRoSaVi11}: Links between web pages on the .sk domain.
\item graph500-23, graph500-26, graph500-29 \cite{GRAPH500,LeChKlFa10,SePiKo11,RMATCODE}: Artificially generated
  graphs according to the Graph500 benchmark using the SKG method. The
  number (23, 26, 29) indicates the number of levels of recursion and
  the size of the graph.
\end{compactitem}
The properties of the networks are summarized in \Tab{networks};
specifically, we report the number of vertices, the
number of \emph{undirected} edges, the total number of wedges, and 
estimates for the total number of triangles and the global clustering
coefficients, calculated according to \Eqn{global_est}.
(To the best of our knowledge, we are the only group
that has calculated the last three columns, so these numbers have not
been independently validated.)

\begin{table}[htpb]\footnotesize
  \centering
  \caption{Network characteristics. All edges are treated as
      undirected. The triangle counts and global clustering coefficients (GCC)
      are our estimates.}
    \label{tab:networks}
  \begin{tabular}{|c|l|r|r|r|r|c|}      
  \hline 
  \MC{1}{|c|}{\bf ID} & \MC{1}{c|}{\bf Graph Name} & \MC{1}{c|}{\bf Nodes} & 
  \MC{1}{c|}{\bf Edges} & \MC{1}{c|}{\bf Wedges} & 
  \MC{1}{c|}{\bf Triangles} & {\bf GCC} \\ 
  \MC{1}{|c|}{} & \MC{1}{c|}{} & \MC{1}{c|}{\bf (millions)} & 
  \MC{1}{c|}{\bf (millions)} & \MC{1}{c|}{\bf (millions)} & 
  \MC{1}{c|}{\bf (millions)} & \\
  \hline
   1 &     amazon-2008 &     1 &     4 &            51 &         4 &  0.2603 \\ \hline
   2 &   ljournal-2008 &     5 &    50 &         9,960 &       408 &  0.1228 \\ \hline
   3 &  hollywood-2009 &     1 &    56 &        47,645 &     4,907 &  0.3090 \\ \hline
   4 &  hollywood-2011 &     2 &   114 &       120,899 &     7,097 &  0.1761 \\ \hline
   5 &     graph500-23 &     5 &   128 &       567,218 &     3,673 &  0.0194 \\ \hline
   6 &         it-2004 &    41 & 1,027 &    16,163,308 &    48,788 &  0.0091 \\ \hline
   7 &     graph500-26 &    34 & 1,054 &     9,087,164 &    28,186 &  0.0093 \\ \hline
   8 &    twitter-2010 &    42 & 1,203 &   123,435,590 &    34,495 &  0.0008 \\ \hline
   9 &      uk-2006-06 &    80 & 2,251 &    16,802,569 &   186,453 &  0.0333 \\ \hline
  10 &         sk-2005 &    43 & 2,543 & 5,196,166,169 &   256,556 &  0.0001 \\ \hline
  11 &      uk-2006-05 &    77 & 2,636 &   167,591,218 &   363,111 &  0.0065 \\ \hline
  12 &        uk-union &   132 & 4,663 &   203,567,548 &   447,133 &  0.0066 \\ \hline
  13 &     graph500-29 &   240 & 8,502 &   158,727,767 &   272,931 &  0.0052 \\ \hline
  \end{tabular}
\end{table}

\subsection{Experimental Setup}
We have conducted our experiments on a  Hadoop cluster with 32 compute nodes. Each compute node has an
Intel i7 930 CPU at 2.8GHz (4 physical cores, HyperThreading enabled),
12 GB of memory, and 4 2TB SATA disks. All experiments were run using
Hadoop Version 0.20.203.0.
Unless otherwise stated, all experiments use the following parameters:
\begin{compactitem}
  \item Number of samples per bin: $k=10,000$
  \item Bin parameters: $\tau=2$ and $\omega=2$ (i.e., bins are $\set{2},
    \set{3,4}, \set{5,6,7,8}, \dots$.)
  \item Number of reducers: 64
\end{compactitem}

\subsection{Experimental Results and Timings}
\label{sec:timings}

We ran our MapReduce code on the 13 networks described in
\Sec{data}. 
The runtimes are reported in \Fig{timings}, broken down by the phases of the algorithm. 
The largest real-world graph, uk-union (\#12) with over 100 million
vertices and over 4.6 billion edges, took less than 30 minutes to
analyze. For all the networks, the most expensive step is Phase 1a,
calculating the degree per node, because every edge in the
\textbf{edge list file} generates two
key-value pairs.
For uk-union, this step takes over 12 minutes. 
Phases 1b and 2a are essentially constant time (approximately 30
seconds) because they process only the \textbf{vertex degree file}.
Phase 2c (\PhaseTwoC) is the next most expensive step. Here we collect
the edges adjacent to wedge centers, reading the entire
\textbf{edge list file} again in the map phase; however, the
\textbf{wedge centers file} created in Phase 2b minimizes the number of edges that are passed to reducers.
Nevertheless, for the wedge centers with
high degree, many edges are transmitted (though substantially less
than the entire edge list). 
Phase 3b (\PhaseThreeB) is the next most expensive,
again reading the entire \textbf{edge list file}. The last few
postprocessing steps are close to constant time (approximately 30-60
seconds each). 

\begin{figure}[tpb]
  \centering
  \includegraphics[width=4in]{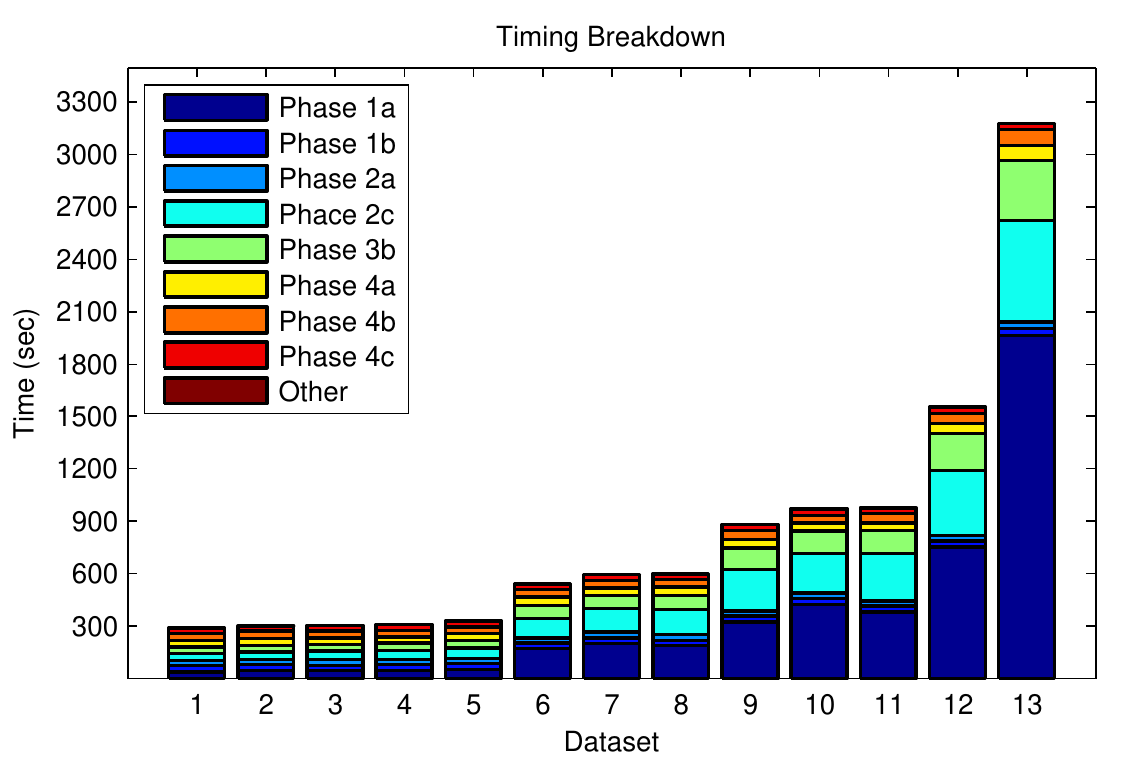}
  \caption{Runtimes broken down by phases}
  \label{fig:timings}
\end{figure}

\begin{figure}[t]
  \centering
  \includegraphics[width=2.5in]{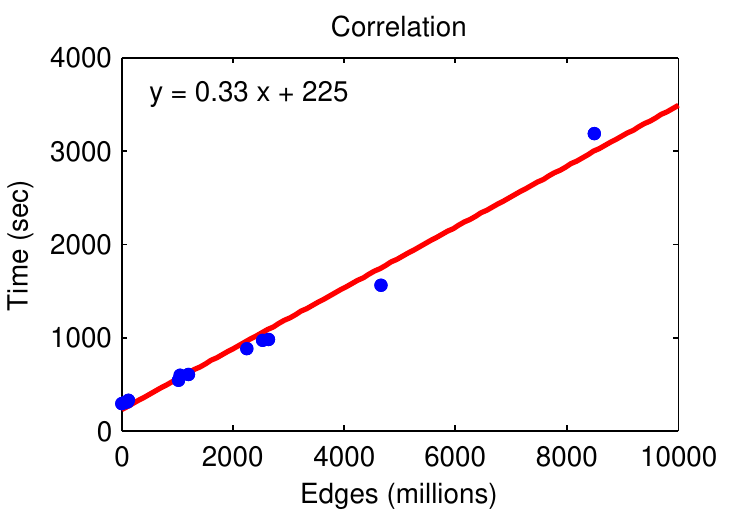}
  \caption{Correlation of the runtime to the number of edges. There is
    a setup cost of \Offset\@ seconds due to the MapReduce
    overhead and an incremental cost of \Slope\@ seconds per million
    edges.}
  \label{fig:time-vs-edges}
\end{figure}

In an enumeration approach, we expect the runtimes to be proportional to
the number of wedges. In that case, the sk-2005 (\#10) would be the
most expensive by an order of magnitude. For our method, however, the
runtime is proportional to the number of edges.
\Fig{time-vs-edges} shows that there is a near-linear relationship
between the number of edges and the runtime. The x-axis is the number
of edges (in millions) and the y-axis is the runtime. All 13 examples
are included. We see that there is a constant cost of \Offset\@ seconds,
which accounts for the MapReduce overhead and then an incremental cost
of \Slope\@ seconds per million edges.

\paragraph{Comparisons to Other Methods}

As one point of comparison, we ran Plantenga's implementation that fully
enumerates triangles \cite{Pl12} on several smaller graphs. The results
are summarized in \Tab{sgi}. 
\begin{table}[bht]\footnotesize
  \centering
   \caption{Comparison of sampling and exact enumeration.}
    \label{tab:sgi}
  \begin{tabular}{|c|l|r|r|r|r|r|}
    \hline
  \MC{1}{|c|}{\bf Id} & \MC{1}{c|}{\bf Graph Name} & \MC{1}{c|}{\bf Wedges} & 
  \MC{1}{c|}{\bf Run Time} & \MC{1}{c|}{\bf Exact} &  
  \MC{1}{c|}{\bf Sampling}  &  \MC{1}{c|}{\bf Sampling} \\ 
  \MC{1}{|c|}{}& \MC{1}{c|}{} &  \MC{1}{c|}{\bf Checked} &  \MC{1}{c|}{\bf (sec.)} &
  \MC{1}{c|}{\bf GCC} &  \MC{1}{c|}{\bf Error} &    \MC{1}{c|}{\bf
    Speed-up} \\ \hline 
   1 &     amazon-2008 &    25\% &      158 &    0.2603 & 0.0001 &     1 \\ \hline
   2 &   ljournal-2008 &    19\% &    3,385 &    0.1228 & 0.0010 &    11 \\ \hline
   3 &  hollywood-2009 &    29\% &   21,665 &    0.3090 & 0.0006 &    72 \\ \hline
   4 &  hollywood-2011 &    27\% &   90,598 &    0.1761 & 0.0006 &   293 \\ \hline
  \end{tabular}
\end{table}
This implementation is efficient because it only checks wedges
where the center degree is smallest, i.e., 
$(u,v,w)$ such that $d_v \leq d_u$ and $d_v \leq d_w$. This means that
only one wedge per triangle is checked and moreover that many open
wedges centered at high-degree nodes are ignored. The table lists the
percentage of wedges that are checked for closure.
The run times range from 3 minutes for the smallest graph, up to 25
hours for hollywood-2011 (2M nodes, 114M edges). Because this code
enumerates every triangle, we can calculate the exact global
clustering coefficient (GCC). The error from our sampling method is
also reported. At a confidence
level of 99.9\%, using $k=2000$ samples yields an error of
$\eps=0.05$.  The true errors are one to two orders of magnitude less
than this worse case probabilistic bound. Finally, we observe the main
advantage of sampling in terms of the observed speed-up, up to 293X
for hollywood-2011. Larger graphs cannot be completed in a reasonable
amount of time on our cluster. 

Suri and
Vassilvitskii \cite{SuVa11} have an efficient enumeration method that was
able to process two of the same data sets on a 1636-node
Hadoop cluster: (a) ljournal-2008 required 5.33 minutes, and 
(b) twitter-2010 took 483 minutes. 
(In subsequent work, Arifuzzaman et al.~\cite{ArKhMa12,ArKhMa12a} have 
reduced the run time \cite{PaCh13} were able to reduce the MapReduce run time to 213 minutes on 47 nodes.)
On our 32-node cluster, our method requires approximately the same running time for the 
smaller data set (since most of the work is overhead) but required only 10 minutes
for twitter-2010.

We also compare to in-memory methods. We used an SGI Altix UV 10 System with 4 Xeon 8-core nodes 
and 64 $\times$ 8GB memory, for a total of 32 cores and 512GB of memory. 
We consider two methods for exact node-level triangle counting on the twitter-2010 graph: 
(1) MTGL~\cite{BaBeMuWh09} using the algorithm described in \cite{Co09} and 
(2) GraphLab/PowerGraph \cite{LoGoKyBi10,GoLoGuBi12}.
On 32 cores, MTGL takes approximately 
3.5 hours and consumes 125GB of memory. 
On 16 cores, GraphLab/PowerGraph takes approximately 17 minutes to run, 
and 430GB of memory at its peak usage, nearly maxing out the 
machine's total memory. 
In both cases, just loading the graph takes 8--15 minutes.
Compare these times to a total of 10 minutes using Hadoop and our inexact 
triangle-counting method. 

The other class of methods for comparison are edge-sampling
methods such as Doulion~\cite{TsKaMiFa09}. The basic idea is to sample
a subset of \emph{edges} and then run a triangle counting method (such
as enumeration) on the reduced graph. Edge sampling has been compared
to wedge sampling in serial \cite{SePiKo13}. Keeping 1 in 25 edges
produces results that are roughly comparable to wedge sampling in time
but with much greater variance in the GCC estimate. Keeping fewer
edges yields savings in time but at the expense of much greater
variance. Hence, we have not compared parallel implementations.

\begin{figure}[t]
  \centering
  \includegraphics[width=3in]{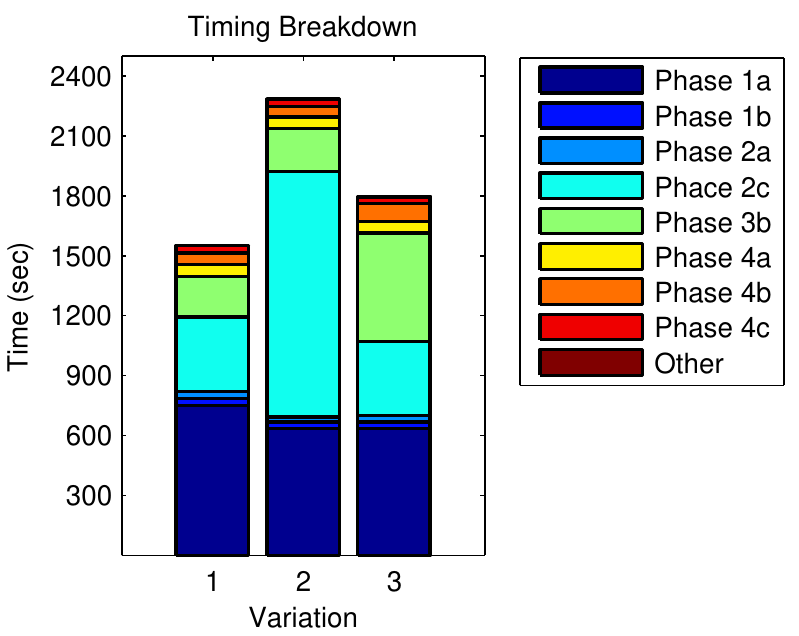}
  \caption{Timings for variations. The variations are (1) Original algorithm, (2) Skip Phase 2b, (3) Skip Phase 3a.}
  \label{fig:timings-variations}
\end{figure}

\paragraph{Impact of Implementation Features}
We have considered many alternatives during the implementation of the
wedge sampling algorithm, and in this subsection we present the impact of
two implementation features.  The three versions of the code we
compare are:
\begin{inparaenum}[(1)]
  \item Original algorithm.
  \item Skip Phase 2b.
  \item Skip Phase 3a.
\end{inparaenum}
We show results for uk-union in \Fig{timings-variations}.  
Skipping Phase
2b means that every edge generates two messages in Phase 2c,
increasing the time in that phase from 372 seconds to 1235 seconds (3X
increase), twice as expensive as Phase 1a.
Skipping Phase 3a means that every edge generates one message in
Phase 3b, increasing the time in that phase from 207 seconds to 543
seconds (2.5X increase).
Hence, taking measures to reduce the data that must be shuffled to the
reducers has major pay-offs in terms of performance.

\begin{figure}[tpb]
  \centering
  \subfloat{\includegraphics[width=.3\textwidth]{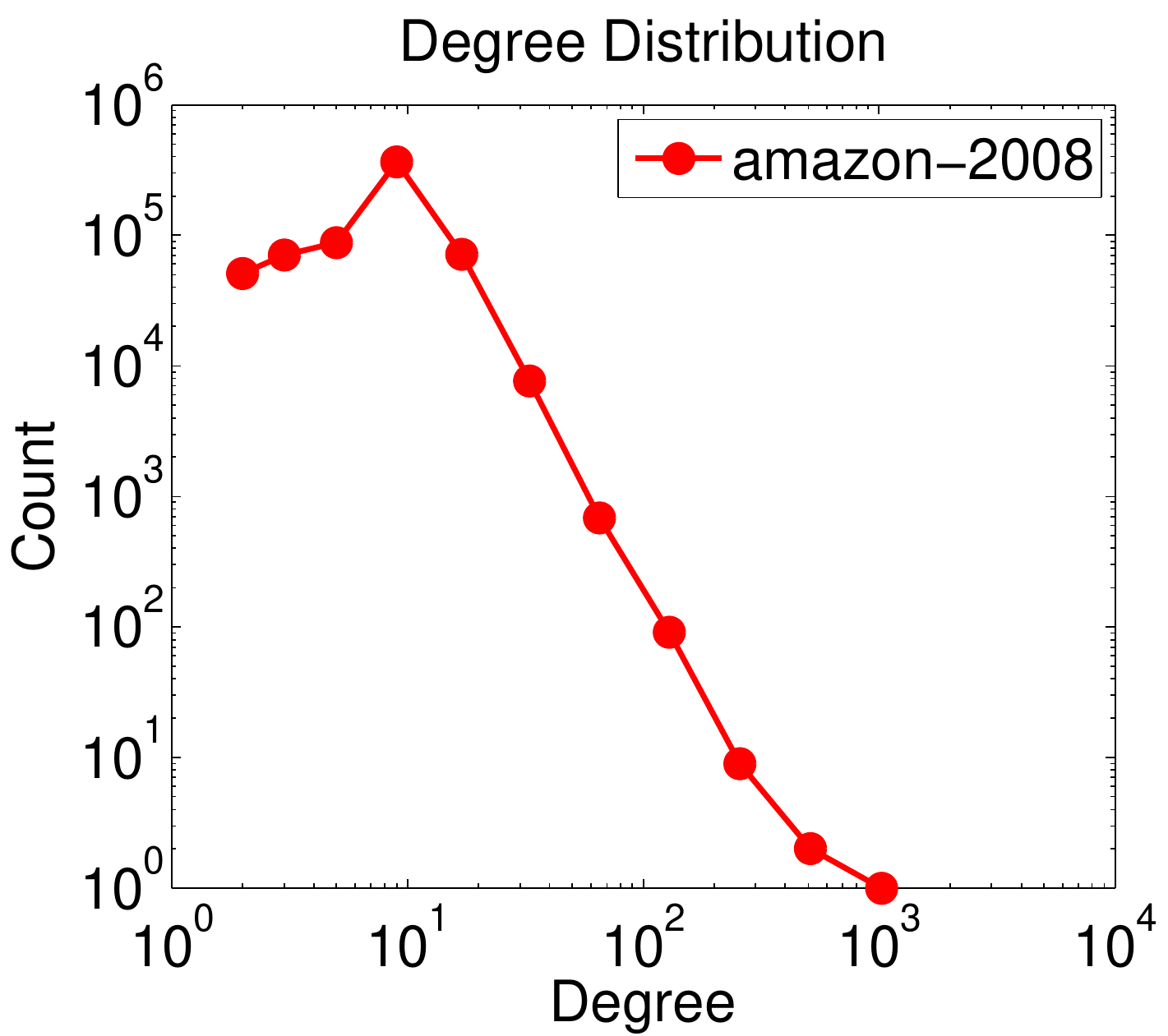}}
  \subfloat{\includegraphics[width=.3\textwidth]{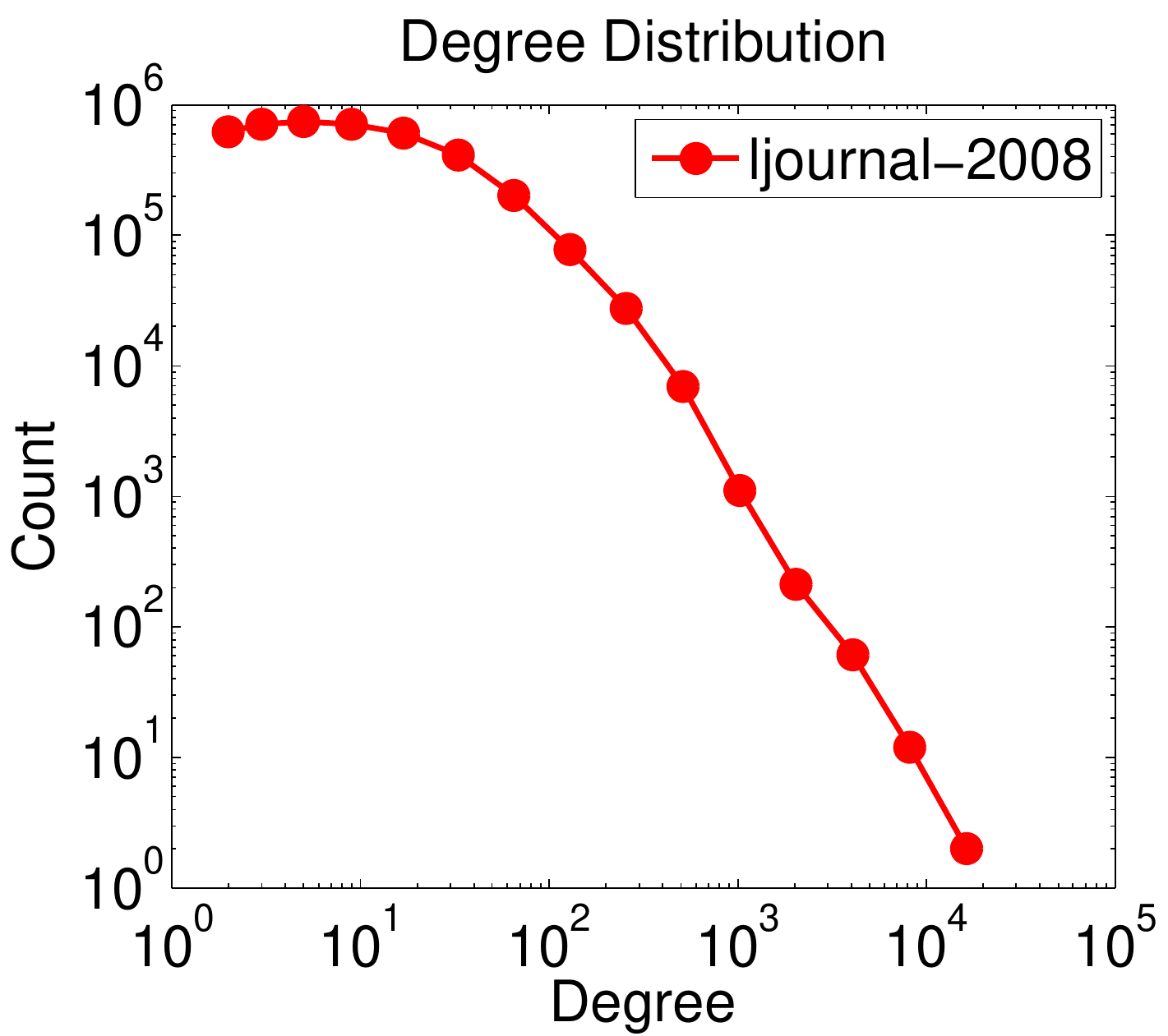}}
  \subfloat{\includegraphics[width=.3\textwidth]{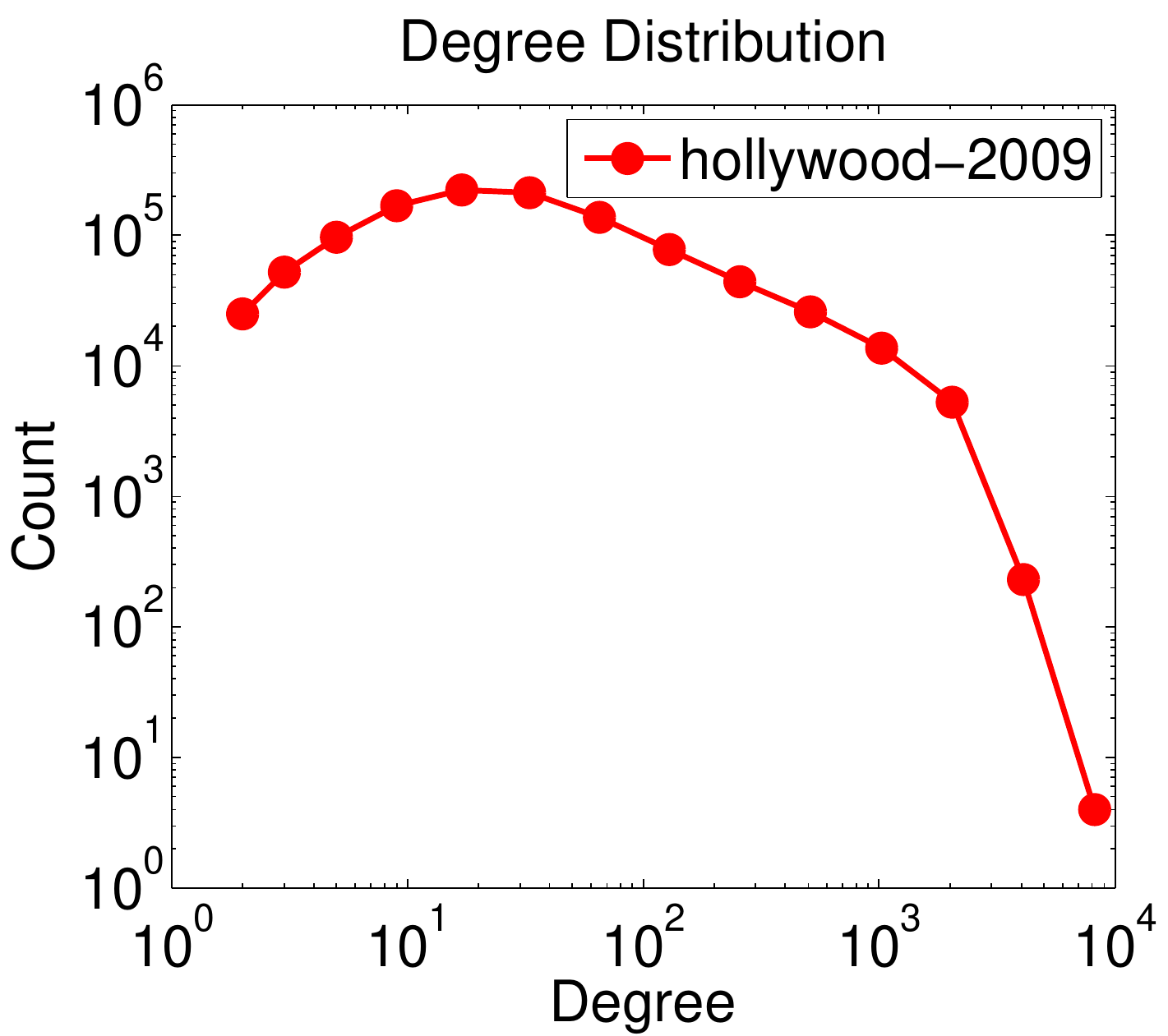}} \\
  \subfloat{\includegraphics[width=.3\textwidth]{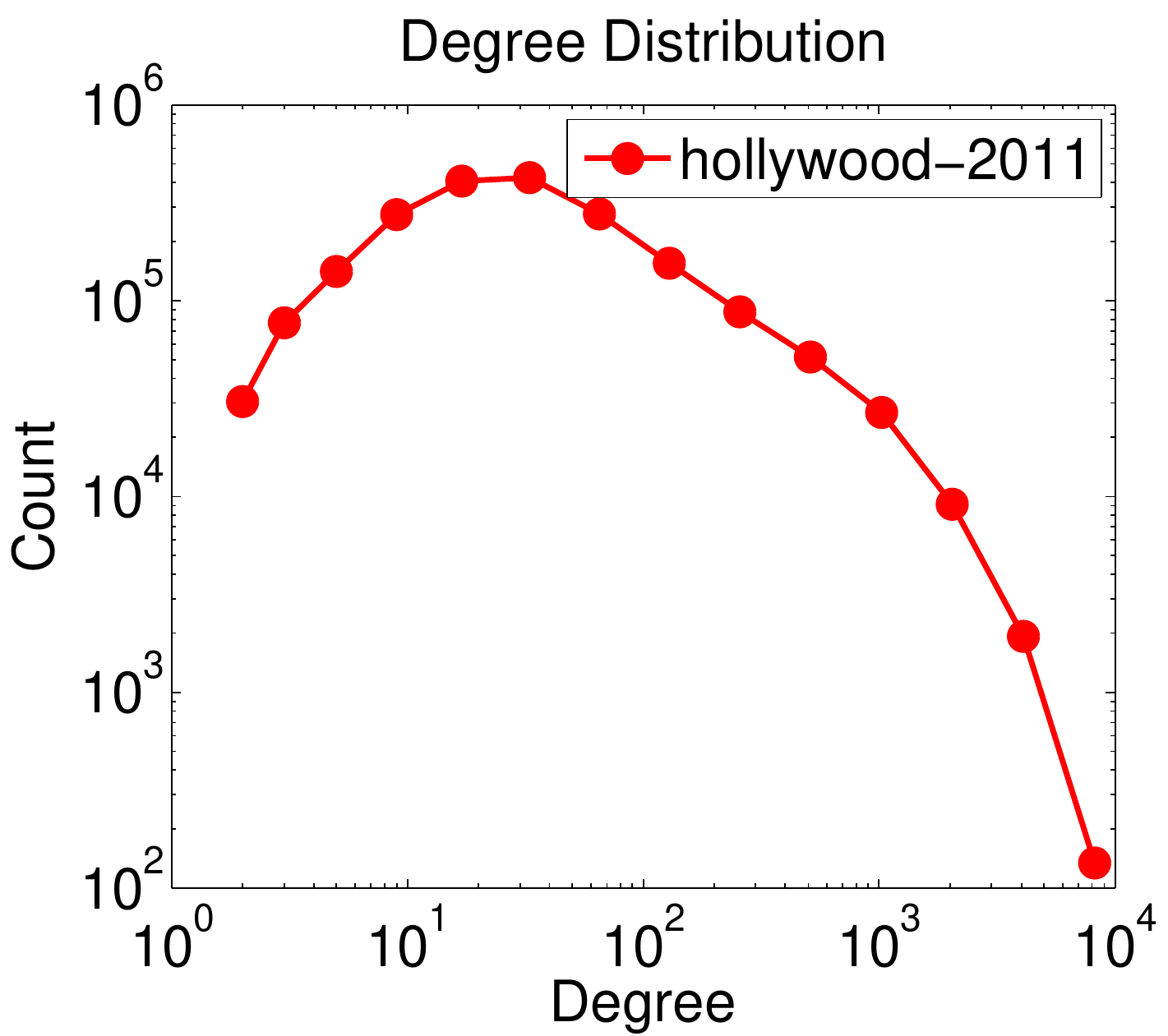}}
  \subfloat{\includegraphics[width=.3\textwidth]{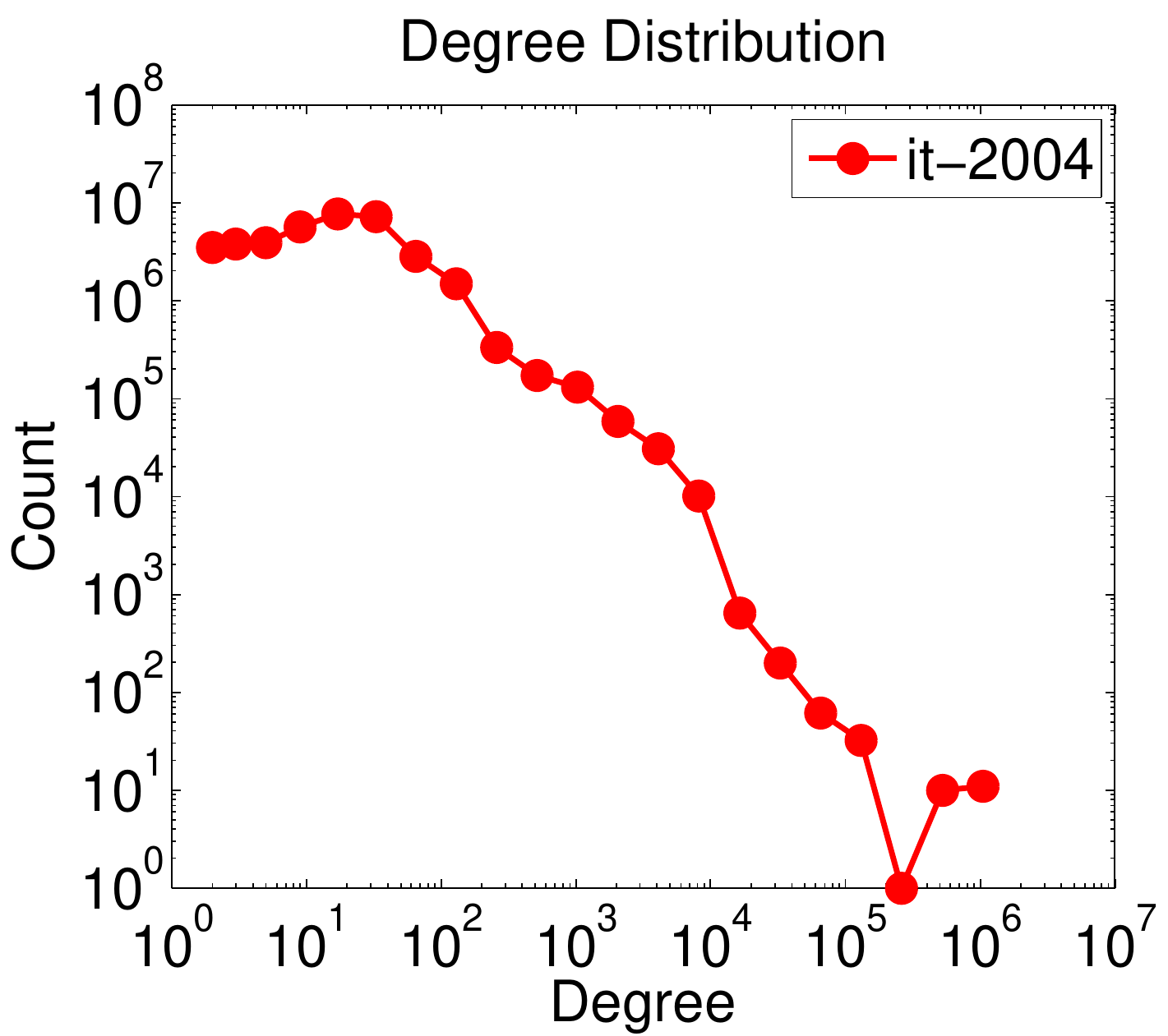}}
  \subfloat{\includegraphics[width=.3\textwidth]{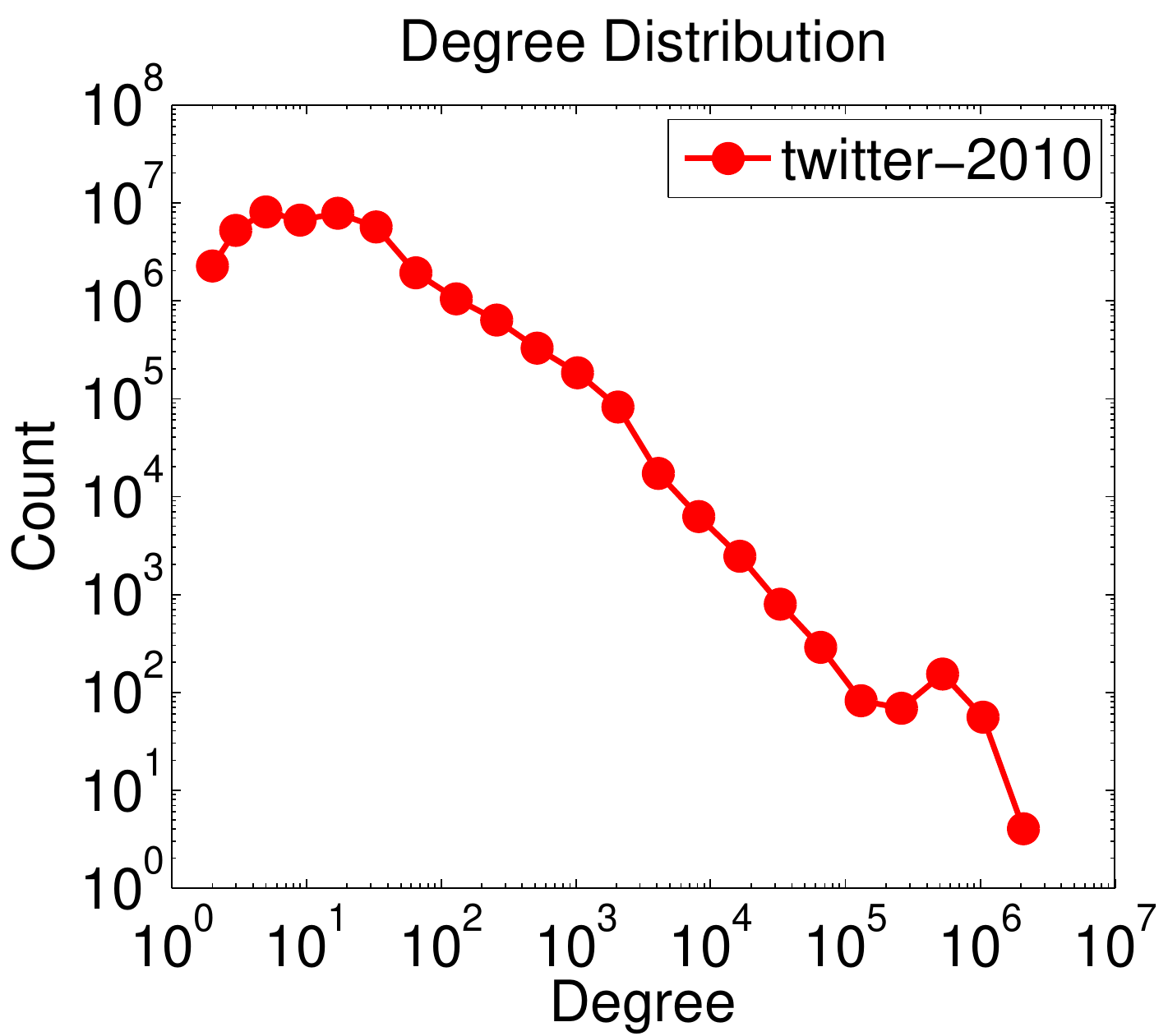}}\\
  \subfloat{\includegraphics[width=.3\textwidth]{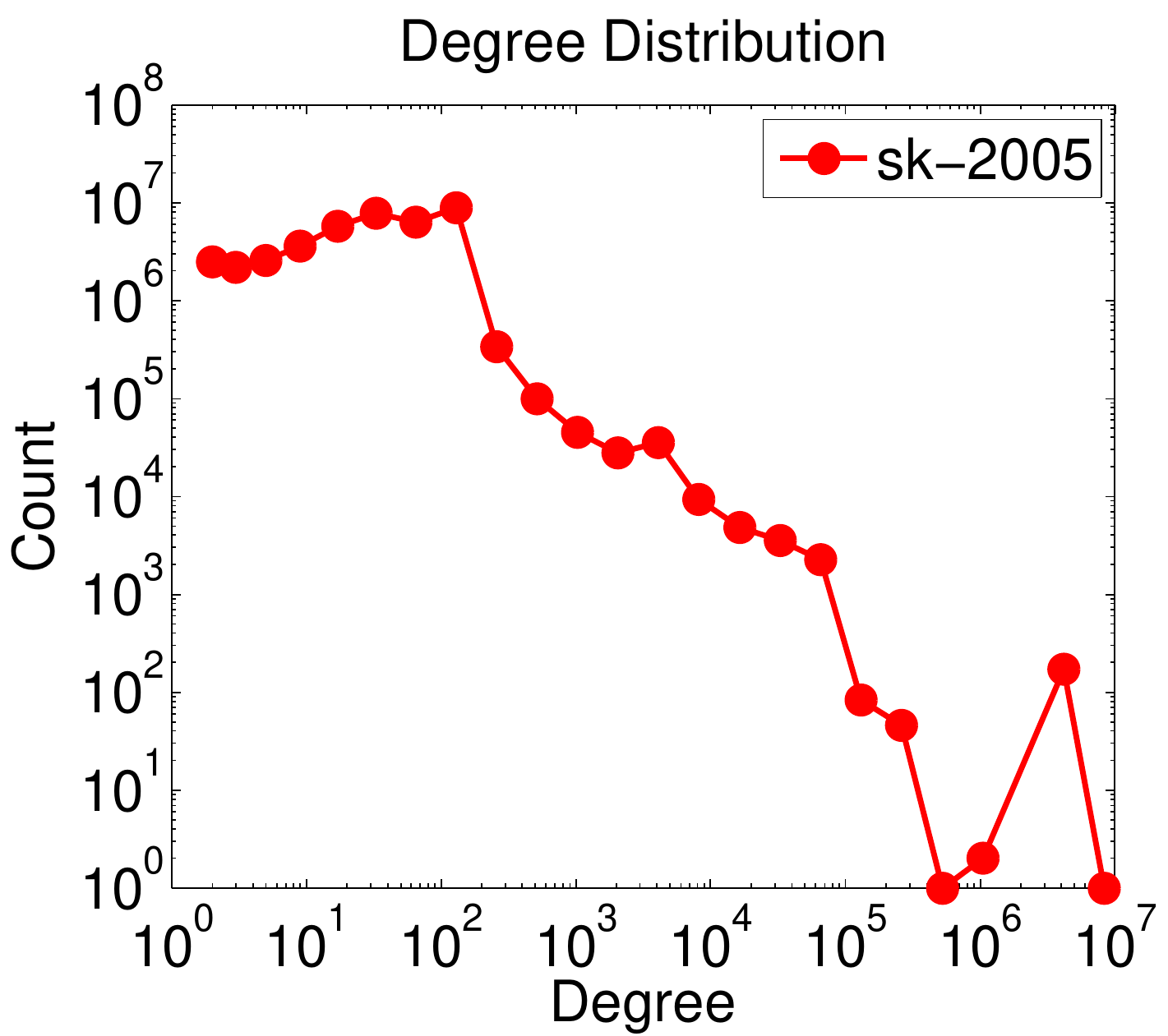}}
  \subfloat{\includegraphics[width=.3\textwidth]{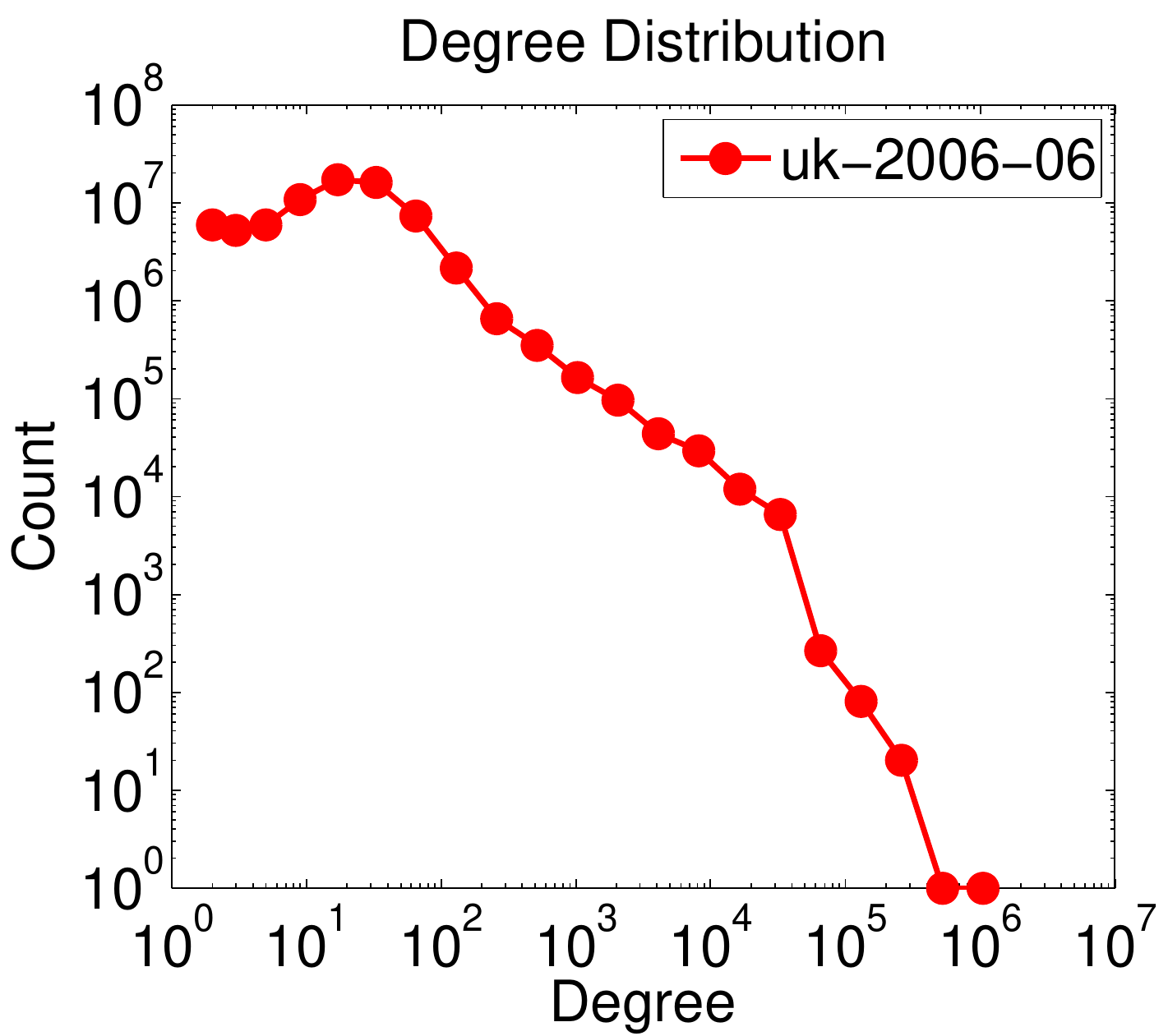}}
  \subfloat{\includegraphics[width=.3\textwidth]{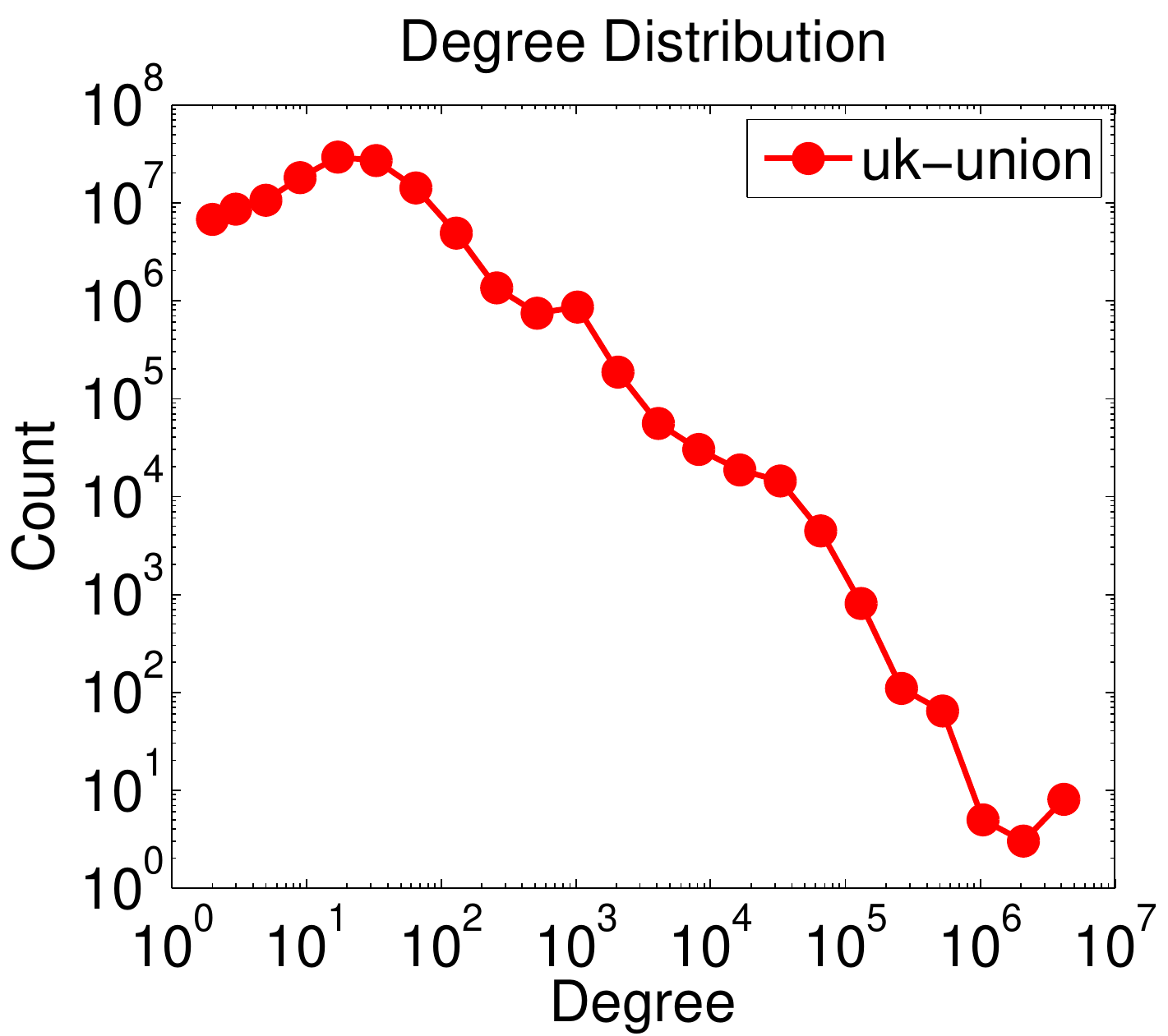}}\\
  \subfloat{\includegraphics[width=.3\textwidth]{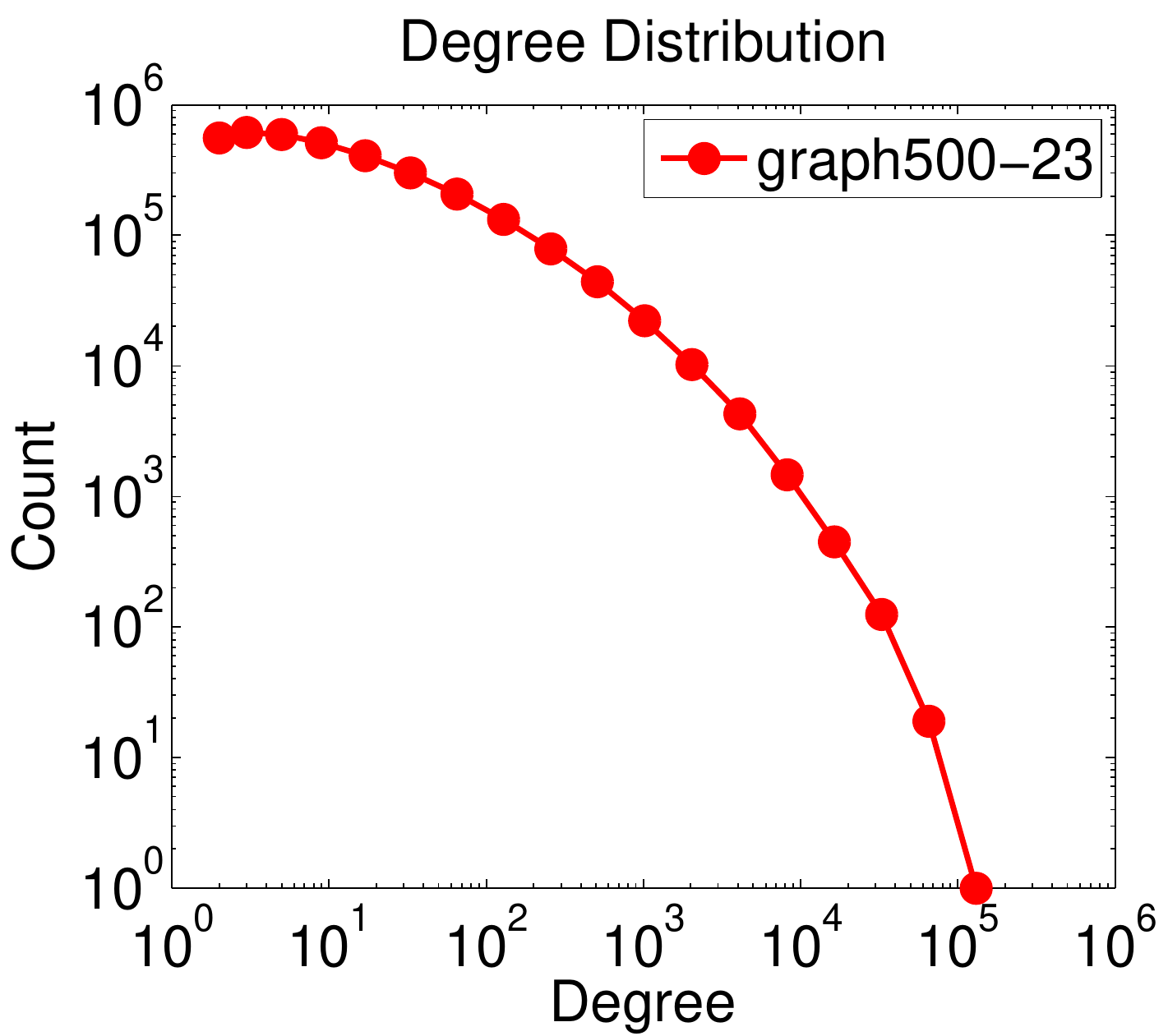}}
  \subfloat{\includegraphics[width=.3\textwidth]{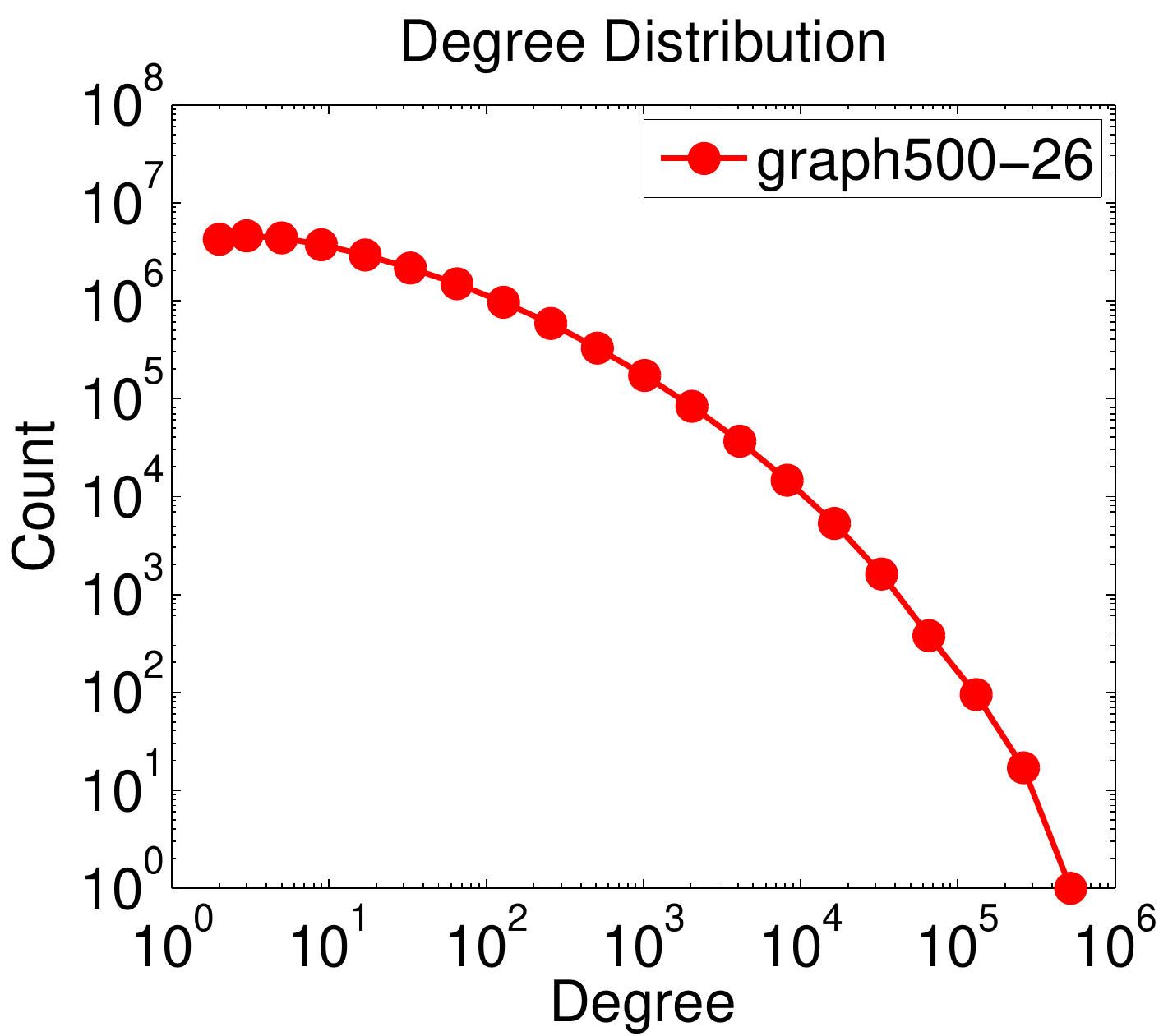}}
  \subfloat{\includegraphics[width=.3\textwidth]{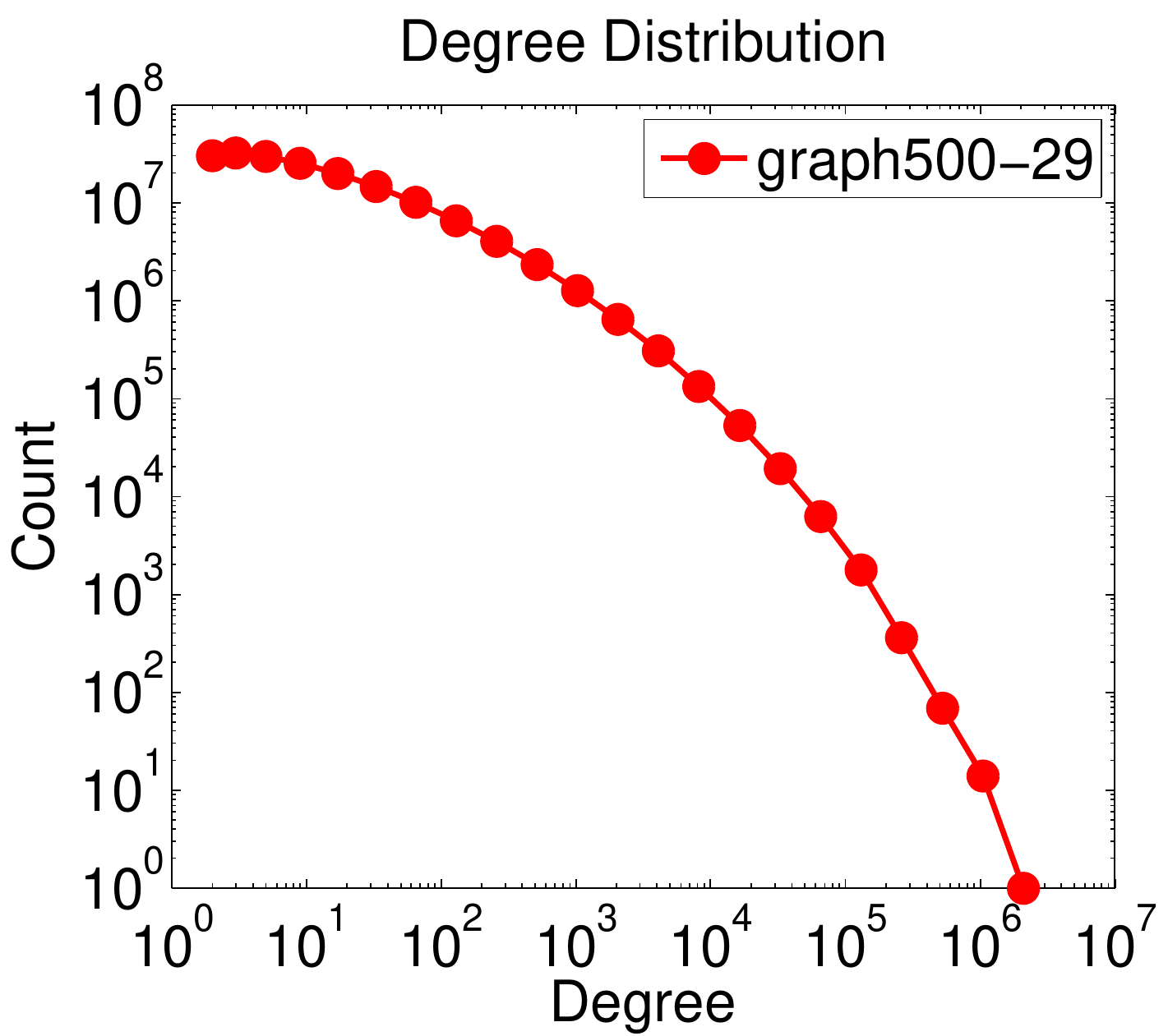}}
  \caption{Degree distributions by bin}
  \label{fig:dd}
\end{figure}

\subsection{Degree Distribution}

The output of Phase 1b yields the (binned) degree distribution. We
show results for 12 networks in \Fig{dd} (we omit uk-2006-05
because it is similar to uk-2006-06). For each data point, the
x-coordinate is the minimum degree of the bin, and the y-coordinate is
the total number of vertices in that bin.

The degree distributions can be roughly
characterized as heavy-tailed. None of the real-world graphs are
particularly smooth in the degree distribution, and some have odd
spikes, especially in the tails (e.g., sk-2005, it-2004, uk-union).
The artificial graphs (graph500-23/26/29) are extremely smooth;
we know from analysis that the noisy version of Graph500 yields lognormal tails
\cite{SePiKo11,SePiKo13}.

\subsection{Clustering Coefficients}

Clustering coefficients for each bin are displayed in \Fig{cc}. Here the x-coordinate is
the minimum degree in the bin, and the y-coordinate is the average
clustering coefficient for wedges with centers in that bin. 

Social
networks are well-known to have not only high clustering coefficients,
but also clustering coefficients that tend to degrade as the degree
increases. This can be seen in the following graphs: amazon-2008,
ljournal-2008, hollywood-2009, hollywood-2011. The twitter-2010 graph
is not as ``social'' in terms of clustering coefficient.

The web graphs (it-2004, uk-2006-06, uk-union, sk-2005) 
are interesting because the clustering coefficients seems to start low,
increase, and then drop off quickly. This may be due to the
design of web sites with many interconnected pages or to some artifact
of the crawling process.

The Graph500 examples have overall low clustering coefficients and do
not behave like the real-world graphs. The closest match is
the twitter-2010 graph.

\begin{figure}[t]
  \centering
  \subfloat{\includegraphics[width=.3\textwidth]{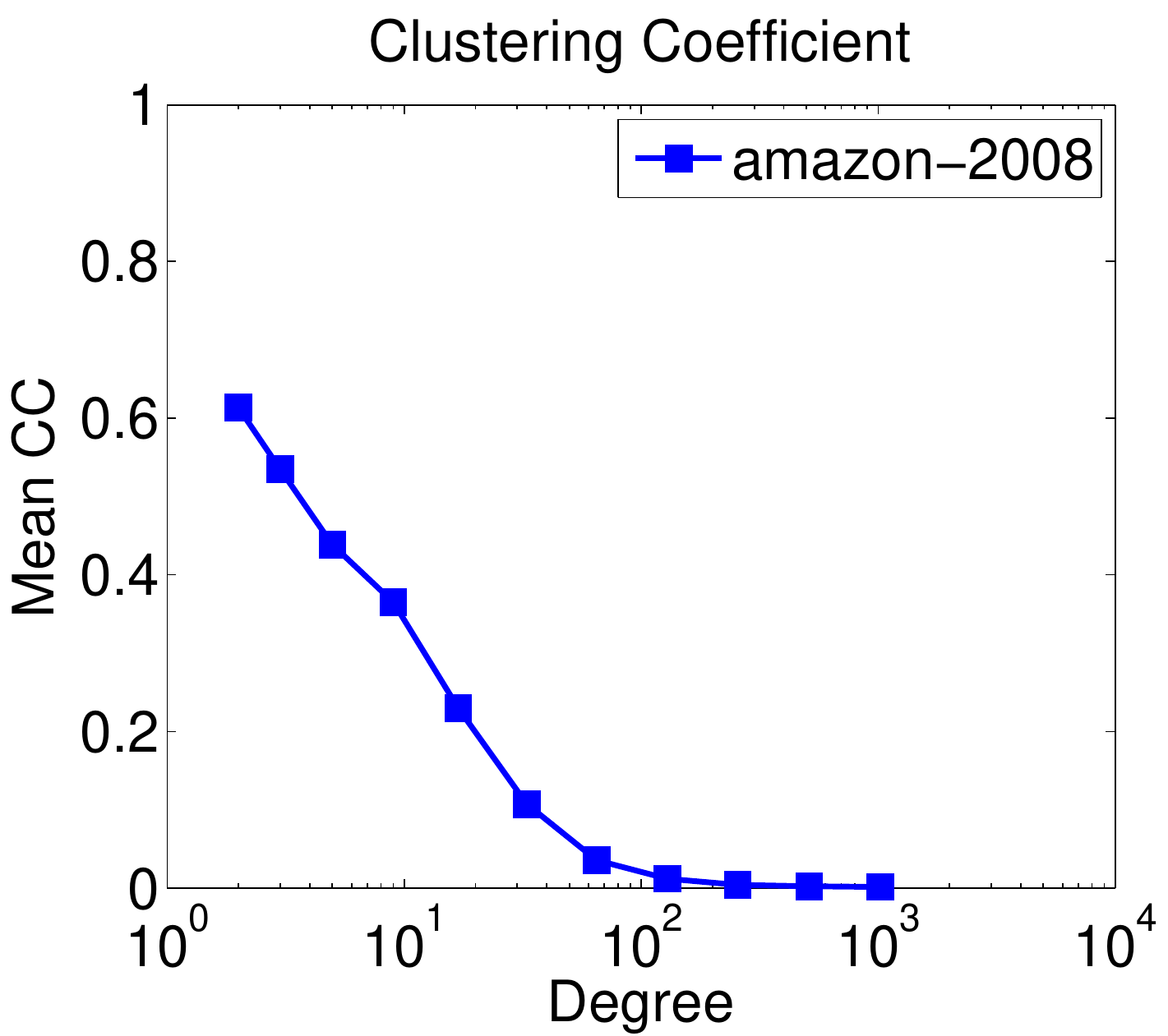}}
  \subfloat{\includegraphics[width=.3\textwidth]{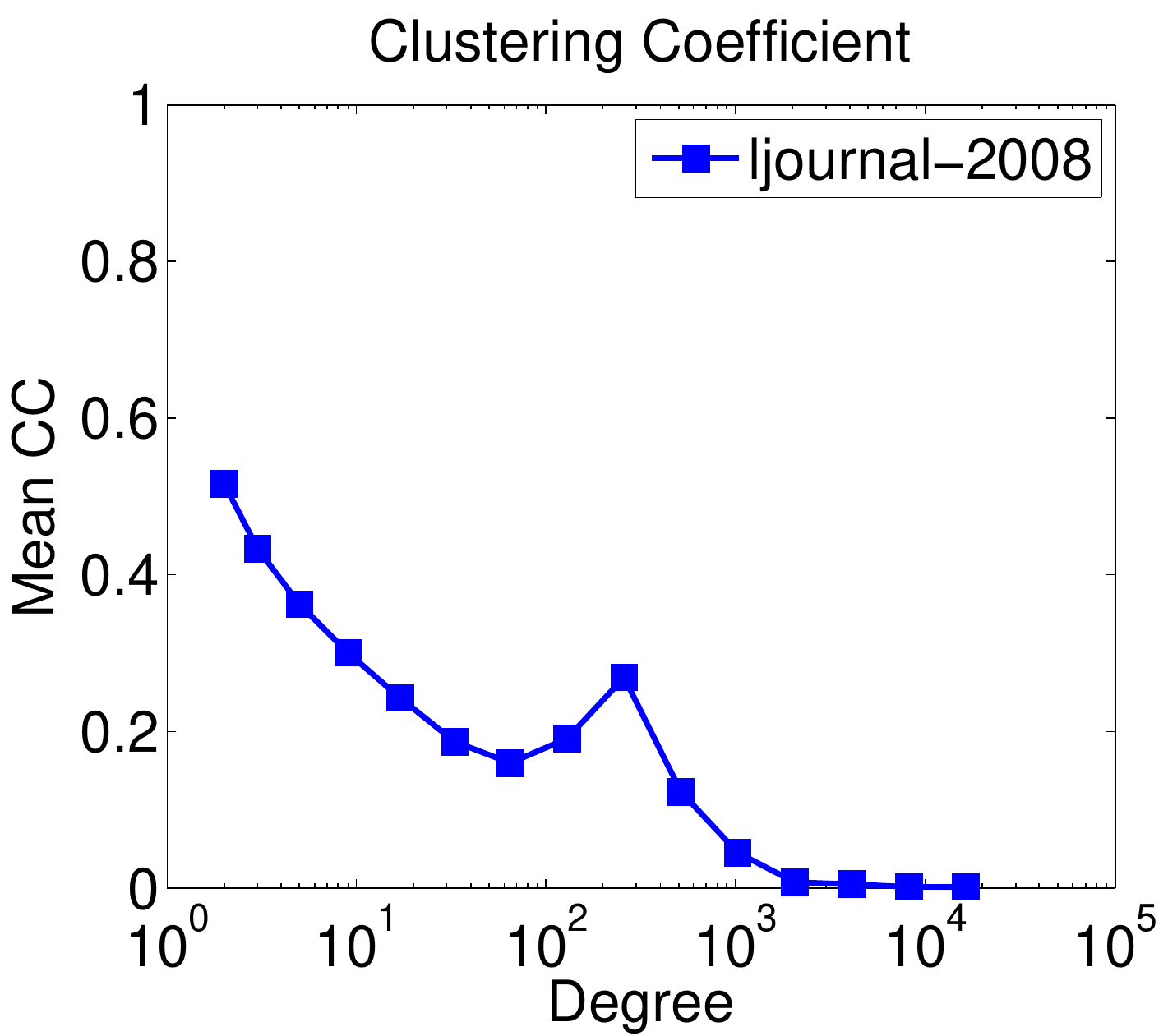}}
  \subfloat{\includegraphics[width=.3\textwidth]{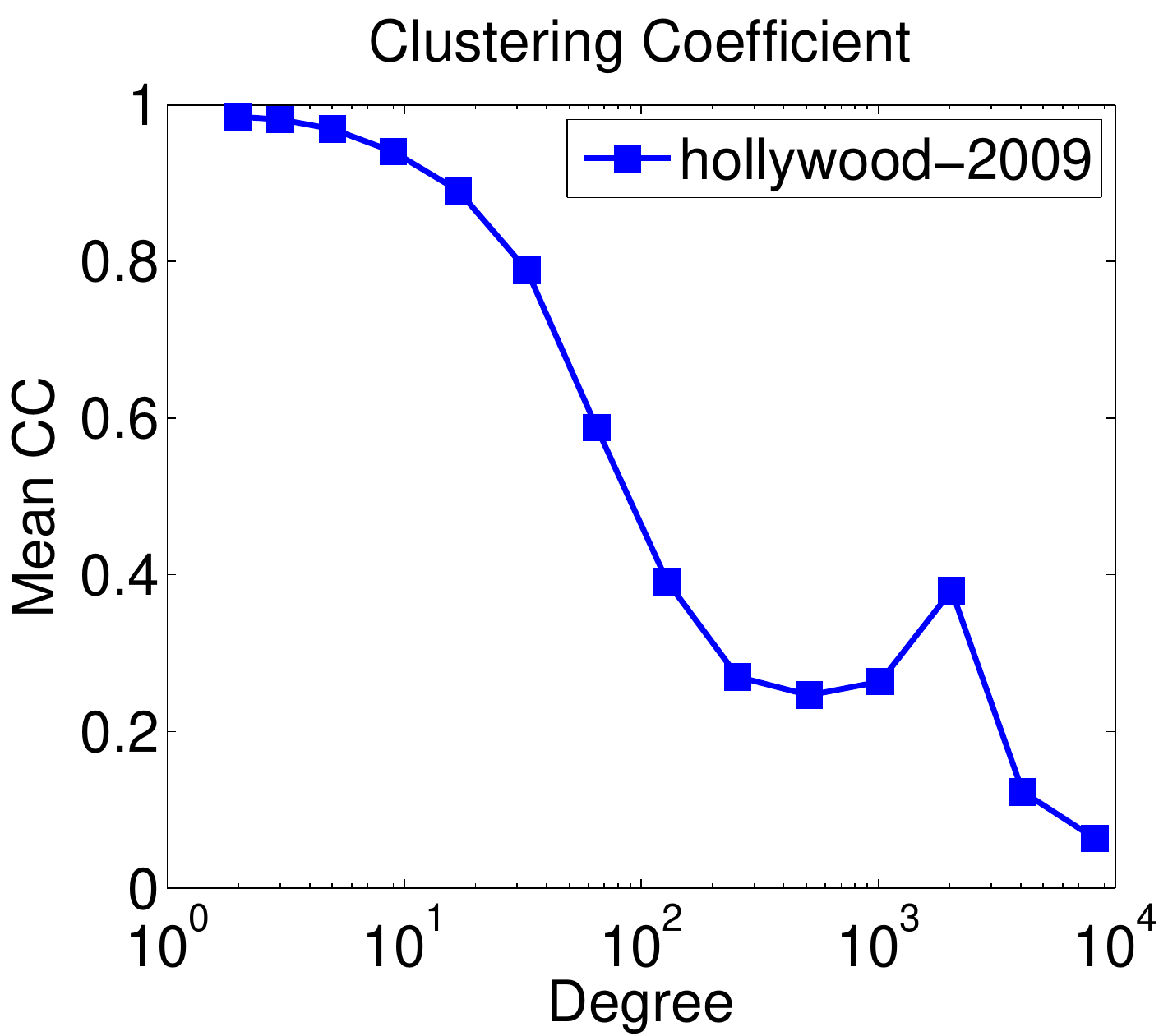}} \\
  \subfloat{\includegraphics[width=.3\textwidth]{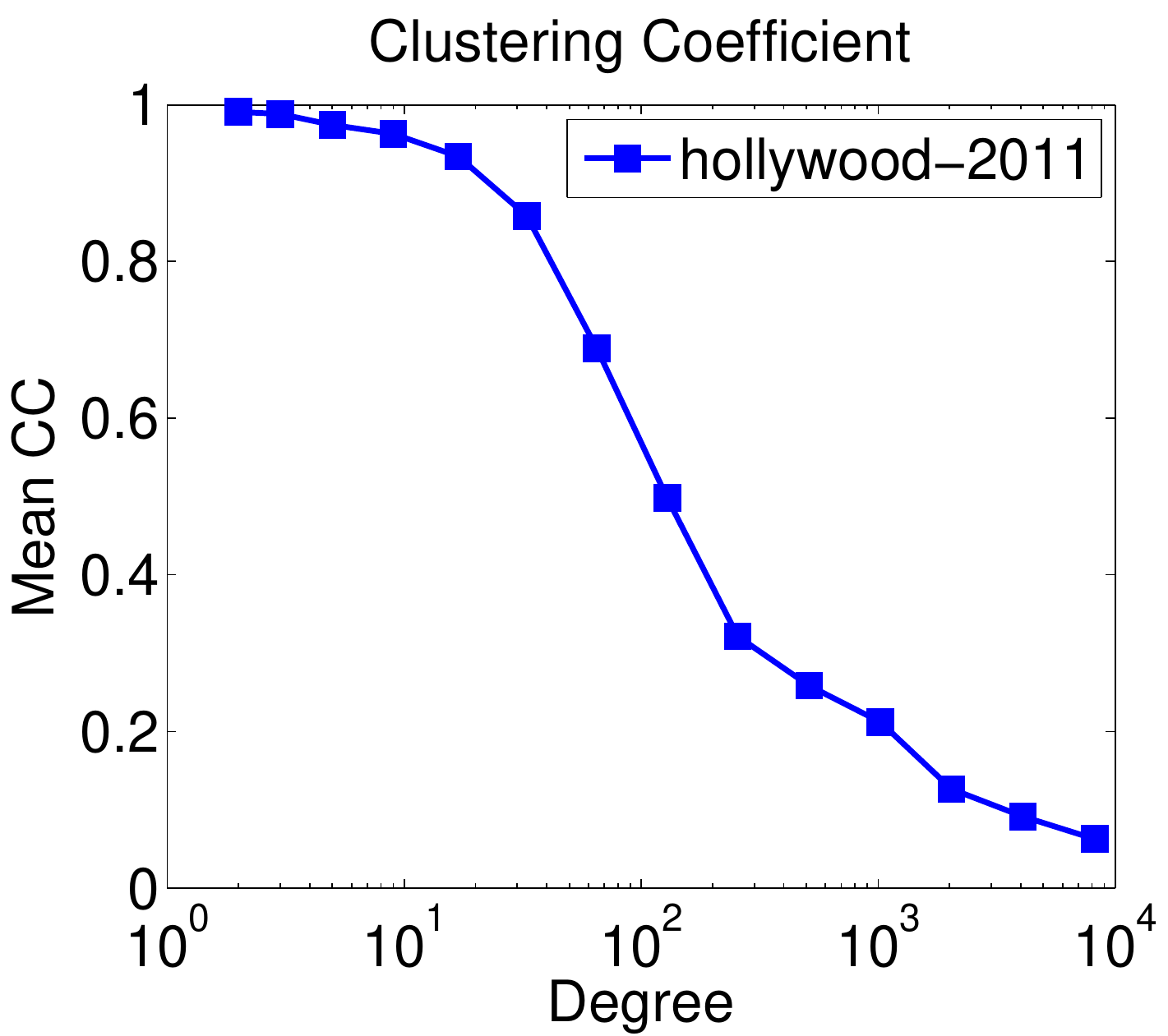}} 
  \subfloat{\includegraphics[width=.3\textwidth]{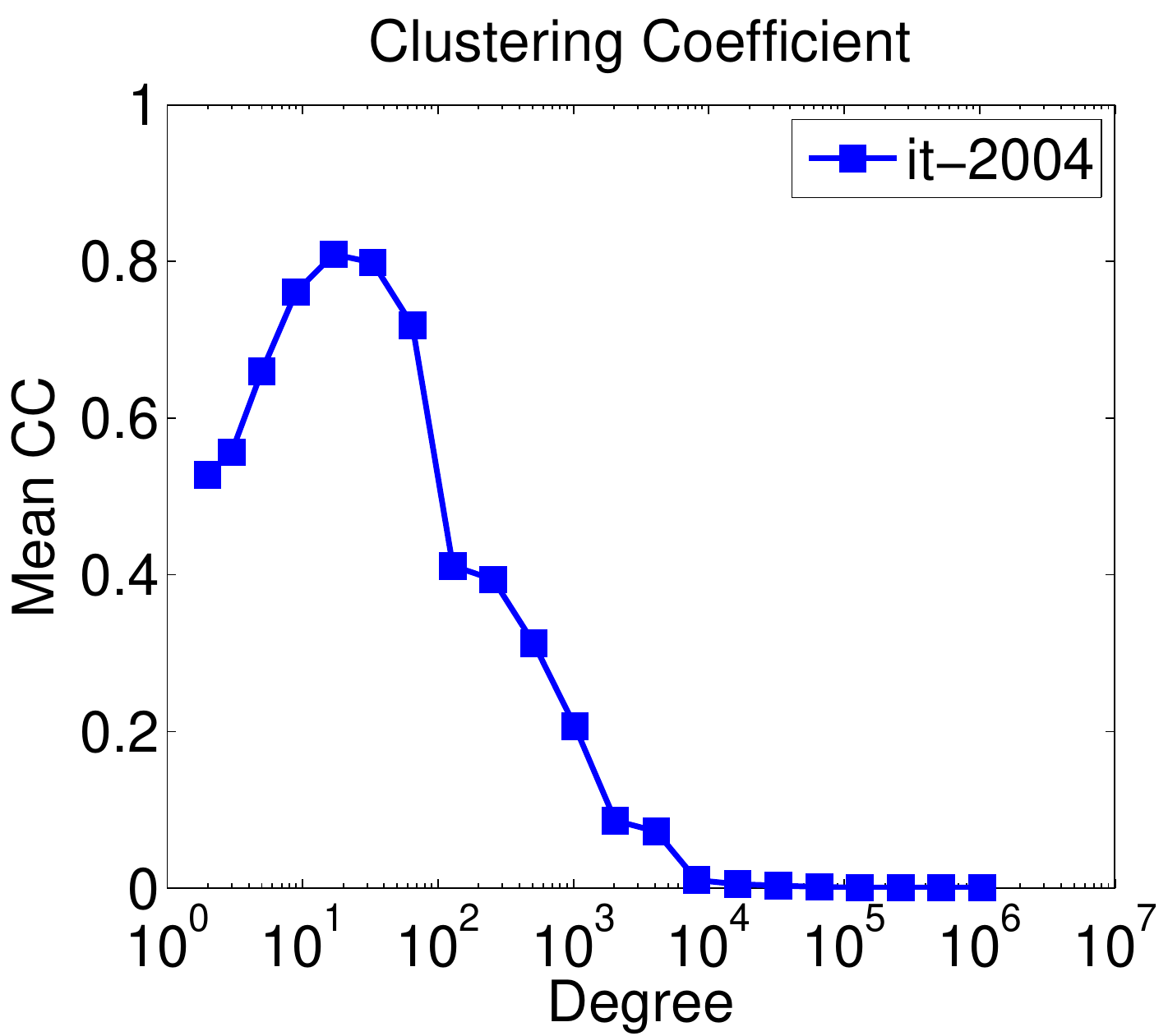}}
  \subfloat{\includegraphics[width=.3\textwidth]{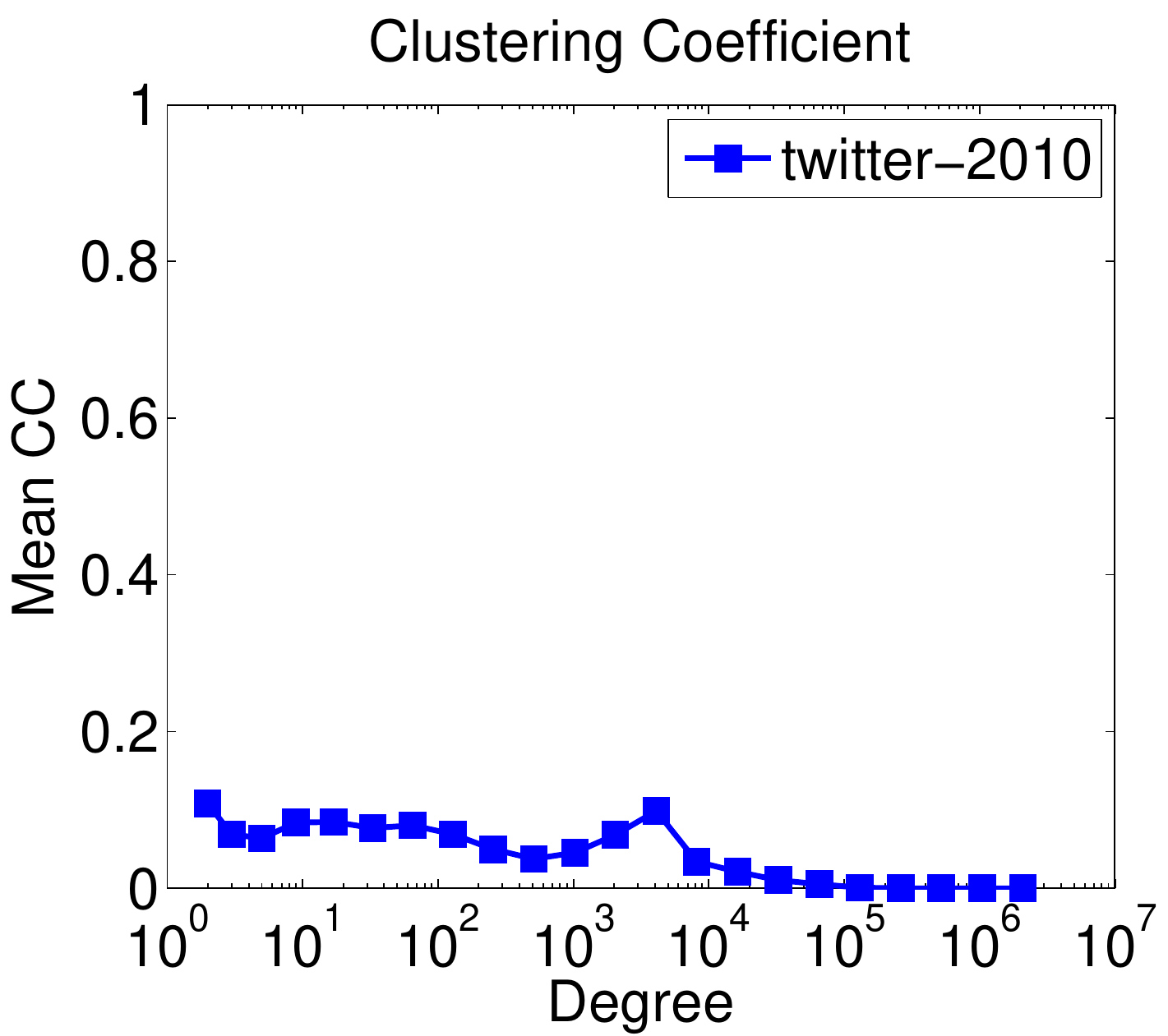}}\\
  \subfloat{\includegraphics[width=.3\textwidth]{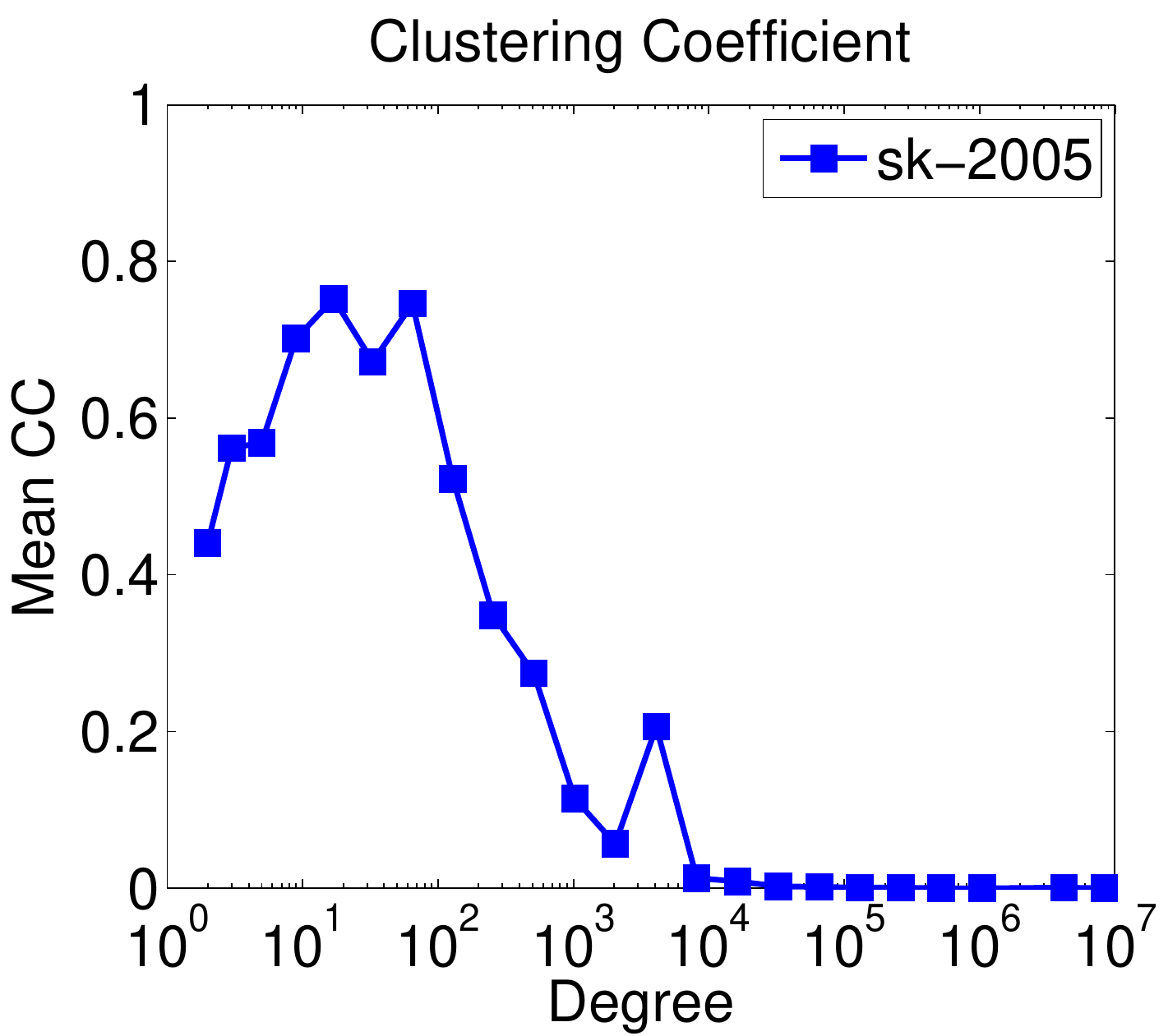}}
  \subfloat{\includegraphics[width=.3\textwidth]{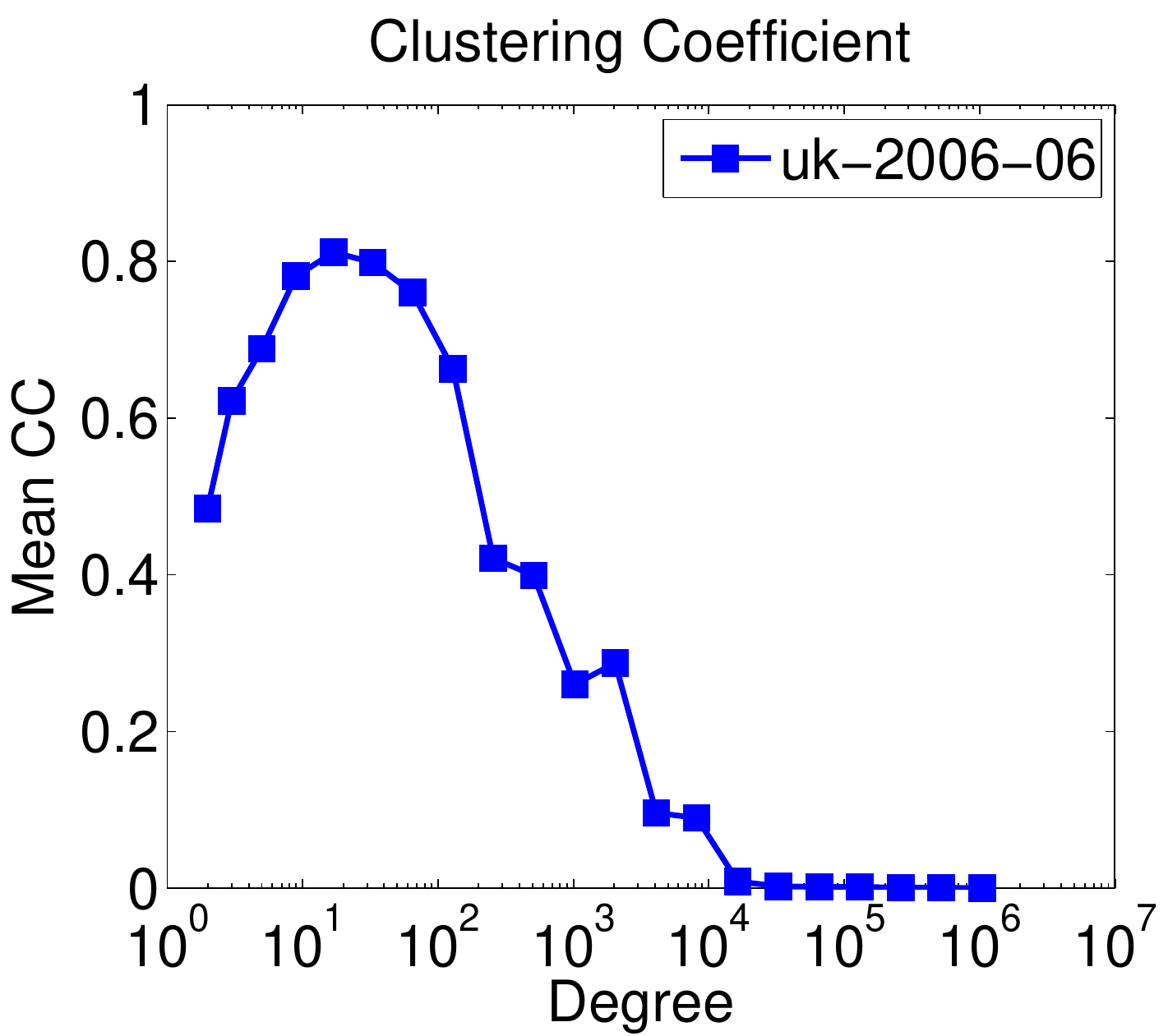}}
  \subfloat{\includegraphics[width=.3\textwidth]{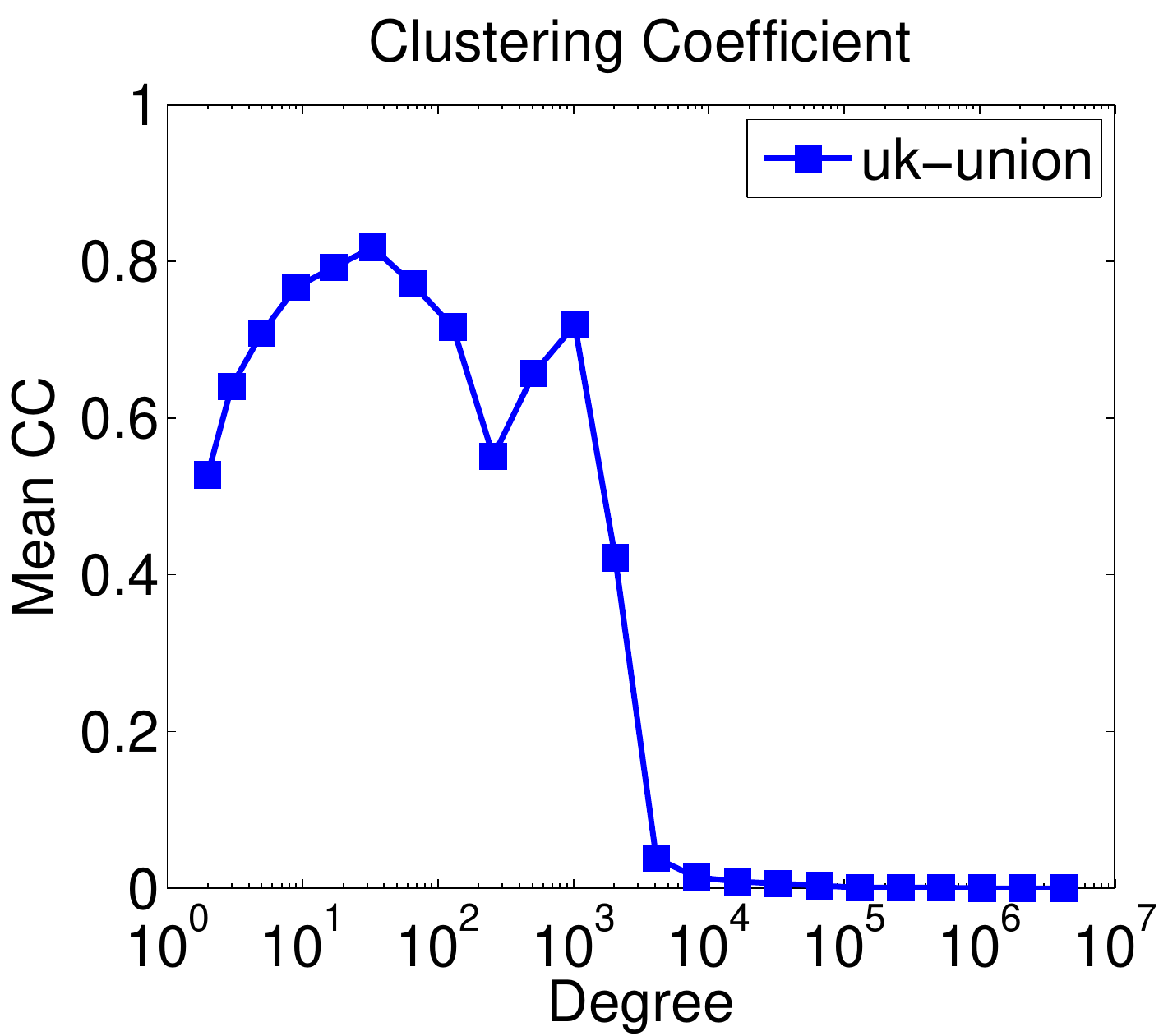}}\\
  \subfloat{\includegraphics[width=.3\textwidth]{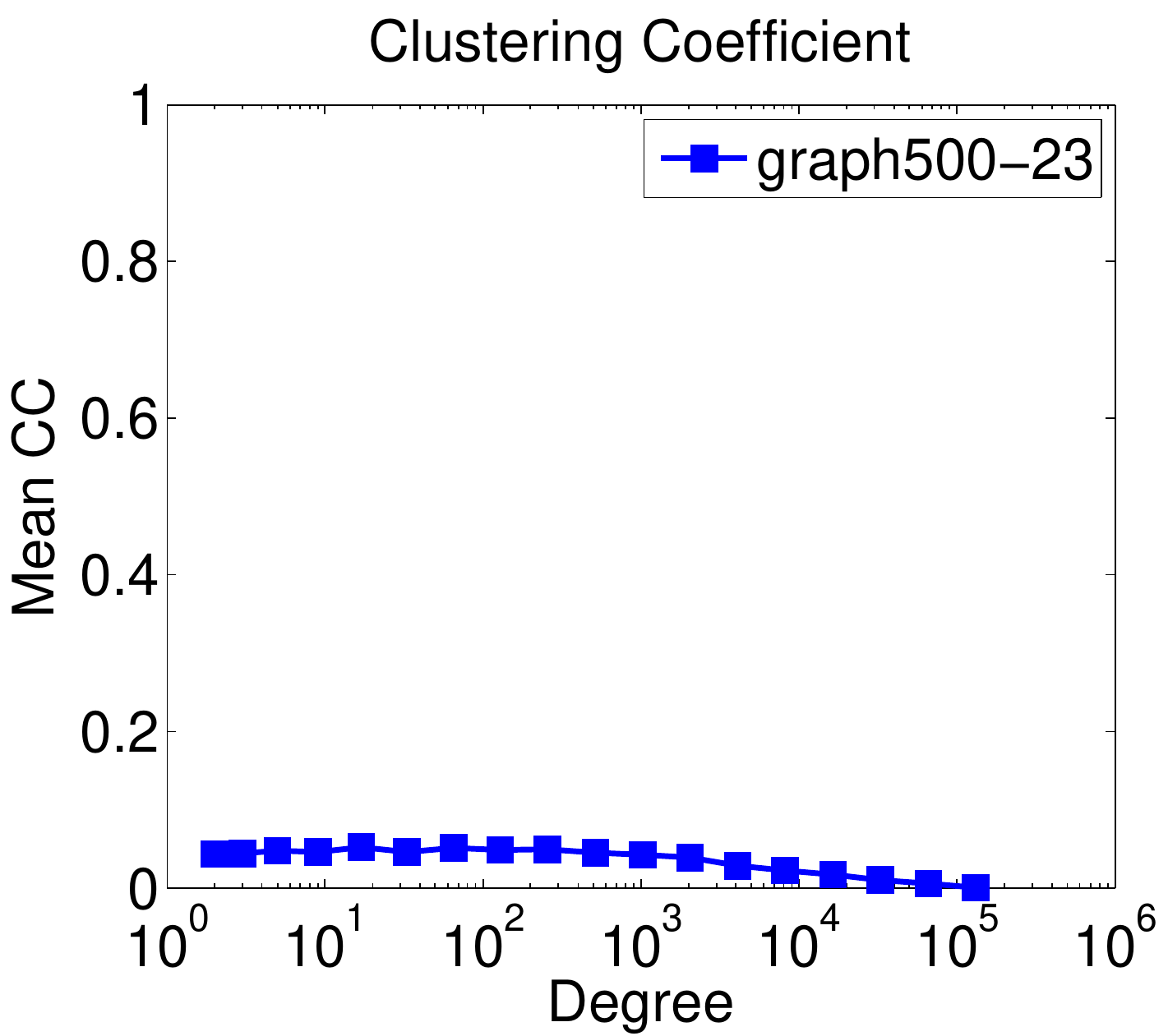}}
  \subfloat{\includegraphics[width=.3\textwidth]{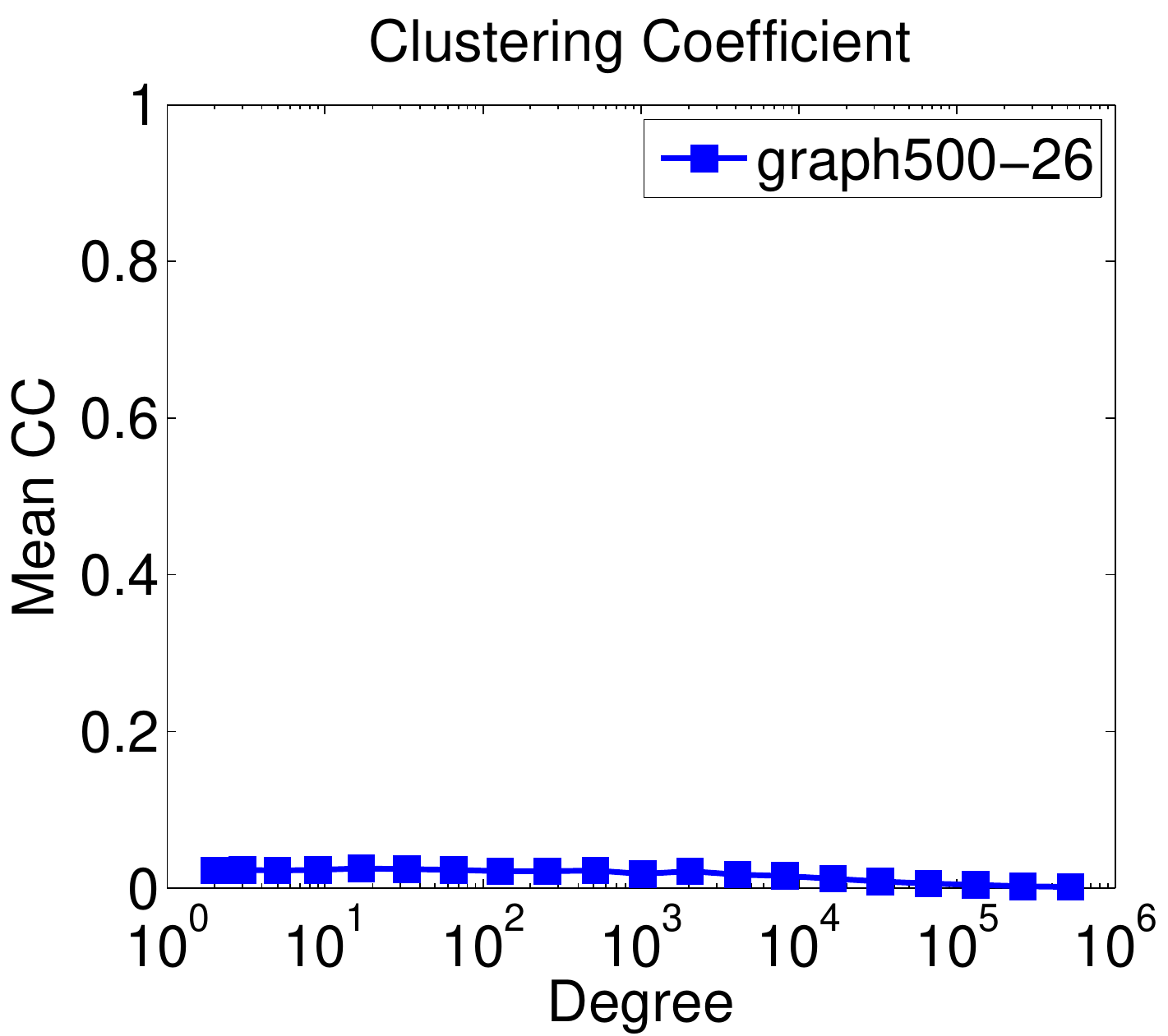}}
  \subfloat{\includegraphics[width=.3\textwidth]{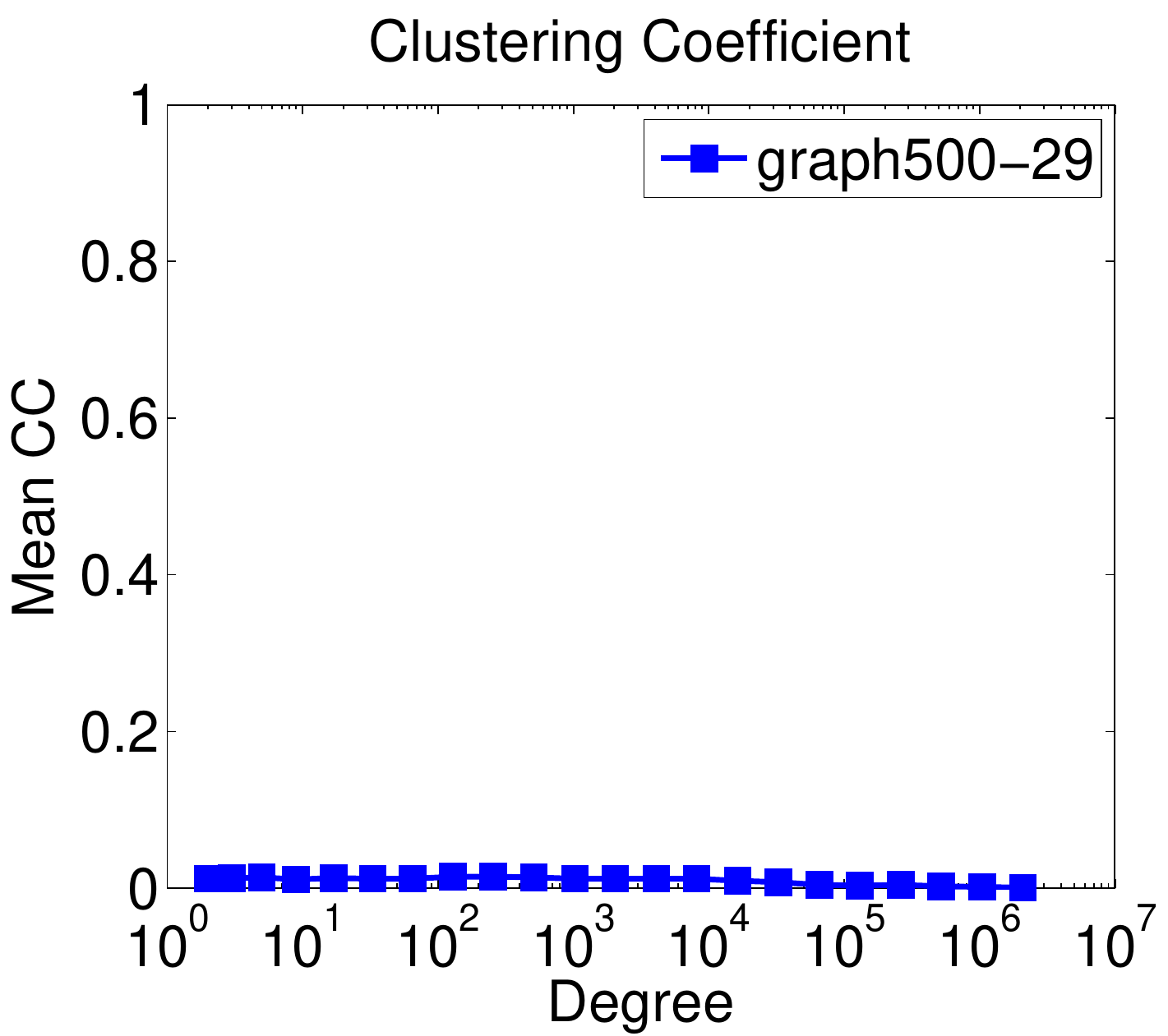}}
  \caption{Clustering coefficients by bin}
  \label{fig:cc}
\end{figure}

\subsection{Triangles}

We also present the number of triangles per bin in \Fig{tri}.  Here the x-coordinate is
the minimum degree in the bin, and the y-coordinate is the proportion
of triangles that have at least one vertex in that 
bin. Triangles may be counted more than once if they have different
vertices in different bins.

It is interesting to observe where triangles come from. Even though
low-degree nodes are the most plentiful, most of the triangles come
from higher-degree. We can roughly sort the graphs into three
categories.

There is only one graph where the triangles come from relatively
low-degree vertices: amazon-2008. Here it seems that most triangles
come from nodes with degrees between 4 and 80. 

There are several graphs where the triangles come from the mid-range
degrees: one social graph (ljournal-2008) and all the web graphs
(it-2004, sk-2005, uk-2006-06, uk-union). The ``double-spike''
behavior of sk-2005 is interesting.

\begin{figure}[tpb]
  \centering
  \subfloat{\includegraphics[width=.3\textwidth]{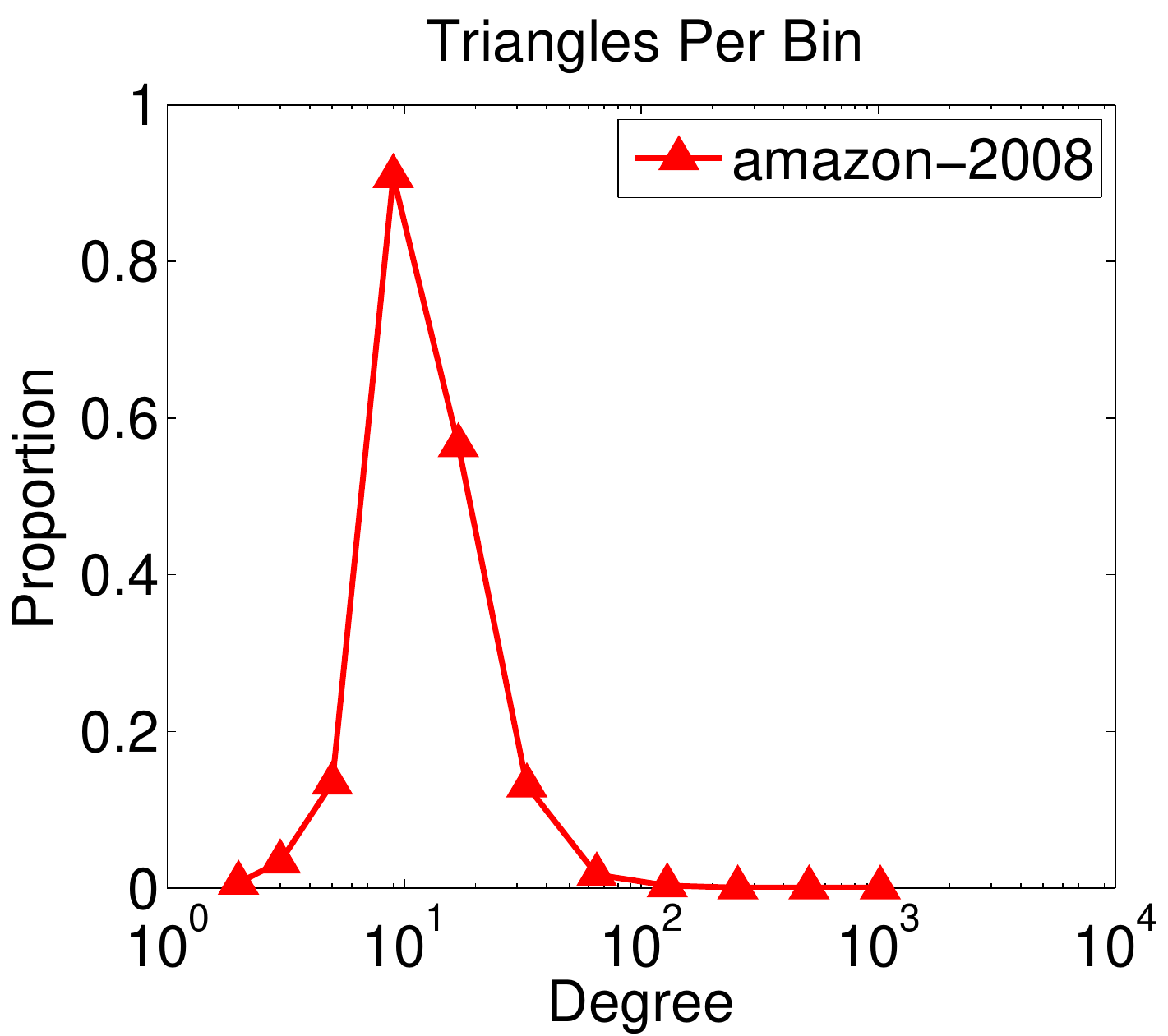}}
   \subfloat{\includegraphics[width=.3\textwidth]{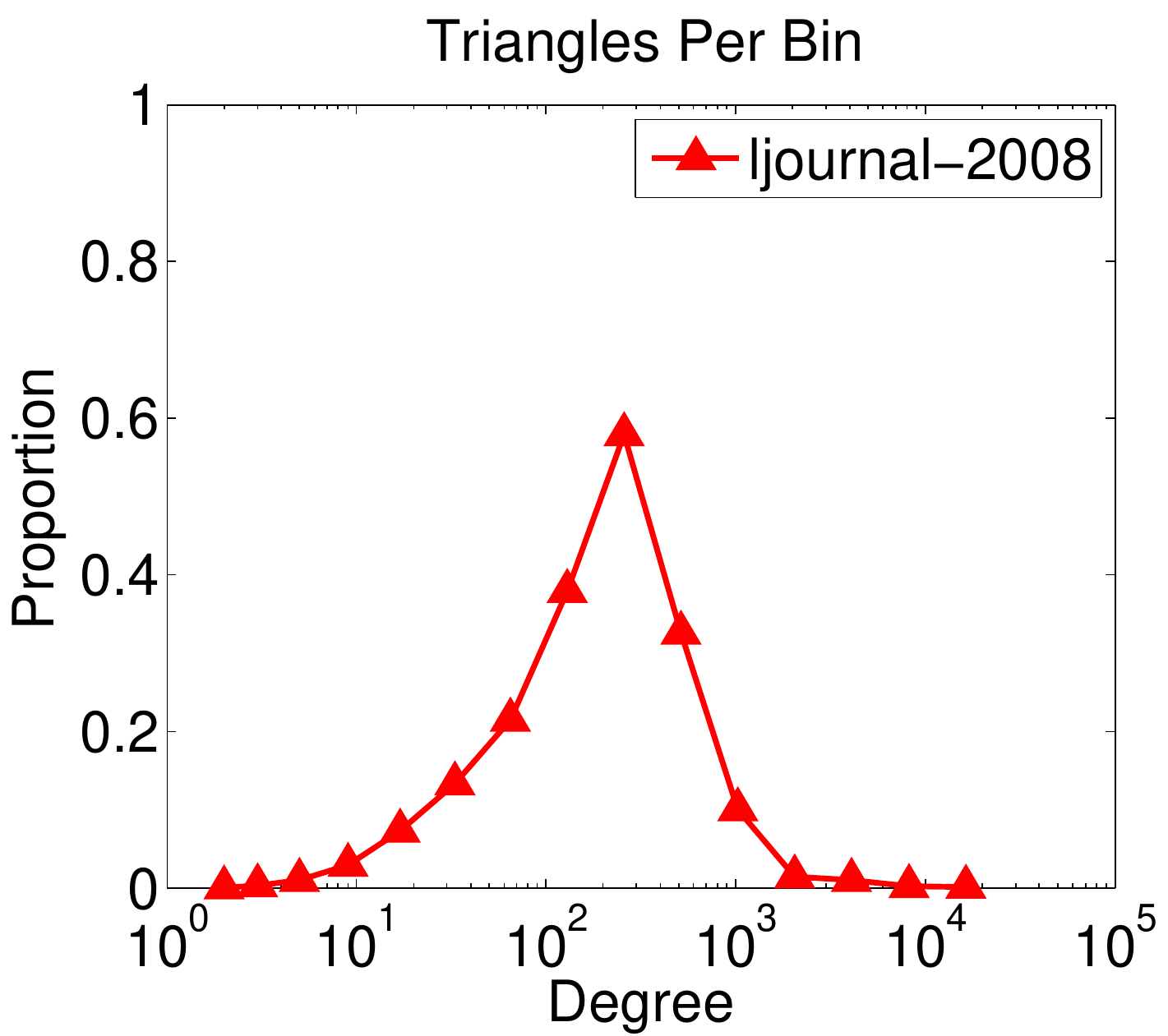}}
  \subfloat{\includegraphics[width=.3\textwidth]{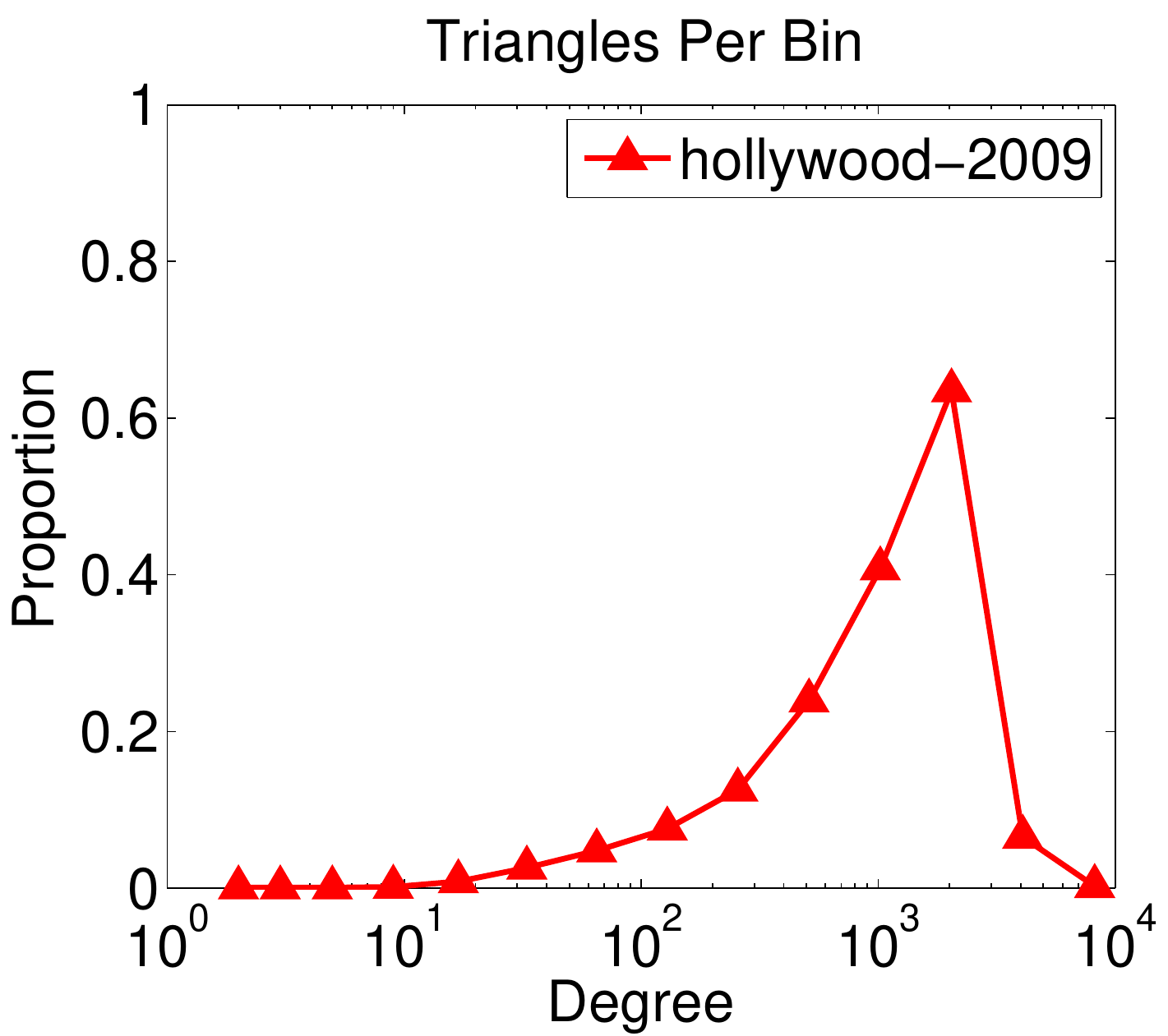}} \\
  \subfloat{\includegraphics[width=.3\textwidth]{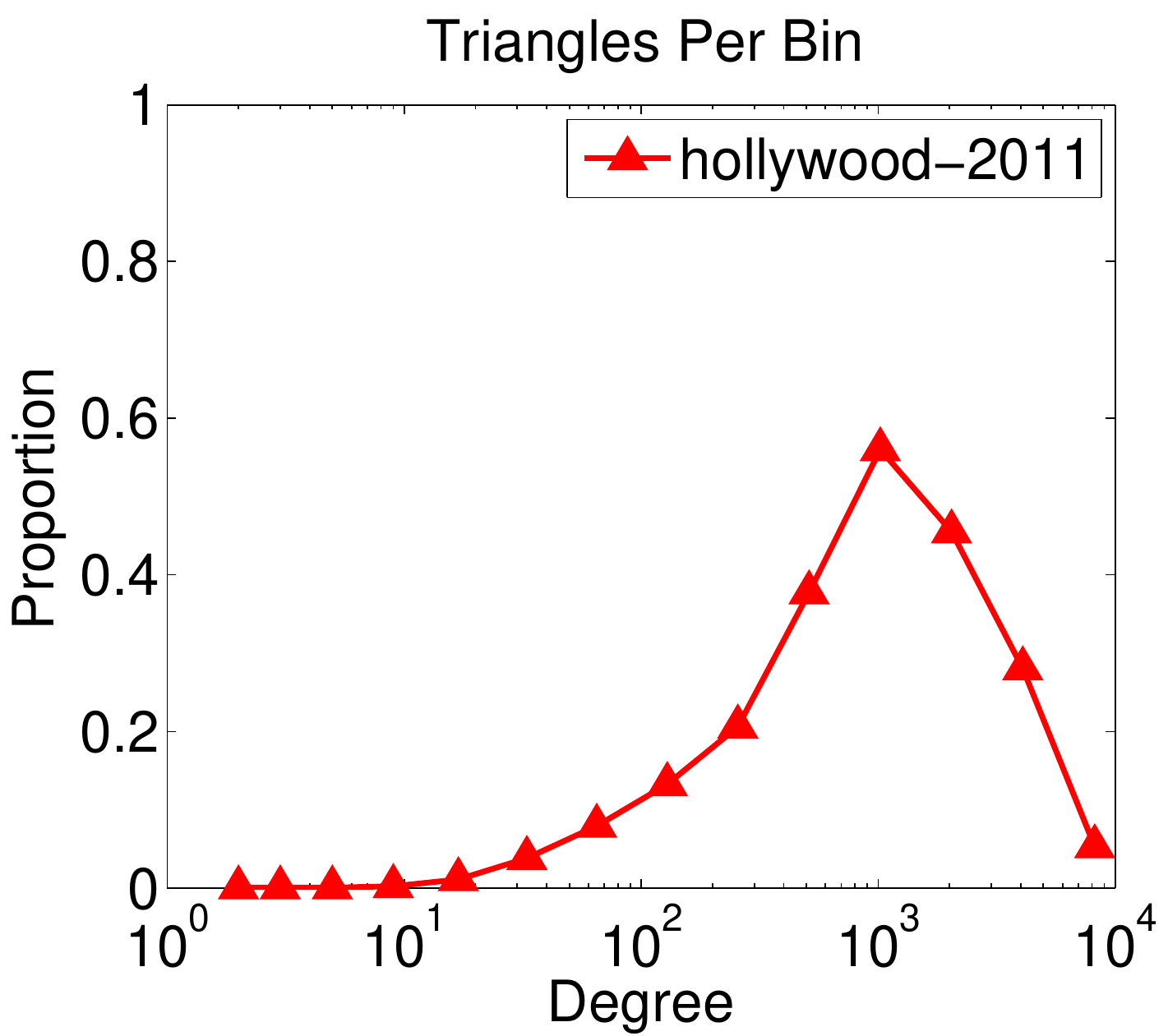}}
    \subfloat{\includegraphics[width=.3\textwidth]{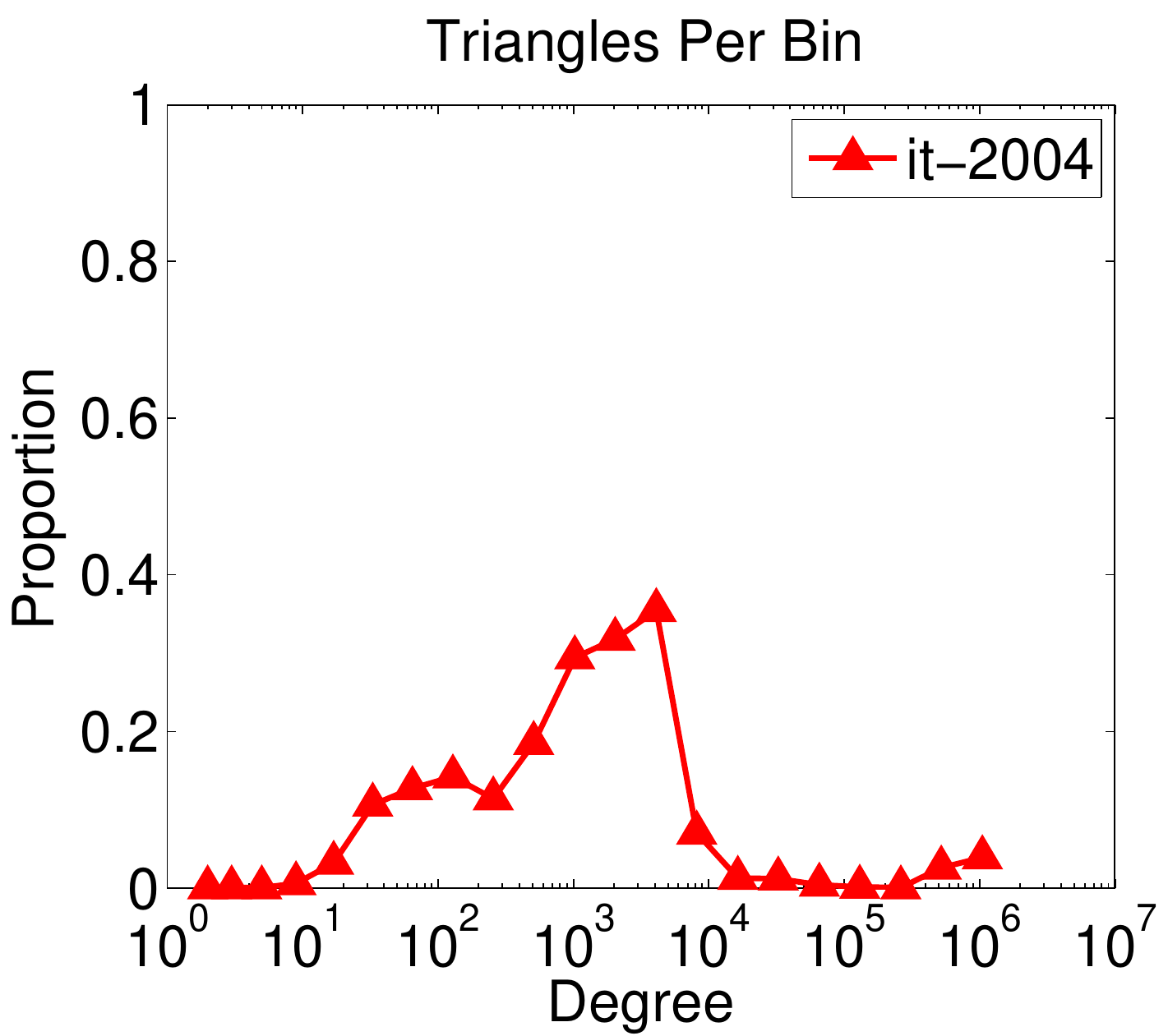}}
  \subfloat{\includegraphics[width=.3\textwidth]{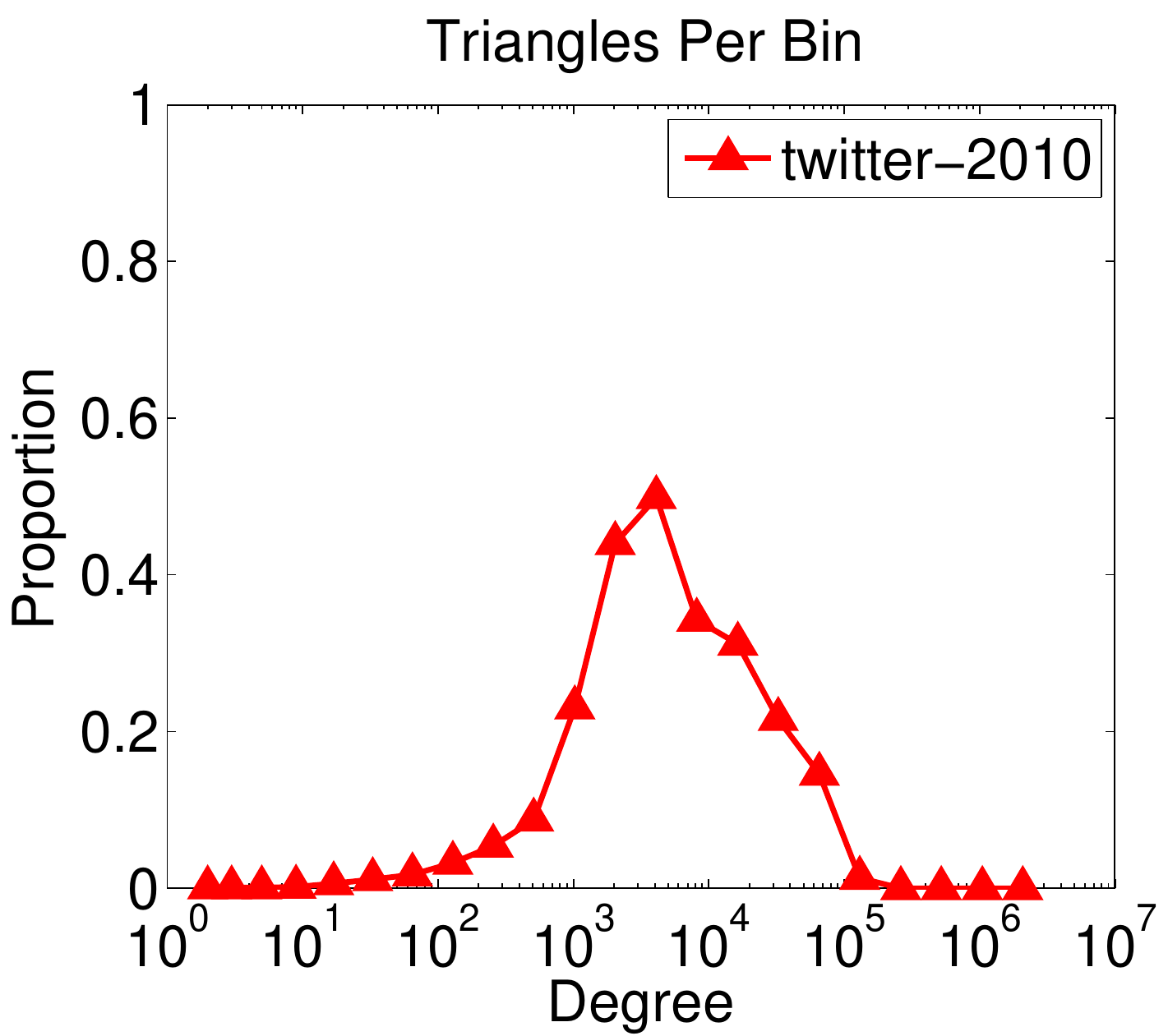}}\\
  \subfloat{\includegraphics[width=.3\textwidth]{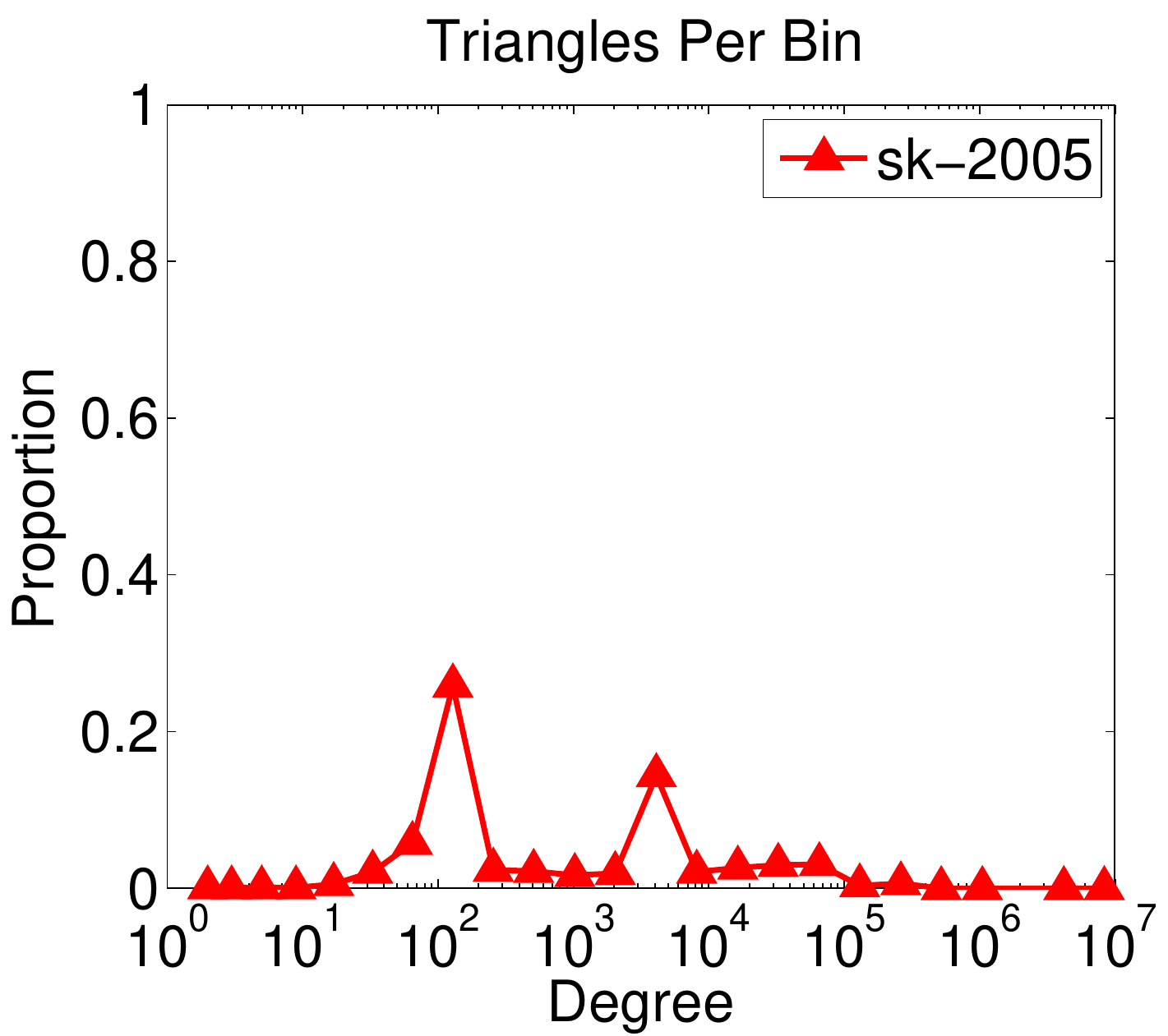}}
  \subfloat{\includegraphics[width=.3\textwidth]{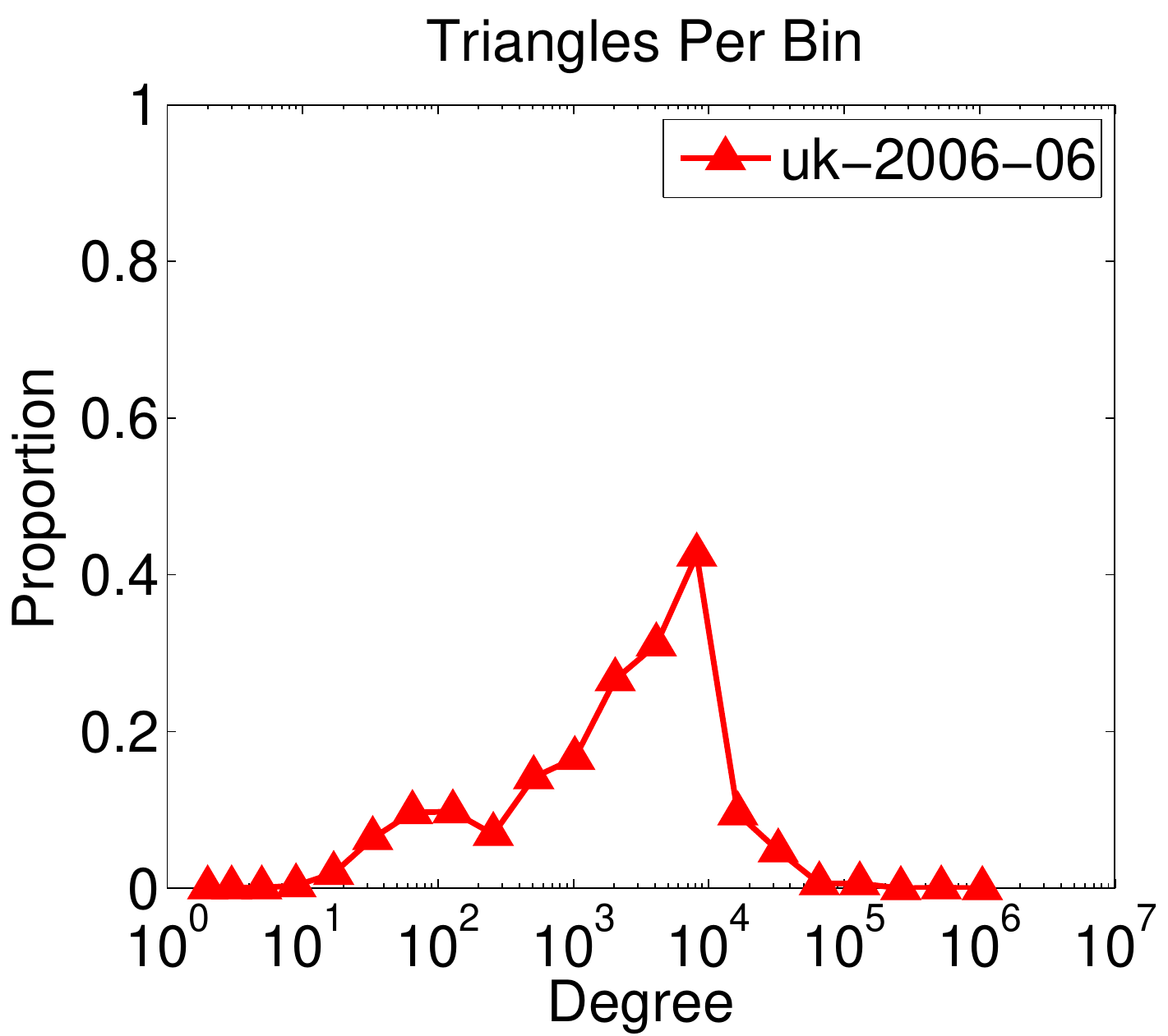}}
  \subfloat{\includegraphics[width=.3\textwidth]{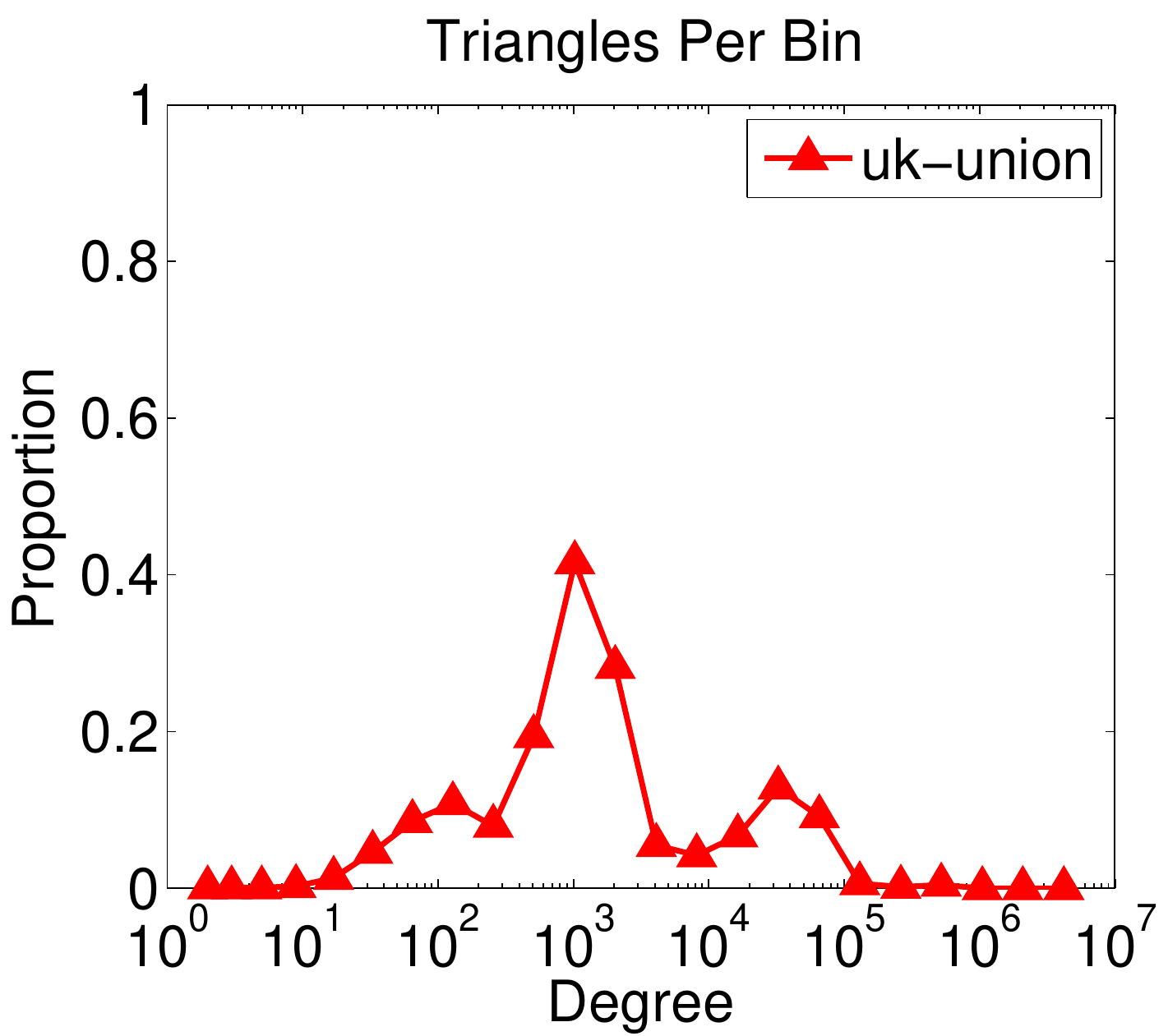}}\\
  \subfloat{\includegraphics[width=.3\textwidth]{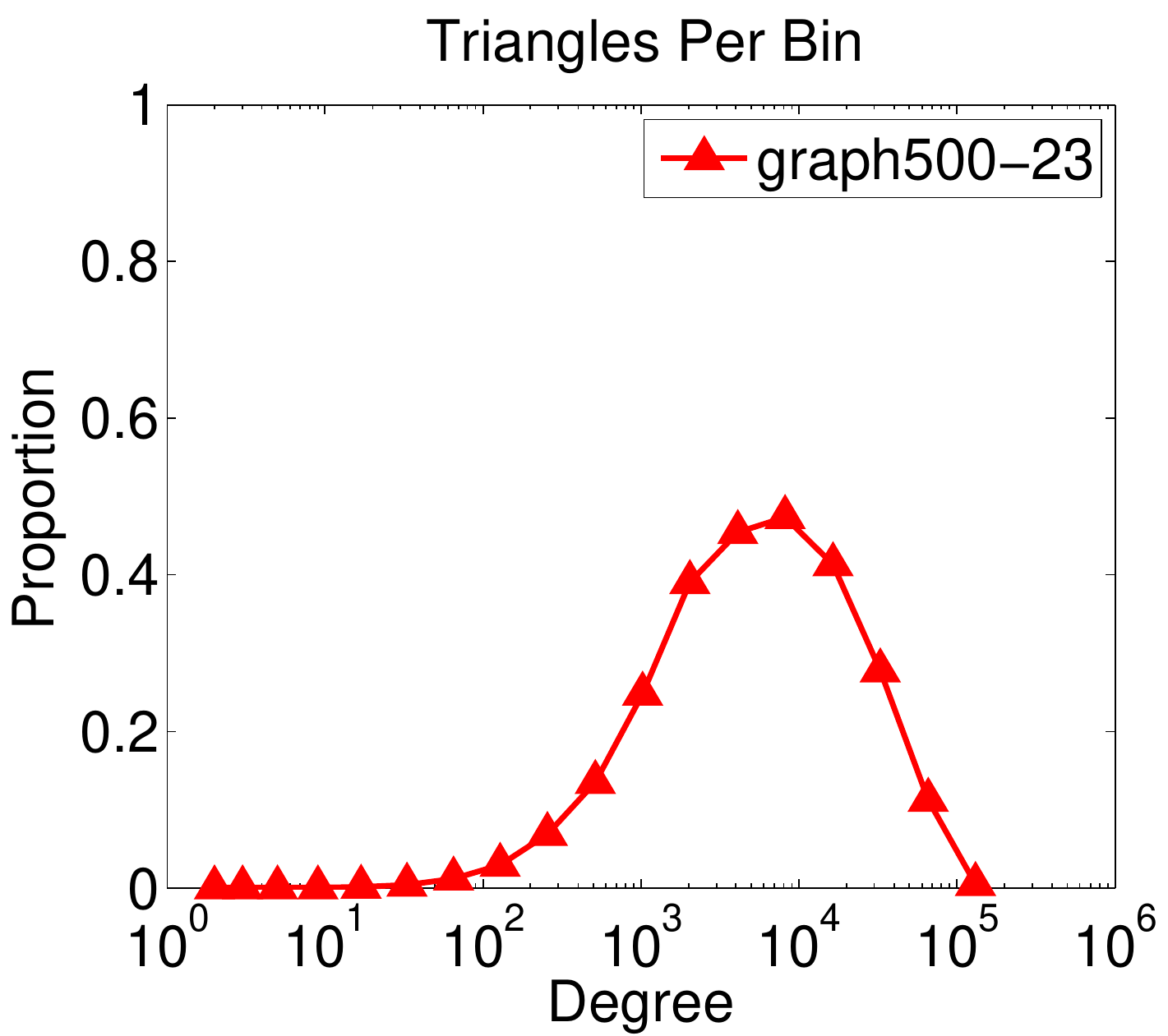}}
  \subfloat{\includegraphics[width=.3\textwidth]{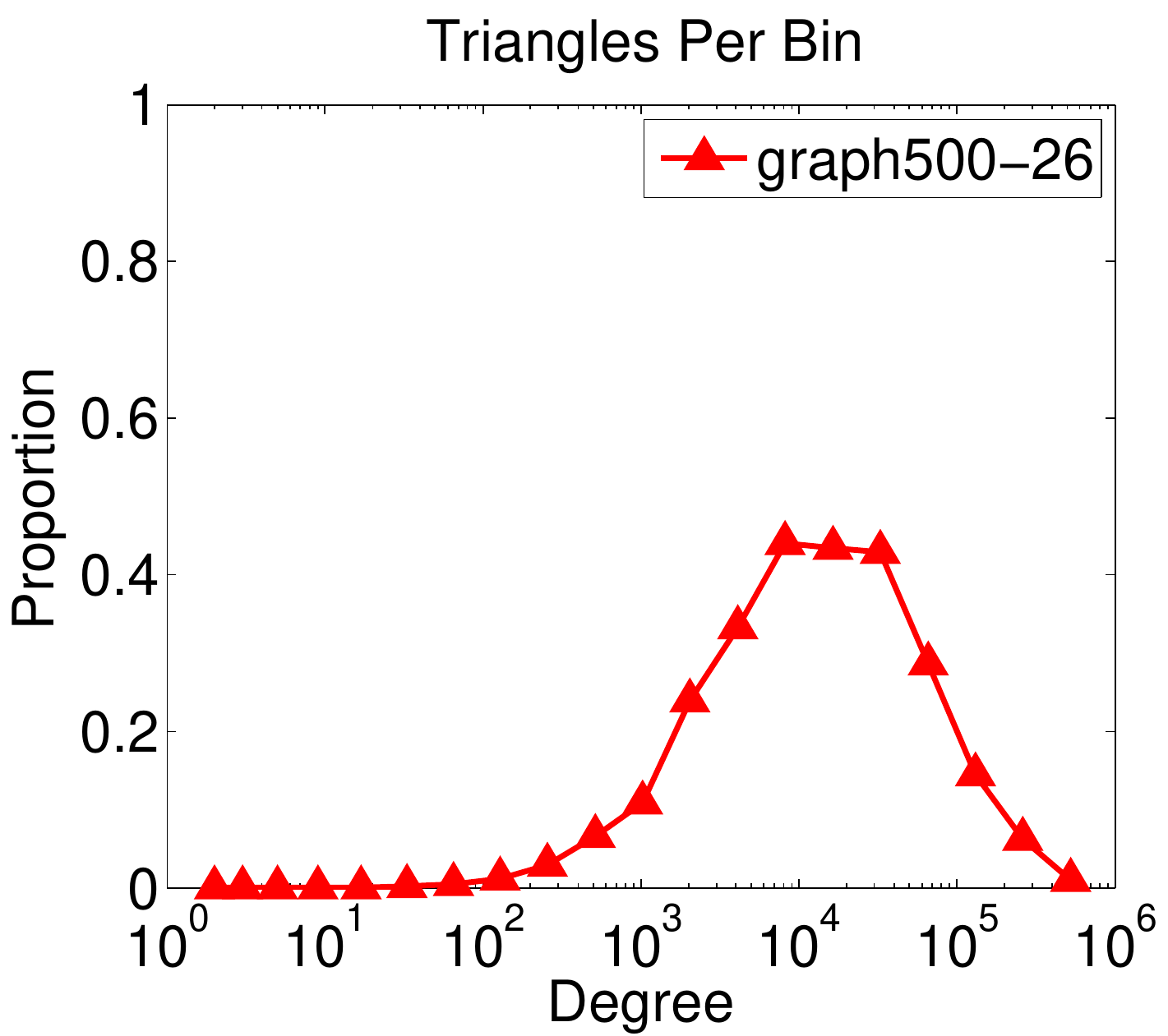}}
  \subfloat{\includegraphics[width=.3\textwidth]{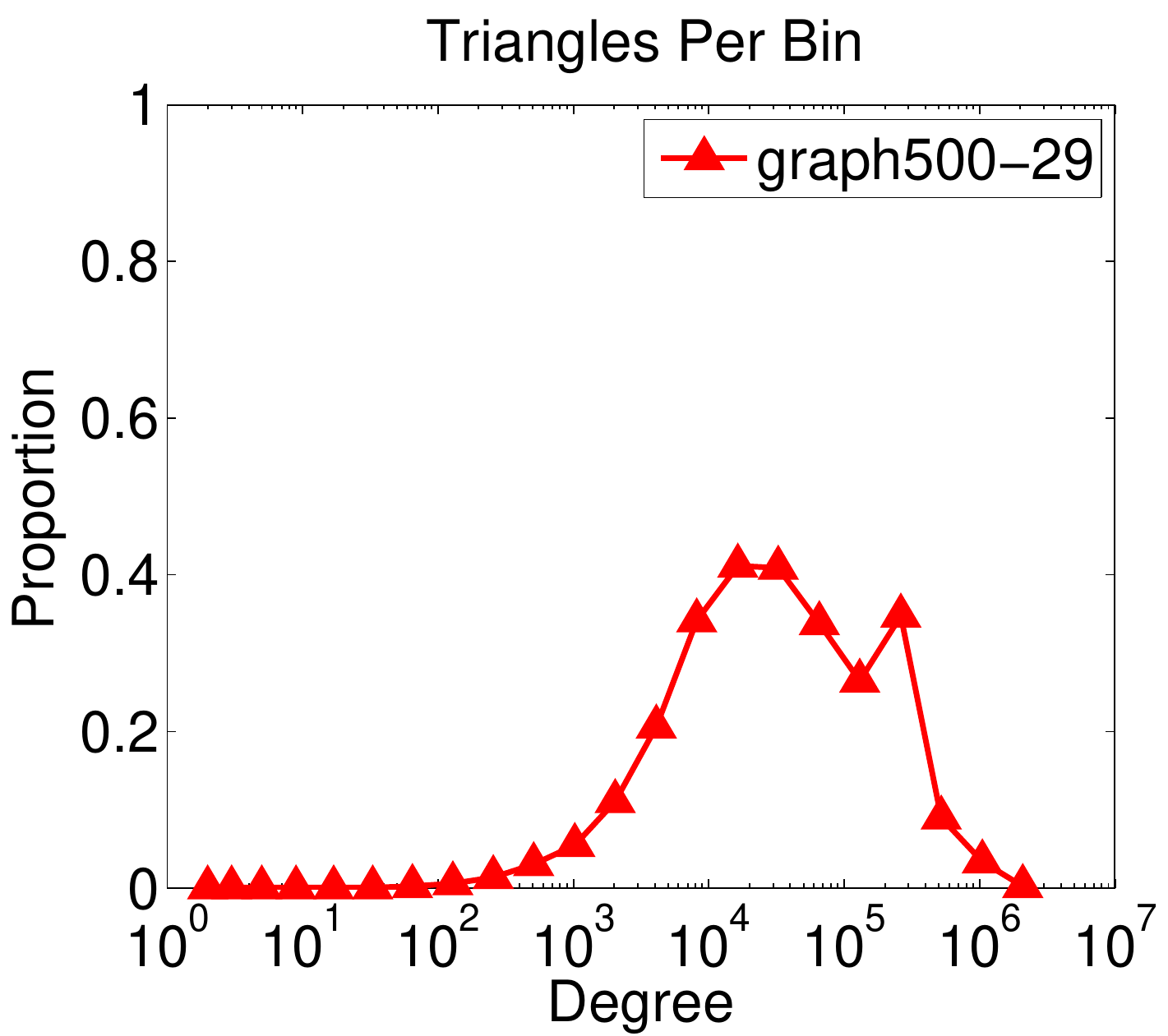}}
  \caption{Proportion of triangles in each bin}
  \label{fig:tri}
\end{figure}

Finally, there are a few graphs where the vast majority of triangles
involve the high-degree nodes. Both Hollywood graphs are of this type;
note that 60\% of the triangles involve one of the nodes in the bin
starting at degree 1025. The Graph500 graphs also have most of the
triangles coming from the highest degree nodes.

\subsection{Triangle Statistics}
\label{sec:exp-sample}
 
An interesting feature of our wedge sampling techniques is that, in
the case of a single bin, all the closed triangles are uniformly randomly
sampled as well.  Such a random sample can be used to analyze the
characteristics of the triangles in the graph, going further than
merely looking at their count.  Examples of such studies can be found
in~\cite{DuPiKo12}, where full enumeration of the
triangles was used.  To avoid the burden of full enumeration a uniform
sampling of the triangles can be used, as we showcase below.

For four example graphs, we ran our MapReduce code with a single bin ($\tau=1$ and
$\omega=10^7$) and
$k=5,000,000$ samples; we skipped phases 2b and 3a to avoid any data
overflow problems in the configuration parameters. Runtimes and the
number of triangles (expected to be roughly $k$ times the global
clustering coefficient) are reporting in \Tab{trisample}.

\begin{table}[th]\footnotesize
  \centering
  \caption{Number of triangles from 5,000,000 wedge samples.}
    \label{tab:trisample}
  \begin{tabular}{|l|r|r|} \hline
     \bf Graph Name & \bf Time(s) & \bf Triangles \\ \hline
       uk-union &    2618 &  33398 \\ \hline
 hollywood-2011 &     348 & 878719 \\ \hline
    graph500-26 &     845 &  46047 \\ \hline
    graph500-29 &    5487 &  25994 \\ \hline
  \end{tabular}
\end{table}

Using these sampled triangles, we can look at the degrees of the
vertices. Each triangle has a minimum, middle, and maximum degree. 
We analyze the \emph{degree assortativity} of the vertices of the triangles
by comparing the minimum and maximum degrees in
\Fig{tristats-assort}. 
Specifically, we assign each vertex to a degree bin, using \Eqn{binid}
with $\tau=2$ and $\omega=2$. 
We group all triangles with the same minimum degree bin
together. The box plot shows the statistics of the bin for the maximum
degree: the
central mark (red) is the median max-degree, while the edges of the
(blue) box are the 25th and 75th percentiles.  The whiskers extend to
the most extreme points considered not to be outliers, and the
outliers (red plus marks) are plotted individually.
 
Observe that the social network, hollywood-2011, shows an assortative relation between the
maximum and minimum for the hollywood graph, since the two quantities rise
gradually together.  For the uk-union web graph on the other hand, the
average maximum degree is essentially invariant to the minimum degree.  These
findings are consistent with the results in~\cite{DuPiKo12} about
networks with high global clustering coefficients having degree
assortative triangles, while this assortativity cannot be observed in
networks with low clustering coefficients.  Here, we were able to
observe the same trend on these massive graphs using sampling in a
much more efficient way, avoiding the enumeration burden.
We see that the Graph500 networks also have almost no assortativity between
the minimum and maximum degrees and therefore do not have the
characteristics of a social network.
    
    \begin{figure}[htbp]
      \centering
      \subfloat{\includegraphics[width=.45\textwidth]{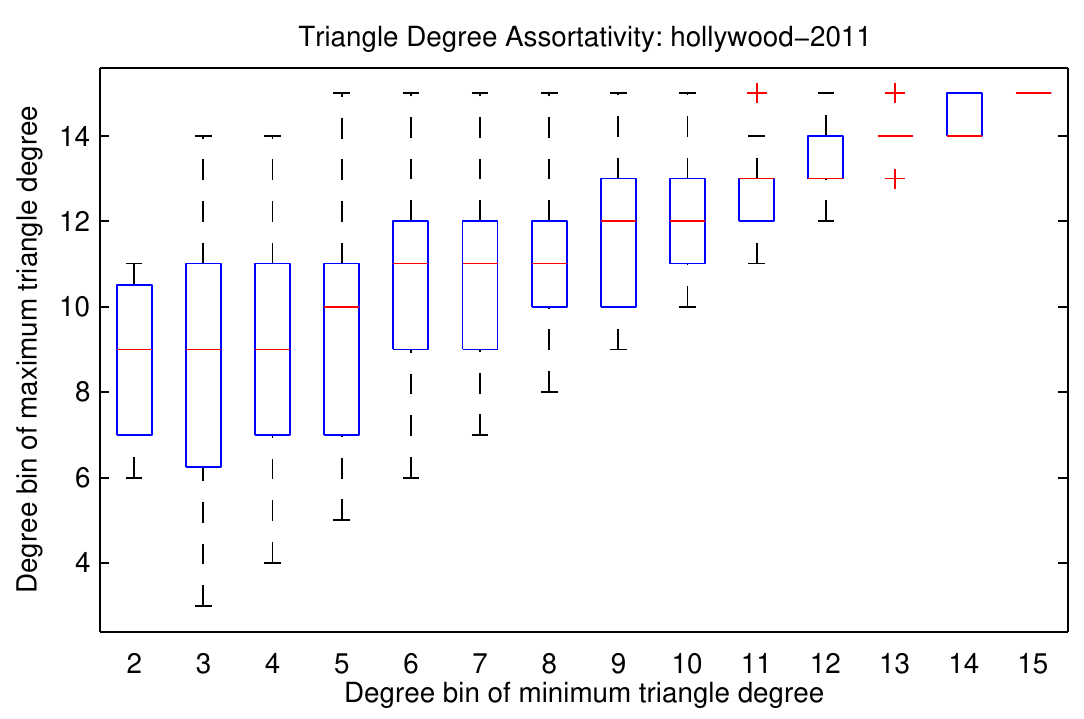}}
      \subfloat{\includegraphics[width=.45\textwidth]{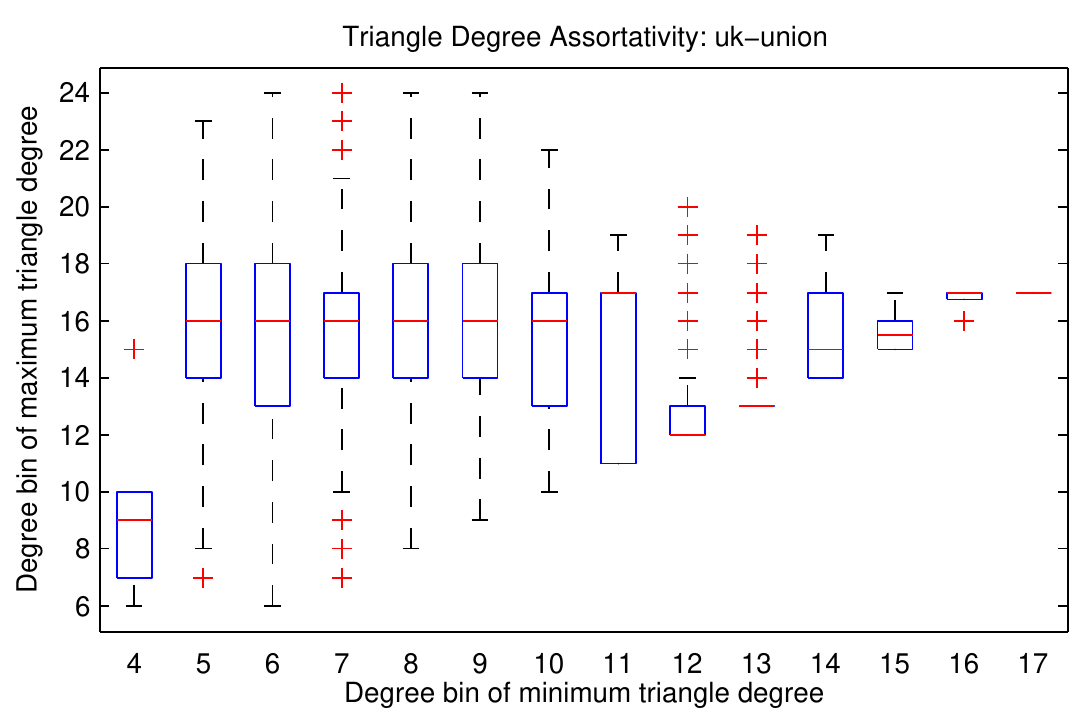}}\\
      \subfloat{\includegraphics[width=.45\textwidth]{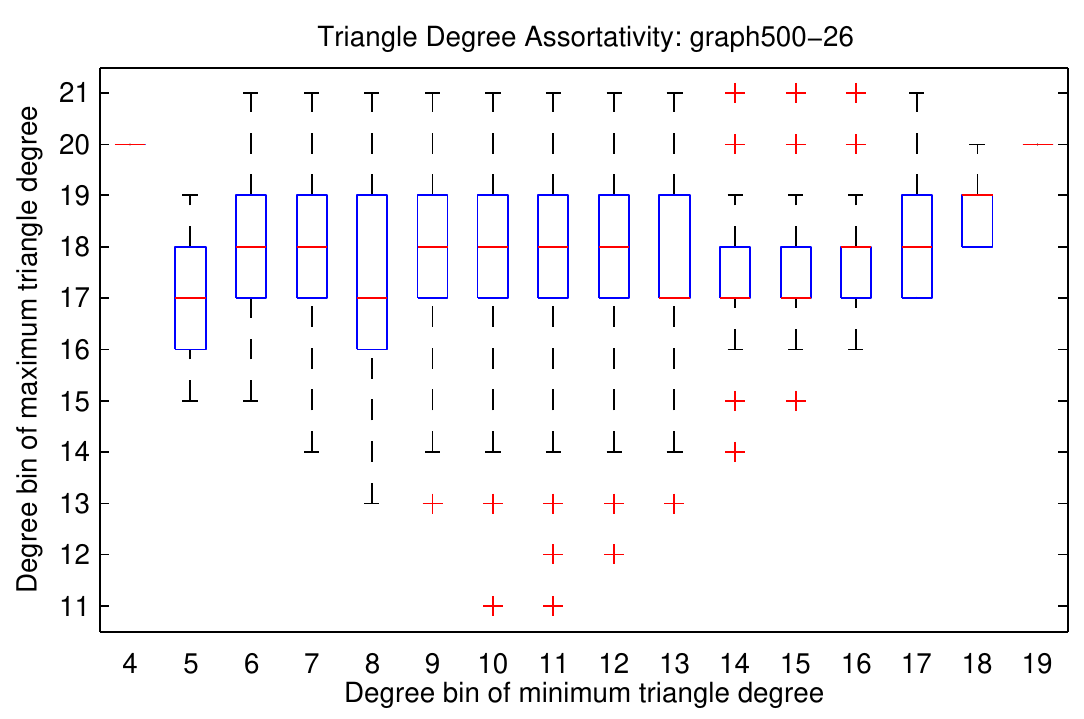}}
      \subfloat{\includegraphics[width=.45\textwidth]{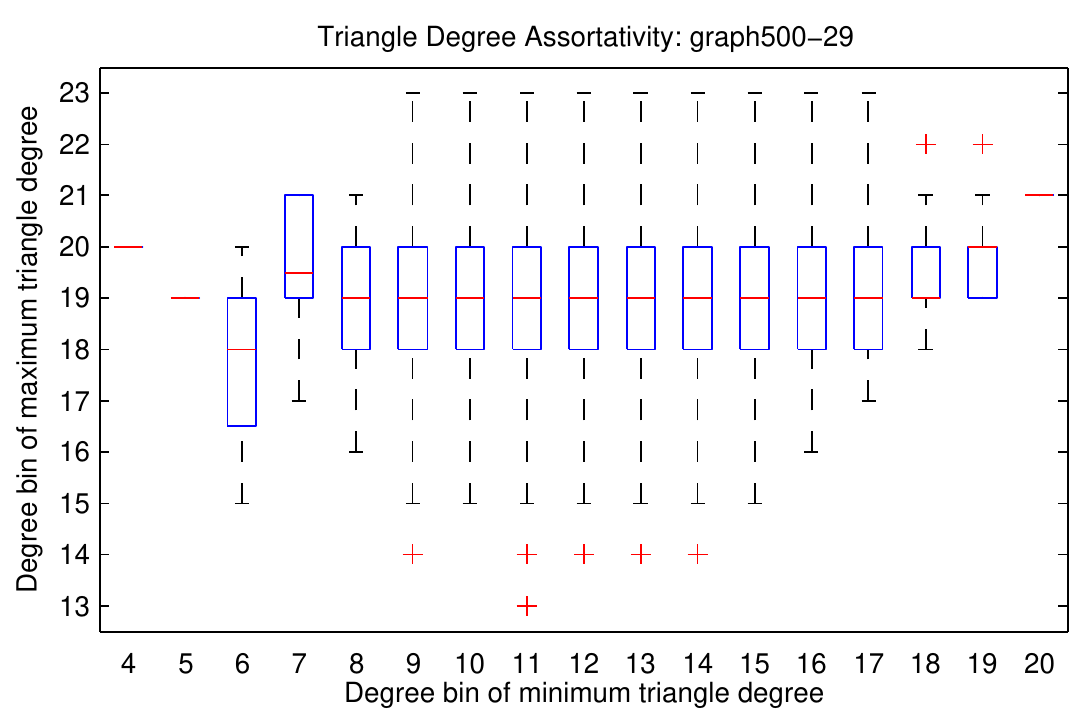}}
      \caption{Triangle Degree Assortativity}
      \label{fig:tristats-assort}
    \end{figure}

\section{Conclusions} 

We have shown that wedge-based sampling can be scaled to massive
graphs in the MapReduce framework. 
On a relatively small MapReduce cluster (32 nodes), 
we have analyzed graphs with up to 240M edges, 8.5B edges, 5.2T
wedges, and 447B triangles. 
Even the largest graph was analyzed in less than one hour, and most
took only a few minutes.
\Fig{Task1} shows a timing analysis of the MapReduce tasks \cite{JOBINFO}
for Phase 1a on the uk-union graph.  Mapper tasks run in waves of 128
parallel jobs, equal to the number of mapper slots available on the cluster.
Note that a larger cluster would be able to run
more Map tasks in parallel, decreasing the overall runtime.
To the best of our knowledge, these are
the largest triangle-based calculations performed to date. 

\begin{figure}[tbhp]
  \centering
  \includegraphics[width=.7\textwidth]{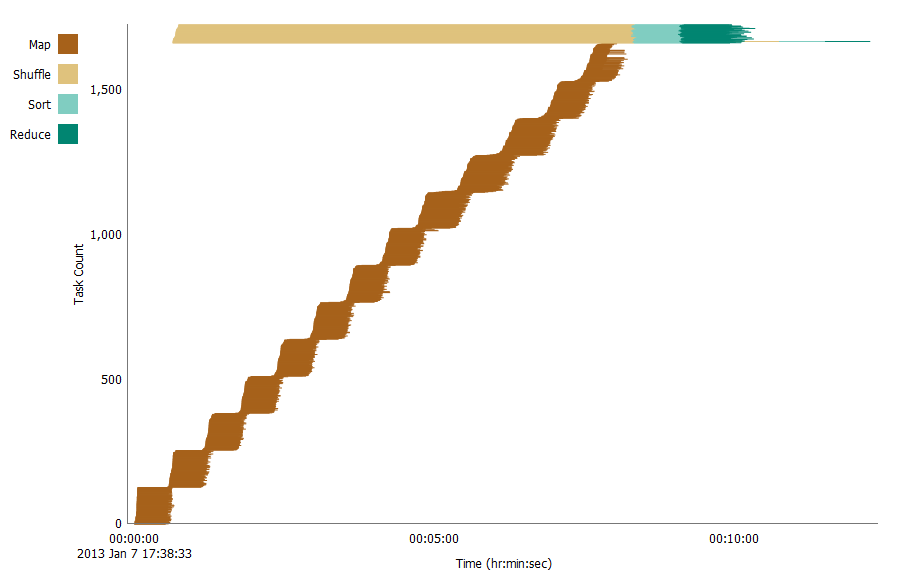}
  \caption{Task breakdown for Phase 1a on uk-union on 32 Hadoop nodes.}
  \label{fig:Task1}
\end{figure}

Unlike enumeration techniques that need to at least validate every
triangle and more often have cost proportional to the number of
wedges, our method is linear in the number of edges.
The most expensive component of the wedge-based sampling is finding the
degree of each vertex (Phase 1a); reducing the time for this is a topic for
future study. On our cluster, the time is approximately \Slope\@ seconds per
million edges, plus a fixed cost of \Offset\@ seconds for overhead. Because
we are using MapReduce, we never need to fit the entire graph into
memory --- we only need to be able to stream through all the edges.

Our MapReduce implementation requires a total of eight MapReduce jobs,
three of which do most of the work because they read the entire edge
list (Phases 1a, 2c, and 3b) and
two of which are optional (Phases 4a and 4b, which are labeling the
degrees of the sampled triangles). We have striven to minimize the
data being shuffled in each phase by using special data structures to
filter the edges.

Using our code, we are able to compute the degree distribution,
approximate the binned degree-wise clustering coefficient
and the number of triangles per bin. Additionally, we can analyze the
characteristics of the triangles (e.g. degree assortativity).
As part of our analysis, we have analyzed the graphs used in the 
Graph500 benchmark.  We are able to give a more detailed
understanding of the empirical properties of the generator and compare
it to real-world data; this is potentially helpful in determining if
performance on the benchmark data is indicative of
performance on real-world data.

\section*{Acknowledgments}
We gratefully acknowledge the peer reviewers for their constructive comments which have substantially improved the paper and brought several additional references to our attention.
We are thankful to Jon Berry for conducting the in-memory experiments described in \Sec{timings}.

\bibliographystyle{siammod}

%%\\n
%

\end{document}